\documentclass{iopart}
\usepackage[numbers]{natbib}

\usepackage{iopams}
\usepackage{euscript}
\usepackage{amssymb}

\expandafter\let\csname equation*\endcsname\relax
\expandafter\let\csname endequation*\endcsname\relax

\usepackage{amsmath}

\usepackage{amsthm}
\usepackage{dsfont}

\usepackage[perpage,symbol]{footmisc}

\usepackage{graphicx}
\usepackage{epstopdf}
\usepackage{dcolumn}
\usepackage{bm}
\usepackage{color}
\usepackage{geometry}
\usepackage[ruled]{algorithm2e}
\usepackage{algorithmic}

\newenvironment{protocol}[1][htb]
  {
   \begin{algorithm}[#1]%
  }{\end{algorithm}}
  
\newtheorem*{lemmarep}{\textsc{Lemma}}
\newtheorem*{theorep}{\textsc{Theorem}}
\newtheorem{theo}{\textsc{Theorem}}

\newtheorem{lemma}[theo]{\textsc{Lemma}}

\newtheorem{defi}[theo]{\textsc{Definition}}

\renewcommand{\theenumi}{(\arabic{enumi})}

\renewcommand{\theenumii}{(\alph{enumii})}

\usepackage{bbm}
\usepackage[colorlinks]{hyperref}
\def\EQ#1{\begin{align}#1\end{align}}
\def\ket#1{\left| #1 \right\rangle}
\def\bra#1{\left\langle #1 \right|}
\def\dm#1{\left|#1 \right\rangle \left\langle #1 \right|}

\newcommand{\ketbra}[3]{\left| #1 \right\rangle \left\langle #2 \right|_{#3}}

\newcommand{\px}{\sigma_x}
\newcommand{\py}{\sigma_y}
\newcommand{\pz}{\sigma_z}

\newcommand{\ptr}[2]{\text{tr}_{#1}\left[ #2 \right]}

\newcommand{\id}{id}
\hypersetup{
	bookmarksnumbered,
	pdfstartview={FitH},
	citecolor={darkgreen},
	linkcolor={darkblue},
	urlcolor={darkblue},
	pdfpagemode={UseOutlines}}

\definecolor{darkgreen}{RGB}{50,190,50}
\definecolor{darkblue}{RGB}{0,0,190}
\definecolor{darkred}{RGB}{238,0,0}
\usepackage{soul}

\newcommand{\E}{\mathcal{E}}
\renewcommand{\S}{\mathcal{S}}
\renewcommand{\H}{\mathcal{H}}

\newcommand{\F}{\mathcal{F}}

\newcommand{\N}{\mathbb{N}}
\newcommand{\C}{\mathbb{C}}

\newcommand{\T}{\mathcal{T}}

\usepackage[dvipsnames]{xcolor}

\newcommand{\bea}{\begin{eqnarray}}
\newcommand{\eea}{\end{eqnarray}}

\def\tr{\mathrm{tr}}

\makeatletter
\newtheorem*{rep@theorem}{\rep@title}
\newcommand{\newreptheorem}[2]{%
\newenvironment{rep#1}[1]{%
 \def\rep@title{#2 \ref{##1}}%
 \begin{rep@theorem}}%
 {\end{rep@theorem}}}
\makeatother

\begin{document}

\title{Entanglement generation secure against general attacks}
\author{Alexander Pirker$^{1}$, Vedran Dunjko$^{1,2}$, Wolfgang D\"ur$^{1}$ and Hans J. Briegel$^{1}$}

\address{${1}$
Institut f\"ur Theoretische Physik, Universit\"at Innsbruck, Technikerstr. 21a, 
A-6020 Innsbruck,
Austria \newline
${2}$ 
Max-Planck-Institute of Quantum Optics, Hans-Kopfermann-Strasse 1,
D-85748 Garching, Germany}

\eads{\href{mailto:alexander.pirker@student.uibk.ac.at}{alexander.pirker@student.uibk.ac.at}, \href{mailto:vedran.dunjko@mpq.mpg.de}{vedran.dunjko@mpq.mpg.de}, \href{mailto:wolfgang.duer@uibk.ac.at}{wolfgang.duer@uibk.ac.at} and \href{mailto:hans.briegel@uibk.ac.at}{hans.briegel@uibk.ac.at}}
%
\begin{abstract}
We present a security proof for establishing private entanglement by means of recurrence-type entanglement distillation protocols over noisy quantum channels. We consider protocols where the local devices are imperfect, and show that nonetheless a confidential quantum channel can be established, and used to e.g. perform distributed quantum computation in a secure manner. While our results are not fully device independent (which we argue to be unachievable in settings with quantum outputs), our proof holds for arbitrary channel noise and noisy local operations, and even in the case where the eavesdropper learns the noise. Our approach relies on non-trivial properties of distillation protocols which are used in conjunction with de-Finetti and post-selection-type techniques to reduce a general quantum attack in a non-asymptotic scenario to an i.i.d. setting. As a side result, we also provide entanglement distillation protocols for non-i.i.d. input states.
\end{abstract}
\maketitle

\section{Introduction}

Entanglement is a key resource in quantum information processing. Entanglement can be used to teleport quantum information \cite{BennettTele}, to implement remote quantum gates \cite{BennettRemote}, or for distributed quantum computation \cite{Cirac}. It allows one to perform tasks that are not possible by classical means, such as secret key expansion vital for secure classical communication. The latter is achieved through the famous and extensively studied quantum key distribution (QKD) protocols \cite{bib.lo, bib.got, bib.sho, bib.bai, bib.postselect, bib.renner.diss, bib.applyexpdef}. In these works, security was proven in a variety of ever more general scenarios, considering noisy channels, imperfect devices and device-independent (DI) settings, where even the local quantum devices are untrusted \cite{bib.acin,bib.lim,bib.umseh}.

In contrast, the perhaps equally important task of establishing private entanglement, and the closely related problem of establishing secure quantum channels, has not been resolved in equal generality. The latter has, historically, received significantly less attention \cite{Barnum}, until the very recent increase of interest \cite{Hayden,SecChannelBroadbent,SecChannelPortmann,SecChannelGarg} in security under ideal settings. The task of establishing private entanglement has been considered in the context of noisy channels and both perfect \cite{bib.dejmps} operations, and operations with local depolarizing noise \cite{bib.aschauer, AschauerPRA}. In these works, either initial states that are identical and independently distributed (i.i.d.), or asymptotic scenarios are assumed.

Here, we present a comprehensive treatment for the security of distillation protocols. To make our results broadly applicable, we generalize the security model (i.e. powers of the adversary) over standard settings for protocols with quantum outputs. Furthermore, we remove the need for asymptotic, or i.i.d. assumptions, allow for more general noise models, and formulate and prove security criteria which ensure composability -- i.e. the security of the protocols when they are used in arbitrary contexts, e.g. as sub-routines of larger protocols.

More specifically, we consider arbitrary attacks employed by an adversary (Eve, the distributer of noisy or corrupt Bell-pairs) and assume noisy communication channels and noisy local operations --  essentially arbitrary noise describing imperfect single- and two-qubit gates. We also extend adversarial powers beyond standard: the noisy apparatus may leak all the information about the noise processes which occurred in a run of the protocol to Eve.

Our scenario, by necessity, falls short from full DI, as security under such weakest assumptions is not attainable for protocols with a quantum output -- any device used in any protocol with which a client can interact classically, perhaps to test its performance, but which eventually outputs a quantum system, can always deviate from honest behavior when the final quantum output is eventually demanded (independent of how elaborate the testing may have been).
This raises the questions of how DI assumptions can be relaxed such that security becomes possible also for quantum output protocols, or how standard security models can be further extended.

DI assumptions can be understood as an extreme noisy scenario, where Eve has absolute control over the noise processes. Our model relaxes this:  Eve's control is not exact (deterministic), but rather probabilistic, however still perfectly heralded -- while Eve may fail in her interventions, she still learns the noise realized. In this sense, generalizing the types of noise the protocol is provably secure under in our model, corresponds to scenarios which are ever closer to DI. Naturally, other generalizations of DI settings which make sense for protocols with quantum outputs may be possible \footnote{E.g., we assume very primitive, but trusted, quantum devices, such as a device which can either forward an input quantum system, or measure it in one basis. Already such a simple device invalidates our no-go observation.}.

We proceed by first providing a security analysis for i.i.d. inputs, and then generalize to non-i.i.d. states. This is done by employing de-Finetti and post-selection symmetrization-based techniques. However, since we are interested in security in arbitrary contexts, we must go beyond standard scenarios considered in entanglement distillation works \cite{bib.dejmps, bib.aschauer, AschauerPRA} and explicitly consider the adversarial quantum systems (containing e.g. purifications of all quantum states) as well. Therefore the symmetrization-based techniques cannot be straightforwardly applied, but need to be adapted. We present and discuss the required additional steps of preprocessing, and provide entanglement distillation protocols that are not restricted to i.i.d. inputs, but are capable of dealing with general inputs. The latter is related to recent results in \cite{Brandao, Buscemi,Waeldchen}. \newline

\section{Structure of the paper}

The paper is organized as follows. In Sec. \ref{sec:model} we introduce the basic concepts, specify the overall setting and define the confidentiality of entanglement distillation protocols. Next, we summarize our main contribution in Sec. \ref{sec:contr}. In Sec. \ref{sec:conf} we show confidentiality of recurrence-type entanglement distillation protocols by proving confidentiality for i.i.d. inputs in Sec. \ref{sec:conf:iid} and we extend this results to arbitrary initial states in Sec. \ref{sec:conf:arb} and \ref{sec:conf:1}. Finally we prove confidentiality whenever the noise transcripts leak to Eve in Sec. \ref{sec:conf:2}. We summarize and discuss our results in Sec. \ref{sec:disc}.

\section{The model and security guarantees}\label{sec:model}

Entanglement distillation is modelled by considering three players, Alice and Bob, who wish to generate a shared Bell pair, and Eve, who provides the initial pairs. Thus, Eve is connected to Alice and to Bob via a (generally noisy) quantum channel which may be completely under her control.
Alice and Bob are connected by a classical authenticated, but not confidential, channel.
In entanglement distillation protocols Alice and Bob apply local, in general noisy, quantum operations to their pairs.
To model this noise, we extend the approach of \cite{bib.aschauer}, where a noise register, referred to as the ``lab demon'' (L) register $L$ is used to store classical information about the local noise history, is appended to Alice and Bob's pairs. In this work, the L register is a quantum register, attached to Alice and Bob. We represent the noisy maps of the entanglement distillation process as unitaries acting on an enlarged Hilbert space. L thereby coherently applies Pauli operators onto the registers of Alice and Bob. Due to the symmetry of Bell states $\ket{B_{00}} = 1/\sqrt{2}(\ket{00}+\ket{11}),$ it suffices to consider the case when the noise is applied on Alice's register only.  To model the setting where Eve acquires information about the noise transcript during the execution of the protocol, we assume that L informs Eve which noise operator was applied at each step. The setting is illustrated in Fig. \ref{fig:model}. In the remainder of this paper we elaborate further on the full quantum treatment of L and Eve in terms of purifications, going beyond the setting of \cite{bib.aschauer}.

\begin{figure}[h!]
\begin{center}
\scalebox{1}{
\includegraphics{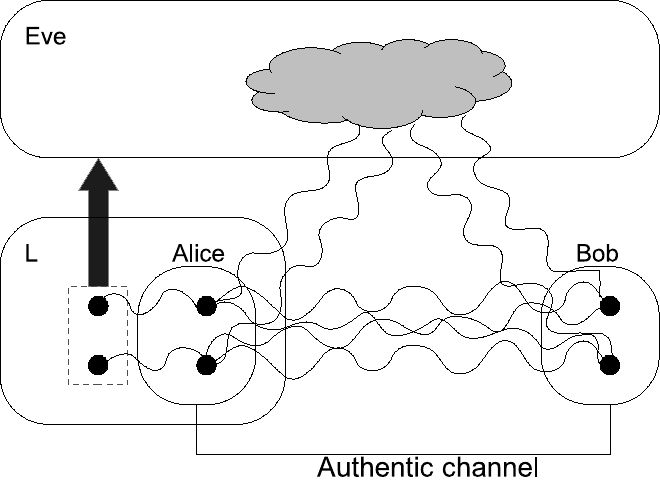}
}
\caption[h!]{\label{fig:model} Illustration of the overall setting: Eve provides the initial pairs to Alice and Bob, who run the entanglement distillation protocol. The noisy apparatus may leak the specification of the realized noise map to Eve after every step of the protocol.}
\end{center}
\end{figure}

The proposed overall protocol under i.i.d. assumption involves several steps. First, Eve distributes $n$ pairs (the \emph{initial states}), to Alice and Bob who apply local ``twirl'' operations (random, correlated local operations). Next, Alice and Bob sacrifice some $m\approx\sqrt{n}$ pairs to check whether the fidelity, given with $F(\rho, \sigma)=\text{tr}\sqrt{\rho^{1/2}\sigma \rho^{1/2}}$ for density operators $\rho$ and $\sigma$, of the pairs is sufficient for entanglement distillation, via local $\sigma_x$ and $\sigma_z$ measurements. If the fidelity $F$ relative to $\ket{B_{00}}$ is insufficient, they abort. Otherwise they proceed with a recurrence-type entanglement distillation to produce a high fidelity Bell-pair from the remaining initial states, which may also be aborted. Finally, Alice and Bob output their final state. For i.i.d. inputs, the twirl ensures that local $\sigma_z$ and  $\sigma_x$ correlation measurements can be used to estimate the fidelity of individual pairs. This estimate is crucial for ensuring entanglement distillation via recurrence-type entanglement distillation protocols. Later, we will generalize to non-i.i.d. settings by prepending the protocol with symmetrization (permuting of the pairs) and tracing-out steps.

To formalize the security requirements, we define the ideal map $\F^{\alpha,l}$, mapping the initial states of Alice and Bob to a single Bell-pair, where $\alpha$ (abstractly) characterizes the noise levels in the channels connecting Eve to Alice and Bob, and also the noise of the local devices, and $l$ indicates that the noise transcripts leak to Eve. The ideal map can intuitively be thought of as a map which simulates a real protocol as follows. In the case of an abort, it replaces the final state with a fixed state $\sigma^{\perp}_{ABE}$. In the non-aborting case, however, it replaces the actual output with a special state $\sigma^{\alpha,\mathcal{P},l}_{ABE},$ which corresponds to the output of a real protocol where the noise transcripts leak to Eve, utilizing distillation protocol $\mathcal{P}$, that was successfully run with asymptotically many high-fidelity i.i.d. initial pairs. This is the best the noisy entanglement distillation protocol $\mathcal{P}$ could ever do. As we show later, $\sigma^{\alpha,\mathcal{P},l}_{ABE}$ is a well-defined state for the entanglement distillation protocols and noise parameters considered here. That is, it depends on the local noise parameters \emph{only}, and not the initial states.
Formally, we have for a given real map (that is, the map realized by the execution of a real protocol)
\begin{align}
(\E^{\alpha,l} \otimes \id_E) \left(\dm{\psi}_{ABE} \right) = p_{\rho} \sigma_{ABE} \otimes \dm{ok}_{f} + (1-p_{\rho}) \sigma^{\perp}_{ABE} \otimes \dm{fail}_{f} \label{eq:realmap}
\end{align}
a corresponding ideal map 
\begin{align}
(\F^{\alpha,l} \otimes \id_E) \left(\dm{\psi}_{ABE} \right) = p_{\rho} \sigma^{\alpha,\mathcal{P},l}_{ABE} \otimes \dm{ok}_{f} + (1-p_{\rho}) \sigma^{\perp}_{ABE} \otimes \dm{fail}_{f}  \label{ideal}
\end{align}
where $\ket{\psi}_{ABE}$ is a purification of the initial $n$-partite ensemble $\rho^{(n)}_{AB}$ provided by Eve, $p_{\rho}$ is the success probability depending on the initial state $\rho^{(n)}_{AB}$, and $\sigma^{\perp}_{ABE}$ is a fixed state output if the protocol is aborted. Observe that the corresponding success probabilities $p_{\rho}$, per definition, are identical for the real and ideal maps $\E^{\alpha,l}$ and $\F^{\alpha,l}$ in (\ref{eq:realmap}) and (\ref{ideal}) respectively. The two-level flag system $f$ distinguishes the accepting and aborting branches. The state $\sigma^{\alpha,\mathcal{P},l}_{ABE}$ is the asymptotic state of the entanglement distillation protocol $\mathcal{P}$ and is of the form
\begin{align}
\sigma^{\alpha, \mathcal{P},l}_{ABE} = \left(\sum\limits^1_{i,j = 0} \omega_{ij}(\alpha,\mathcal{P}) \dm{B_{ij}}_{AB} \otimes \dm{\eta_{ij}}_{E} \right)  \otimes \sigma_{E}
\label{eq:asympstate}
\end{align}
where $\ket{\eta_{ij}}$ are the leaked noise transcripts of Eve, $\ket{B_{ij}} = (\id \otimes \px^j \pz^i) \ket{B_{00}}$ the Bell-basis states, and $\omega_{ij}(\alpha,\mathcal{P})$ are probabilities which depend on the noise level of the local devices and the entanglement distillation protocol $\mathcal{P}$. For instance, if the local devices are perfect, then $\omega_{ij} = 1$ if and only if  $i=j=0,$ hence $AB$ contains a perfect Bell-pair. Finally, the states $\ket{\eta_{ij}}$ specify the sequences of noise operations, and are orthogonal for different $i,j$. If the noise transcripts are not leaked to Eve, we denote the ideal protocol by $\F^{\alpha}$. In that case, $\ket{\eta_{ij}}$ in (\ref{eq:asympstate}) is not accessible to Eve, hence we replace $\sigma^{\alpha, \mathcal{P},l}_{ABE}$ by $\sigma^{\alpha, \mathcal{P}}_{ABE} = \left(\sum_{i,j} \omega_{ij}(\alpha,\mathcal{P}) \dm{B_{ij}}_{AB} \right) \otimes \sigma_{E}$ in (\ref{ideal}). Observe that the ideal map $\F^{\alpha,l}$, which mathematically defines the type of process we wish to realize, is a global operation beyond LOCC (local operations and classical communication) which can be decomposed by concatenating the real protocol $\E^{\alpha,l}$ and a replacement map $\mathcal{S}$ (which replaces the final state only if the real protocol succeeds according to the system $f$ in (\ref{ideal})), i.e. $\F^{\alpha,l} = \mathcal{S} \circ \E^{\alpha,l}$.  \newline

An entanglement distillation protocol (together with the noise maps), given as a CPTP map $\mathcal{E}^ {\alpha,(l)}$, is confidential if it is close to the ideal map:
\begin{defi}\label{def.confidentiality}
The protocol $\E^{\alpha,(l)}$ is $\varepsilon$-confidential, if
\begin{align}\label{eqn.confidentiality}
\| (\E^{\alpha,(l)} & \otimes id_{E} -\F^{\alpha,(l)} \otimes id_{E}) (\ketbra{\psi}{\psi}{ABE}) \|_{1} \leq \varepsilon
\end{align}
holds for all initial states $\ket{\psi}_{ABE},$ where $\| \rho \|_{1} = \tr{\sqrt{\rho \rho^\dagger }}$ is the operator 1-norm for a density operator $\rho$. 
\end{defi}
The system $E$ above may contain any purification of the initial states Eve provided. \newline
In this work, we use the term security in a generic sense, and the precise meaning depends on the context. For instance, in QKD applications, security means that Alice and Bob establish a perfectly random and secret key which the adversary has negligible information about \cite{Gottesman,Lo,bib.sho,bib.got,bib.renner.diss,Koenig}. In recent times, composable security definitions have become commonplace, in which, roughly speaking, security is defined via an ideal process, and security level via the amount by which the process realized by the protocol deviates from the ideal process. In the context of QKD, this distance reduces to the distance on the generated final states of the ideal vs. realized protocol. The ideal protocol outputs a completely mixed state on Alice and Bobs system which is in tensor product with Eve. More formally, see also \cite{bib.renner.diss}, a QKD protocol $\mathcal{Q}$ is said to be $\varepsilon-$secure for initial state $\rho_{ABE}$ if  
\begin{align}
\| \sigma_{S_A S_B C E} - \sigma_{SS} \otimes \sigma_{CE} \|_1 \leq \varepsilon \label{eq:def:qkd:renner}
\end{align}
holds where $\sigma_{S_A S_B C E} = (\mathcal{Q} \otimes \id_E)(\rho_{ABE})$, $S_A$ and $S_B$ denote the output systems of Alice and Bob (corresponding the generated key), $C$ denotes the classical communication and $\sigma_{SS} = 1/|S| \sum_{s \in S} \dm{s} \otimes \dm{s}$ for orthogonal states $s$. The state $\sigma_{SS} \otimes \sigma_{CE}$ corresponds to the output of the ideal protocol. \newline
The confidentiality criterion which we introduce here follows the distance-on-maps approach introduced in the context of QKD like in e.g.  \cite{bib.postselect}. Observe that such an approach is especially tailored to compose different protocols, as the confidentiality definition concerns the distance of the real process with respect to an ideal process. Therefore the real and ideal maps $\E^{\alpha,(l)}$ and $\F^{\alpha,(l)}$ respectively are motivated by abstracting the protocol in terms of processes. It is straightforward to abstract and define the ideal map in terms of input and output relations, reflecting an ideal entanglement distillation process. As we discuss above, the ideal protocol has an $ok-$ and $fail-$branch. The $fail-$branch corresponds to the case whenever Alice and Bob abort the procedure, outputting the state $\sigma^\perp_{ABE}$. However, if the procedure succeeds then we might think of the ideal map as running the entanglement distillation protocol for infinitely many initial states, ending up in the fixed state $\sigma^{\alpha, \mathcal{P},l}_{ABE}$ of the entanglement distillation protocol $\mathcal{P}$ for noise level $\alpha$. We observe two important facts regarding that particular state: first, its the best the entanglement distillation protocol $\mathcal{P}$ can do in the presence of noise of level $\alpha$, and second, as Eve is disentangled from Alice and Bob, this state is useful for applications like quantum teleportation. Hence we refer to this state also as a private state, or equivalently, Alice and Bob share private entanglement. In contrast to (\ref{eq:def:qkd:renner}), the target state $\sigma^{\alpha, \mathcal{P},l}_{ABE}$ in the $ok-$branch is only in tensor product with respect to Eve if the noise transcripts do not leak to the adversary. In that case a secure quantum channel is feasible in terms of quantum teleportation. Otherwise, that is if the noise transcripts $\ket{\eta_{ij}}$ leak to Eve, she is in a separable state with respect to Alice and Bob, but still enabling for confidential applications. By confidential we mean here that when the final state is used for quantum  teleportation no information about the teleported state is leaked, but the final state does not guarantee that Eve cannot change the teleported state. This observation motivates the term confidentiality rather than security. \newline
The classical communication is not correlated to the output of the real protocol, thus it can be ignored, see \ref{app:sec:eppiid} for details. The robustness of the protocol \footnote{The robustness is quantified by the abort probability in the all-honest, but noisy setting.} is considered in \ref{app:sec:robust}, which enables us to assume for the subsequent analysis that all basic distillation steps succeed. 

\section{Main contribution}\label{sec:contr}

We summarize the main findings of our paper as follows: recurrence-type entanglement distillation protocols prepended by a symmetrization and a system discarding step enable confidentiality, provided that the noise transcripts do not leak to the adversary for all noise levels $\alpha$ for which distillation would be possible in the i.i.d. case. We also show that this alone implies that the final state in the accepting branch, is close to a tensor product state -- Eve is factored out. The  results regarding the BBPSSW protocol \cite{Bennett} are analytic whereas for the DEJMPS protocol \cite{bib.dejmps}  the results rely on strong numerical evidence. For low noise rates, we achieve better results via the post-selection-based reduction. In that case, no system discarding step is necessary. Finally we find that if an entanglement distillation protocol is confidential when the noise transcripts do not leak, then it also confidential if they do leak to the adversary. In particular, even in the case that Eve picks up information about all the realized noise processes during the protocol, the final output system still enables confidential quantum applications like e.g. quantum teleportation. \newline
The paper proceeds as follows. We establish necessary conditions to guarantee confidentiality for recurrence-type entanglement distillation protocols restricted to i.i.d. inputs whenever the noise transcripts are not leaked to Eve. Then, we generalize this to arbitrary initial states via the de-Finetti theorem \cite{bib.oneandahalf}. Next, we use them to prove the confidentiality criterion (\ref{eqn.confidentiality}) for entanglement distillation protocols where the noise transcripts are not leaked. Finally, this will be used to derive the confidentiality bound whenever the noise transcripts are leaked. \newline

\section{Confidentiality of entanglement distillation protocols}\label{sec:conf}

\subsection{Entanglement distillation for i.i.d inputs}\label{sec:conf:iid}
The basic step of a recurrence-type entanglement distillation protocol is summarized as follows: Alice and Bob share two noisy Bell-pairs, i.e. both have two qubits, each representing a "half" of a noisy Bell pair, and they first apply local operations to their respective parts of the Bell-pairs; next, they measure one Bell-pair and classically communicate their outcomes. Depending on the entanglement distillation protocol and the outcomes they either keep or discard the unmeasured pair. The basic step is applied to all pairs of the initial states, which comprises one distillation round. This distillation round is iterated where output states of the previous round are used as inputs for the next round. In the limit, a noiseless entanglement distillation protocol outputs a perfect Bell-pair (implying that Eve is factored out). \newline
Here, we allow for any type of noise acting (independently) on the single- and two-qubit gates appearing in the protocol \footnote{We assume that the noise characteristics of the quantum gates are constant throughout the protocol.}. Using the results of \cite{DuerStd}, by utilizing random basis changes and adding additional noise, any such general noise can be brought to a standard form: depolarizing noise for imperfect single- and two-qubit CNOT-type operations, see \ref{app:sec:eppiid}. Thus, it is sufficient to address noise in such standard form.

For such noise, one can analytically show \cite{Duer} that for the BBPSSW protocol \cite{Bennett}, there exists a unique attracting fixed point of the protocol which only depends on the noise parameters. That is, whenever the fidelity of the initial states is above some minimum fidelity $F_{\text{min}}$, depending on the noise parameters, the protocol converges towards that unique fixed point which we denote by $\sigma^{\alpha;\text{B}}_{AB}$. Observe that $\sigma^{\alpha;\text{B}}_{AB}$ is related to $\sigma^{\alpha, \mathcal{P},l}_{ABE}$ of (\ref{eq:asympstate}) by letting $\mathcal{P} = \text{B}$ and tracing out Eves system, i.e. $\sigma^{\alpha;\text{B}}_{AB} = \ptr{E}{\sigma^{\alpha, \text{B},l}_{ABE}}$. In particular, we mean by $\mathcal{P} = \text{B}$ that the BBPSSW protocol is used for entanglement distillation. We find that the output state $\sigma^{N}_{AB}$, where $N = \log_2 n$ denotes the number of successfully completed distillation layers, satisfies $\| \sigma^{N}_{AB} - \sigma^{\alpha;\text{B}}_{AB} \|_{1} \leq \epsilon_{\text{B}}$, where $\epsilon_{\text{B}}$ is a function of $N$, and it holds that $\epsilon_{\text{B}} \leq F(n) \in O\left(n^{-b_{\text{B}}(\alpha)} \right)$ and $0 < b_{\text{B}}(\alpha) \leq \log_2 3 - 1$.
\newline
For the entanglement distillation protocol of Deutsch et. al. \cite{bib.dejmps} (referred to as the DEJMPS protocol) the fixed point analysis is more complicated. In the noiseless case, DEJMPS was proven to have a unique attracting fixed point \cite{Macchiavello}. For the noisy case, we can only provide extensive numerical evidence that there exists a unique attracting fixed point, depending on the noise parameters only which we denote by $\sigma^{\alpha;\text{D}}_{AB}$, see \ref{sec.sup.dejmps}. Again, observe that $\sigma^{\alpha;\text{D}}_{AB}$ is related to $\sigma^{\alpha, \mathcal{P},l}_{ABE}$ of (\ref{eq:asympstate}) by setting $\mathcal{P} = \text{D}$ and tracing out Eves system, i.e. $\sigma^{\alpha;\text{D}}_{AB} = \ptr{E}{\sigma^{\alpha, \text{D},l}_{ABE}}$. We numerically find that for the state $\sigma^{N}_{AB}$ obtained after successfully completing $N = \log_2 n$ layers of distillation that $\| \sigma^{N}_{AB} - \sigma^{\alpha;\text{D}}_{AB} \|_{1} \leq \epsilon_{\text{D}}$ where $\epsilon_{\text{D}}$ is a function of $N$, and it holds that $\epsilon_{\text{D}} \leq F(n) \in O\left(n^{-b_{\text{D}}(\alpha)} \right)$. $b_{\text{D}}(\alpha)$ is a positive function. We note that a similar analysis, but also with analytic findings for the noiseless DEJMPS protocol was first performed in \cite{Macchiavello}. \newline
We reiterate that we assume for our analysis that all basic distillation steps succeed, since we deal with failures due to the entanglement distillation protocol with a quadratic overhead in terms of initial states, see \ref{app:sec:robust}. \newline
The final state of the entanglement distillation protocol $\mathcal{P}$ in the $ok-$branch, $\sigma_{AB}$, depends on whether the parameter estimation on $\sqrt{n}$ initial states was accurate or not. The latter occurs with an exponentially small probability in terms of initial states, see the discussion of the robustness of the protocol in \ref{app:sec:robust}. This in turn implies that the parameter estimation was accurate with probability exponentially close to unity. Therefore the results  regarding $n$ i.i.d. initial states as input to the distillation protocol $\mathcal{P}$ above imply that
\begin{align}
p_{\rho} \| \sigma_{AB} - \sigma^{\alpha;\mathcal{P}}_{AB}\|_1 \leq \epsilon_{\mathcal{P}}(n) + 2p_{\mathrm{PE}} \leq \epsilon'_{\mathcal{P}}(n) =: \varepsilon_{\mathcal{P}}(n+\sqrt{n}) \label{eq:scale:iid:dist}
\end{align}
where $p_{\mathrm{PE}} \in O(\exp(-\sqrt{n}))$ for all i.i.d. inputs $\rho_{AB}^{\otimes n + \sqrt{n}}$. This equation attains exactly the same form for both protocols with the difference in the labels, so if we substitute $\mathcal{P}$ with $\mathrm{B}$ (by writing, for example $\epsilon_{\mathrm{B}}(n)$) we refer to the BBPSSW protocol, where substituting $\mathcal{P}$ with $\mathrm{D}$ refers to the DEJMPS protocol. In similar fashion we refer from now by $\epsilon_{\mathcal{P}}(n)$ to $\epsilon'_{\mathcal{P}}(n)$ for the sake of clarity. So to summarize, the distance for $n+\sqrt{n}$ i.i.d. initial states in the $ok-$branch of the protocol is bounded by $\varepsilon_{\mathcal{P}}(n+\sqrt{n})$. \newline
Since, in the abort case, the outputs of the overall protocol $\E^\alpha$ and the ideal protocol $\F^{\alpha}$ are identical we obtain that
\EQ{
\| (\E^\alpha  -\F^\alpha ) (\rho_{AB}^{\otimes n}) \|_{1} = p_{\rho} \| \sigma_{AB} - \sigma^{\alpha;\mathcal{P}}_{AB} \|_1
 \leq \varepsilon_{\mathcal{P}}(n) \label{eq:local-iid},
}
where the probability $p_{\rho}$ depends on the initial state $\rho$ for both protocols and corresponds to the probability of parameter estimation succeeding
and completing $\log_{2}(n-\sqrt{n})$ distillation layers successfully for initial state $\rho$. Hence, in both cases, the final distance to the respective fixed points scales polynomial in terms of $n$. \newline
The functions $b_{\text{B}}(\alpha)$ and $b_{\text{D}}(\alpha)$ of the local noise level $\alpha$ govern the rate of convergence of the real protocol to the ideal protocol in the i.i.d case for entanglement distillation protocols. We numerically found that these functions monotonically increase as the local noise rate $\alpha$ tends to zero \ref{app:sec:eppiid}. Thus, increasing the fidelity of local devices (through e.g. fault tolerance) directly influences the rate of convergence, which in turn governs the confidentiality level. \newline
In contrast to $b_{\text{B}}(\alpha)$, the function $b_{\text{D}}(\alpha)$ is not upper bounded, which implies that for certain noise parameters $\alpha$ the DEJMPS protocol needs to perform fewer distillation rounds than the BBPSSW protocol to achieve the required confidentiality levels. This fast convergence is crucial for the powerful post-selection technique \cite{bib.postselect} for non i.i.d. initial states, which is not applicable for the BBPSSW protocol.
\newline
Now we use the established fixed point properties of entanglement distillation protocols for i.i.d. initial states to show that similar results hold for arbitrary initial states.

\subsection{Entanglement distillation for arbitrary inputs}\label{sec:conf:arb}

In generalizing the previous results to arbitrary initial states we make use of the de Finetti theorem \cite{bib.oneandahalf}. The basic de-Finetti results guarantee that the reduced state $\text{tr}_{n-k}\left(\rho^{(n)}_{AB}\right)$ of a permutation-invariant $n-$partite state $\rho^{(n)}_{AB}$ is close to an i.i.d state $\int \sigma_{AB}^{\otimes k}\ d\sigma$, with distance which scales as $O(k/n).$ This enables the following Lemma.
\begin{lemma}\label{lem.localconvergence}
Let $n,k \in \N$ where $k \leq n$. Furthermore, let $\E^{s \& t}$ be the real protocol and $\F^{s \& t}$ the ideal protocol including symmetrization and the tracing out of $n-k$ pairs. Moreover, let $\rho_{AB}$ be a bipartite mixed state of $n$ systems shared by Alice and Bob and let $\E$ and $\F$ denote the real and ideal protocol after symmetrization and tracing out $n-k$ pairs. Then
\begin{align}\label{eq.basicdefin}
\| \E^{s \& t}(\rho_{AB}) - \F^{s \& t}(\rho_{AB}) \|_1 \leq \frac{64k}{n} + \max_{\mu_{AB}} \|\E(\mu^{\otimes k}_{AB}) - \F(\mu^{\otimes k}_{AB})\|_1
\end{align}
\end{lemma}
\begin{proof}
Let $\rho_{AB}$ be a mixed state. After Alice and Bob apply a symmetrization they share a permutation invariant state $\tilde{\rho}_{AB}$. Thus we can apply Theorem II.7 of \cite{bib.oneandahalf} and have for $\xi^k_{AB} := \ptr{n-k}{\tilde{\rho}_{AB}}$ the inequality $\|\xi^k_{AB} - \int \mu^{\otimes k}_{AB} dm(\mu_{AB})\|_1 \leq 32k/n$ for some probability measure $m$ on the set of mixed states on $AB$. Moreover we note that $\E$ and $\F$ are CPTP maps. We define $\tau_k := \int \mu^{\otimes k}_{AB} dm(\mu_{AB})$. A straightforward computation shows
\begin{align*}
\|\E^{s \& t}(\rho_{AB}) - \F^{s \& t}(\rho_{AB})\|_1 & = \|\E(\xi^k_{AB}) - \F(\xi^k_{AB})\|_1 \leq \|\E(\xi^k_{AB}) - \E(\tau_k)\|_1 + \|\E(\tau_k) - \F(\xi^k_{AB})\|_1 \\
& \leq \|\E(\xi^k_{AB}) - \E(\tau_k)\|_1 + \|\E(\tau_k) - \F(\tau_k)\|_1 + \|\F(\tau_k) - \F(\xi^k_{AB})\|_1 \\
& \leq 2 \|\tau_k - \xi^k_{AB} \|_1 + \|\E(\tau_k) - \F(\tau_k)\|_1 \leq \frac{64k}{n} + \left\|(\E-\F)\left(\int \mu^{\otimes k}_{AB} dm(\mu_{AB}) \right) \right\|_1 \\
& \leq \frac{64k}{n} + \max_{\mu_{AB}} \left\|(\E-\F)\left(\mu^{\otimes k}_{AB} \right) \right\|_1
\end{align*}
which completes the proof.
\end{proof}
Therefore the application of the de-Finetti theorem introduces an additive term $\frac{64k}{n}$ when reducing arbitrary initial states to i.i.d. initial states. As the right hand side of (\ref{eq.basicdefin}) is independent of the initial state $\rho_{AB}$, (\ref{eq.basicdefin}) holds for all initial states $\rho_{AB}$. \newline
In (\ref{eq.basicdefin}) we have omitted the superscript $\alpha$ characterizing the noise level, and we will use it only if it is specifically needed. Inequality (\ref{eq.basicdefin}) implies that the properties of the fixed point (uniqueness, attractivity, noise-dependence) also hold for arbitrary initial states, if the protocol is prepended by symmetrization and a trace-out step. This enables us to prove the confidentiality criterion of Definition \ref{def.confidentiality} for entanglement distillation protocols, where the noise transcripts of L are not leaked, which will, in turn, imply the confidentiality criterion (\ref{eqn.confidentiality}) whenever the noise transcripts are leaked.

\subsection{Confidentiality of entanglement distillation protocols}\label{sec:conf:1}

The inequality in (\ref{eq:local-iid}) establishes the local properties of the protocol, and is more-or-less typical for studies of the convergence of entanglement distillation protocols in the i.i.d. case. However, it falls short of the complete characterization captured by the confidentiality criterion (\ref{eqn.confidentiality}) in two ways: first, the input states are restricted (i.i.d.); second, it fails to consider the purifying system of Eve \footnote{Technically, inequality (\ref{eq:local-iid}) is a statement about the operator norm-induced distance on maps, where expression of (\ref{eqn.confidentiality}) is the completely bounded diamond norm, relevant for security statements.}, vital in cryptographic contexts.
While the prior issue is the subject of de-Finetti and post-selection-type reductions, the latter issue can be a problem in general, as small distance of corresponding subsystems does not imply a small distance of the total systems.

However, we can resolve this issue by using the fixed point properties of entanglement distillation protocols. More precisely, we relate the two distances by the following general Lemma, proven in \ref{app:sec:defin}.
\begin{lemma}\label{thm.globalclosenessfactor}
Let $\rho$ be an arbitrary mixed state shared by Alice and Bob and let $\ket{\psi}_{ABE}$ be a purification thereof held by Eve. Furthermore, let $\mathcal{P}_1$ correspond to a (distillation-type) real protocol and $\mathcal{P}_2$ correspond to the associated (distillation-type) ideal protocol, i.e. 
\begin{align*}
\mathcal{P}_1 (\rho) &= p_{\rho} \sigma_{AB} \otimes \dm{ok} + (1-p_\rho) \sigma^{\perp}_{AB} \otimes \dm{fail}, \\
\mathcal{P}_2(\rho) &= p_{\rho} \sigma^{\alpha}_{AB} \otimes \dm{ok} + (1-p_\rho) \sigma^{\perp}_{AB} \otimes \dm{fail}. \notag
\end{align*}
where $\alpha$ characterizes the level of the noise, $\sigma^{\alpha}_{AB}$, and $\sigma^{\perp}_{AB}$ are two fixed two qubit states. Furthermore, let $\mathcal{P}_1$ and $\mathcal{P}_2$ satisfy the following properties:
\begin{enumerate}
	\item The noise transcripts do not leak to Eve.
	\item The protocol $\mathcal{P}_1$ guarantees to converge towards some state $\sigma^{\alpha}_{AB}$ within the ok-branch of the protocol and $\max_{\mu_{AB}} \|(\mathcal{P}_1 - \mathcal{P}_2)(\mu_{AB})\|_1 \leq \varepsilon$.
\end{enumerate}
Then it holds that
\begin{align}
\| (\mathcal{P}_1 \otimes id_{E} -\mathcal{P}_2 \otimes id_{E} )(\ketbra{\psi}{\psi}{ABE})\|_1 \leq (34 \cdot 4^8 +1) \varepsilon. \label{eq.globalclosefactor}
\end{align}
\end{lemma}
The factor $34 \cdot 4^8 +1$ arises as an upper bound on the distance of the given states from states in product form based on the notion of non-steerability we introduce (see \ref{app:sec:defin} for details). In our computations we managed to prove the key lemma in a manner which is proportional to the dimension of the systems, more precisely, the overall size of the corresponding density matrix. It may be the case that the bound of Lemma \ref{thm.globalclosenessfactor} could hold without the dependence on the system size (and indeed, with smaller constants), however this was not necessary for our purposes. \newline
Lemma \ref{thm.globalclosenessfactor} is vital as it allows us to employ the de-Finetti theorem \cite{bib.oneandahalf}. Hence, for the protocols $\E^{s\&t}$ and $\F^{s\&t}$, by combining Lemma \ref{lem.localconvergence} with Lemma \ref{thm.globalclosenessfactor}, we obtain the following Theorem.
\begin{theo}[de-Finetti-based reduction technique]
Let $\E^{s \& t}$ be the real protocol and $\F^{s \& t}$ the ideal protocol including symmetrization and the tracing out of $n-k$ pairs, taking $n$ input pairs and $k \leq n$ and utilizing entanglement distillation protocol $\mathcal{P}$. Then we have
\begin{align}
\max_{\ket{\psi}_{ABE}} \|(\E^{s \& t} \otimes id_{E})(\ketbra{\psi}{\psi}{}) - (\F^{s \& t} \otimes id_{E})(\ketbra{\psi}{\psi}{})\|_1 \leq (34 \cdot 4^8 +1) \left(\frac{64k}{n} + \varepsilon_{\mathcal{P}}(k) \right)  \label{eq:defin:tech}
\end{align}
where $\varepsilon_{\mathcal{P}}(k)$ denotes the maximum distance of the real and ideal protocol without symmetrization and tracing out step using entanglement distillation protocol $\mathcal{P}$ in the $ok-$branch for $k$ i.i.d. initial states, i.e. Eq. (\ref{eq:local-iid}).
\end{theo}
\begin{proof}
Suppose Eve prepares a purification $\ket{\psi}_{ABE}$ of the state $\rho_{AB}$ shared by Alice and Bob. Recall that the real and ideal protocol including symmetrization and the tracing out of $n-k$ pairs applied to initial state $\rho_{AB}$ read as 
\begin{align*}
\E^{s \& t} (\rho_{AB}) = p_{\rho} \sigma_{AB} \otimes \dm{ok} + (1-p_{\rho}) \sigma^{\perp}_{AB} \otimes \dm{fail}, \\
\F^{s \& t} (\rho_{AB}) = p_{\rho} \sigma^{\alpha,\mathcal{P}}_{AB} \otimes \dm{ok} + (1-p_{\rho}) \sigma^{\perp}_{AB} \otimes \dm{fail}
\end{align*}
and observe that we have for the initial state $\rho_{AB}$ by Lemma \ref{lem.localconvergence} that
\begin{align}
\| (\E^{s \& t} -\F^{s \& t})(\rho_{AB})\|_1 = p_{\rho} \| \sigma_{AB} - \sigma^{\alpha,\mathcal{P}}_{AB} \|_1 \leq \left(\frac{64k}{n} + \max_{\mu_{AB}} \left\|(\E-\F)\left(\mu^{\otimes k}_{AB} \right) \right\|_1 \right) \label{eq:defin:proof:2}
\end{align}
where $\E$ and $\F$ denote the real and ideal protocol after symmetrization and tracing out $n-k$ pairs. Since the right-hand side of (\ref{eq:defin:proof:2}) is independent of the initial state $\rho_{AB}$ it holds for all initial states of the protocol. Therefore, the properties of the fixed point (unique, attracting and depending on the noise parameters only) translate from i.i.d. initial states to arbitrary initial states. Hence the protocol guarantees that it converges towards the fixed point of the entanglement distillation protocol. \newline
Additionally, by inserting (\ref{eq:local-iid}) in (\ref{eq:defin:proof:2}) we find
\begin{align}
\| (\E^{s \& t} -\F^{s \& t})(\rho_{AB})\|_1 \leq \left(\frac{64k}{n} + \varepsilon_{\mathcal{P}}(k) \right) \label{eq:defin:proof:3}.
\end{align}
This implies that the real protocol indeed converges towards the fixed point, and, thus we can apply Lemma \ref{thm.globalclosenessfactor} to the protocols $\E^{s \& t}$ and $\F^{s \& t}$ for the purification $\ket{\psi}_{ABE}$ of $\rho_{AB}$ and we find by using (\ref{eq:defin:proof:3}) that
\begin{align}
\|(\E^{s \& t} \otimes id_{E})(\ketbra{\psi}{\psi}{}) - (\F^{s \& t} \otimes id_{E})(\ketbra{\psi}{\psi}{})\|_1 \leq (34 \cdot 4^8 +1) \left(\frac{64k}{n} + \varepsilon_{\mathcal{P}}(k) \right). \label{eq:defin:proof:4}
\end{align}
Taking the maximum in (\ref{eq:defin:proof:4}) completes the proof.
\end{proof}
Thus, we can reach arbitrary confidentiality levels, however at the cost of wasting some pairs. The scaling of the confidentiality parameter, i.e. the right-hand side of (\ref{eq:defin:tech}), is linear in the number of initial states $n$, due to the use of the ``basic'' de Finetti approach. 
\newline
If the local noise is low, we can do better in terms of scaling and efficiency, using the post-selection technique \cite{bib.postselect}. For that purpose, we first establish a result similar to (\ref{eq.globalclosefactor}) by using the fact that the resulting state of the protocol, including L, is pure, see \ref{app:sec:eppiid}. More precisely, we have the following Lemma, proven in \ref{app:sec:post}.
\begin{lemma}\label{lem:forpost}
Let $\E$ be the real protocol which guarantees to converge towards a unique and attracting fixed point depending on the noise parameter only and let $\F$ be the ideal protocol. Furthermore let $\rho$ be a mixed state (consisting of $n$ systems) shared by Alice and Bob. If the extension of $\E$ and $\F$ to the system of L satisfies $\|\E_{\text{L}}(\rho) - \F_{\text{L}}(\rho)\|_1 \leq \varepsilon(n)$, then
\begin{align*}
\|(\E \otimes id_{E'})(\ketbra{\psi}{\psi}{ABE'}) &- (\F \otimes id_{E'})(\ketbra{\psi}{\psi}{ABE'})\|_1 \leq 4 \sqrt{\varepsilon(n)}
\end{align*}
for all purifications $\ket{\psi}_{ABE'}$ of $\rho$.
\end{lemma}
This Lemma allows us to prove the closeness on any purification from the closeness of the reduced systems, and finally to derive confidentiality from the performance of the ideal protocol via the following Theorem.
\begin{theo}[Post-selection-based reduction technique]
Let $\E^{s}$ be the real protocol and $\F^{s}$ the ideal protocol preceded by a symmetrization step operating on $n$ input pairs. Furthermore let $\max_{\mu_{AB}} \|\E(\mu^{\otimes n}_{AB}) - \F(\mu^{\otimes n}_{AB})\|_1 \leq \varepsilon_{\mathcal{P}}(n)$, see (\ref{eq:local-iid}), where $\E$ and $\F$ denote the sub-protocols after symmetrization (i.e. the protocols without the symmetrization step) and $\mathcal{P}$ the entanglement distillation protocol. Then we have
\begin{align}
\max_{\ket{\psi}_{ABE'}} \|(\E^{s} \otimes id_{E'})(\ketbra{\psi}{\psi}{}) - (\F^{s} \otimes id_{E'})(\ketbra{\psi}{\psi}{})\|_1 \leq 4 \sqrt{2} g_{n,d} \sqrt[4]{\varepsilon_{\mathcal{P}}(n)} \label{eq:post:redudction}
\end{align}
where $g_{n,d} = {n+15 \choose n}$.
\end{theo}
\begin{proof}
We observe that $\E^{s}$ and $\F^{s}$ are permutation invariant maps due to the symmetrizazion step. Thus we can apply the post-selection technique of \cite{bib.postselect} which implies
\begin{align}\label{inequ.post.1}
\max_{\ket{\psi}_{ABE'}} \|(\E^{s} \otimes id_{E'})(\ketbra{\psi}{\psi}{}) - (\F^{s} \otimes id_{E'})(\ketbra{\psi}{\psi}{})\|_1 \leq g_{n,d} \|(\E^{s} \otimes id_{E'})(\ketbra{\tau}{\tau}{ABE'}) - (\F^{s} \otimes id_{E'})(\ketbra{\tau}{\tau}{ABE'})\|_1
\end{align}
where $\ket{\tau}_{ABE'}$ is a purification of the de-Finetti Hilbert-Schmidt state, hence $\ptr{E'}{\ketbra{\tau}{\tau}{ABE'}} = \int \mu^{\otimes n}_{AB} d \eta(\mu) =: \tau'$ where $\eta$ is the measure induced by the Hilbert-Schmidt metric on $\text{End}(\C^4)$. Furthermore, we note that we have for the extensions of $\E^{s}$ and $\F^{s}$ to L, i.e. the maps $\E^{s}_{\mathrm{L}}$ and $\F^{s}_{\mathrm{L}}$, that
\begin{align}
\|\E^{s}_{\mathrm{L}}(\tau') - \F^{s}_{\mathrm{L}}(\tau')\|_1 = \left\|(\E^{s}_{\mathrm{L}} - \F^{s}_{\mathrm{L}})\left(\int \mu^{\otimes n}_{AB} d \eta(\mu) \right) \right\|_1 \leq  \max_{\mu_{AB}} \left\|(\E_{\mathrm{L}}-\F_{\mathrm{L}})\left(\mu^{\otimes n}_{AB} \right) \right\|_1. \label{eq:post:proof:1}
\end{align}
According to \ref{sec:det:l}, which implies that the distance including L  scales as the square root of the $1-$norm induced distance without L, i.e. Alice and Bob only, we find for (\ref{eq:post:proof:1}) by using the assumption $\max_{\mu_{AB}} \|\E(\mu^{\otimes n}_{AB}) - \F(\mu^{\otimes n}_{AB})\|_1 \leq \varepsilon_{\mathcal{P}}(n)$ that
\begin{align}
\left\|(\E_{\mathrm{L}}-\F_{\mathrm{L}})\left(\mu^{\otimes n}_{AB} \right) \right\|_1 \leq 2 \sqrt{\left\|(\E-\F)\left(\mu^{\otimes n}_{AB} \right) \right\|_1} \leq 2 \sqrt{\varepsilon_{\mathcal{P}}(n)}.
\end{align}
As $\ket{\tau}_{ABE'}$ is a purification of $\tau'$ we can apply Lemma \ref{lem:forpost} which gives, for (\ref{inequ.post.1}),
\begin{align*}
\max_{\ket{\psi}_{ABE'}} \|(\E^{s} \otimes id_{E'})(\ketbra{\psi}{\psi}{}) - (\F^{s} \otimes id_{E'})(\ketbra{\psi}{\psi}{})\|_1 & \leq 4 g_{n,d} \sqrt{\max_{\mu_{AB}} \left\|(\E_{\mathrm{L}}-\F_{\mathrm{L}})\left(\mu^{\otimes n}_{AB} \right) \right\|_1} \\
& \leq 4 g_{n,d} \sqrt{2 \sqrt{\varepsilon_{\mathcal{P}}(n)}} \\
& = 4 \sqrt{2} g_{n,d} \sqrt[4]{\varepsilon_{\mathcal{P}}(n)}
\end{align*}
which completes the proof.
\end{proof}
Observe that $\varepsilon_{\mathcal{P}}(n)$, which governs the rate of convergence of the overall protocol, relates to the rate of convergence of the entanglement distillation protocol $\mathcal{P}$ via $\varepsilon_{\mathcal{P}}(n) = \epsilon_{\mathcal{P}}(n - \sqrt{n})$, as $\sqrt{n}$ initial states  are used for parameter estimation. \newline
We remind the reader that the preprocessing steps (symmetrization, tracing out) of the entanglement distillation protocol and the Lemmas of this section are non-trivial and crucial for the proof of the de-Finetti-based and post-selection-based reduction technique. \newline
Furthermore we point out that the proof regarding the BBPSSW protocol is analytic and necessarily relies on the de-Finetti-based reduction technique because of its slow convergence rate. The rate of convergence for the BBPSSW protocol can easily be derived, see \ref{app:sec:eppiid} for details. For the DEJMPS protocol it turns out that we have polynomial scaling depending on the noise parameter $\alpha$, i.e. $\max_{\sigma_{AB}} \left\|(\E-\F)\left(\sigma^{\otimes n}_{AB} \right) \right\|_1 \leq \varepsilon_{\text{D}}(n) \leq O(n^{-b_{\mathrm{D}}(\alpha)})$, see (\ref{eq:local-iid}). \newline
However, the protocol needs to converge sufficiently quickly, as the post-selection technique incurs a multiplicative increase in the effective distance between real and ideal protocols, which scales as a (15 degree) polynomial in $n$, see (\ref{eq:post:redudction}). The resulting confidentiality level scales therefore as $O(n^{15-b_{\mathrm{D}}(\alpha)/4})$, which leads to an acceptable noise level that is rather low, e.g. about $10^{-19}$ for the DEJMPS protocol in the setting of binary pairs \footnote{For this simplified analysis we assumed that no parameter estimation is necessary.}, see \ref{sec:det:l}. This very low rate is due to the polynomial factor introduced by applying the post-selection technique, i.e. $g_{n,d}$ in (\ref{eq:post:redudction}) with $d=4$. Observe that these small rates are determined by properties of recurrence-type entanglement distillation protocols, i.e. $b(\alpha)$ for the recurrence-type entanglement distillation protocols studied here, and may be improved by either considering hashing-type protocols \cite{BennettHash} or through fault-tolerant constructions. Indeed, the noise threshold for fault-tolerant quantum computation also applies to this case, yielding a tolerable noise level of about $10^{-4}$. We reiterate that the post-selection technique is not applicable to the BBPSSW protocol, due to its slow convergence.

\subsection{Confidentiality of entanglement distillation protocols when the noise transcripts leak}\label{sec:conf:2}

Finally, we provide confidentiality guarantees for entanglement distillation protocols when the noise transcripts are leaked to Eve. For that purpose, we relate the confidentiality criterion (\ref{eqn.confidentiality}) for protocols where the noise transcripts are leaked to the earlier results. More formally, we have the following Theorem.
\begin{theo}\label{lem.leaklesstoleak}
Let $\E$ be the real protocol and $\F$ be the ideal protocol satisfying the assumptions of Lemma \ref{thm.globalclosenessfactor}. Furthermore, let $\E^l$ denote the real and $\F^l$ the ideal protocol when the noise transcripts leak to Eve. Then
\EQ{\|(\E \otimes \id_E - \F \otimes \id_E) (\dm{\psi})\|_1 \leq \varepsilon(n) \ \textup{, implies} \label{eq.leakreduction}}
\begin{align*}
\|(\E^l \otimes \id_E - \F^l \otimes \id_E) (\dm{\psi}) \|_1 \leq 2 \sqrt{\varepsilon(n)}
\end{align*}
for all purifications $\ket{\psi}_{ABE}$ of initial state $\rho_{AB}$ consisting of $n$ systems.
\end{theo}
The proof, see \ref{app:sec:leak}, uses the unitary equivalence of purifications. Theorem \ref{lem.leaklesstoleak} establishes via (\ref{eq.leakreduction}) that if an entanglement distillation protocol is $\varepsilon-$confidential according to Definition \ref{def.confidentiality} then the protocol is $2\sqrt{\varepsilon}-$confidential if the noisy apparatus leaks the noise transcripts. \newline

\section{Discussion}\label{sec:disc}

We have shown that recurrence-type entanglement distillation protocols ensure private entanglement without referring to the asymptotic limit. This holds true even when the local devices are noisy, and when the potential eavesdropper is able to completely monitor the operation of these devices in run-time (i.e., the noisy apparatus leaks information about the realized noise processes).  If the noise transcripts are not leaked, Eve is ``factored out'' -- in tensor product with Alice and Bob, and only classically correlated otherwise. Our protocol can, for instance, be used to realize confidential quantum channels by means of teleportation - the only information that may leak to Eve after teleportation is which noise map was applied to the sent state, but nothing about the state itself (see \ref{app:sqc} for details).
More generally, our results imply the confidentiality of the protocols in arbitrary settings (beyond the application to quantum channels), thus opening the way for the confidential realization of various quantum tasks: from establishing quantum channels and quantum networks, to applications such as distributed quantum computation. Aside from cryptographic aspects, the proposed protocol can be used to generate high quality entanglement from non-iid sources.

\noindent\textbf{Acknowledgments:\\}
We acknowledge the support by the Austrian Science Fund (FWF) through the SFB FoQuS F 4012 and project P28000-N27. AP and VD are grateful to Christopher Portmann for useful discussions, comments, and advice concerning technical aspects of this work.

\appendix 
%


\section{Entanglement distillation for i.i.d. inputs}\label{app:sec:eppiid}

\subsection{The DEJMPS protocol}\label{sec.sup.dejmps}
We first provide an overview of the DEJMPS protocol \cite{bib.dejmps} and then extend the description incrementally to our proposed setting (including L and Eve).

The DEJMPS protocol is a recurrence-type entanglement distillation protocol which combines several noisy copies of a  mixed state $\rho$ to distill a state arbitrarily close to the maximally entangled state $\ket{B_{00}}{}$, where $\ket{B_{ij}} = (\id \otimes \px^j \pz^i) (\ket{00} + \ket{11})/\sqrt{2}$ for $i \in \lbrace 0,1 \rbrace$ and $j \in \lbrace 0,1 \rbrace$, provided that the fidelity $F=\bra{B_{00}}{} \rho \ket{B_{00}}{}$ satisfies $F > 1/2$ for the noiseless case. If the apparatus is noisy, then the minimal required fidelity $F$ needs to satisfy $F > F_{min}$ (where $F_{min}$ depends on the noise level of the apparatus) to achieve distillation. For more details on recurrence-type entanglement distillation protocols in general we refer the interested reader to \cite{bib.reviewepp}. A basic step of the DEJMPS protocol is as follows:

\begin{protocol}[h!]
\caption{Basic step of the DEJMPS protocol}
\label{protocol.rd}
\begin{algorithmic}[1]
    \REQUIRE Input state of Alice and Bob: $\rho^{(a_1, b_1)} \otimes \rho^{(a_2, b_2)}$
    \STATE \label{enu.oxford.step1} Alice and Bob apply the local basis change $U_x = e^{-i \pi/4 \px^{(a_1)}} \otimes e^{i \pi/4 \px^{(b_1)}} \otimes e^{-i \pi/4 \px^{(a_2)}} \otimes e^{i \pi/4 \px^{(b_2)}}$:
	\begin{align*}
	 U_x \left(\rho^{(a_1, b_1)} \otimes \rho^{(a_2, b_2)}\right) U_x^\dagger.
	\end{align*}
    \STATE \label{enu.oxford.bcnot} Alice and Bob apply a bilateral CNOT (BCNOT):
	\begin{align*}
	\left(\text{CNOT}_{a_1 \to a_2} \otimes \text{CNOT}_{b_1 \to b_2} \right) \rho^{(a_1, b_1)} \otimes \rho^{(a_2, b_2)} \left(\text{CNOT}_{a_1 \to a_2} \otimes \text{CNOT}_{b_1 \to b_2} \right)^\dagger.
	\end{align*}
    \STATE \label{enu.oxford.meas} Alice and Bob apply a $ \sigma_z^{(a_2)} = \sigma_z \otimes id$ and a $\sigma_z^{(b_2)} = id \otimes \sigma_z$ measurement
    \STATE \label{enu.oxford.steplast} Alice and Bob communicate their measurement outcomes, $z_a$ and $z_b$ respectively, over a classical authentic channel
    \IF{$z_a = z_b$}
        \STATE Alice and Bob keep the subsystems $a_1$ and $b_1$ of step \ref{enu.oxford.bcnot}
        \STATE Alice and Bob discard the measured subsystems $a_2$ and $b_2$
    \ELSE
        \STATE Alice and Bob discard both pairs
    \ENDIF
\end{algorithmic}
\end{protocol}

Hence, we can write one basic distillation step of the DEJMPS protocol as the linear map $O_{\text{2-EPP}} (\rho \otimes \rho) = O'_{\text{2-EPP}} (\rho \otimes \rho) O'^\dagger_{\text{2-EPP}}$ where

\begin{align*}
O'_{\text{2-EPP}} = \left(id_{a_1,b_1} \otimes P^{(a_2)}_{z} \otimes P^{(b_2)}_{z} \right) \left(\text{CNOT}_{a_1 \to a_2} \otimes \text{CNOT}_{b_1 \to b_2} \right) U_x
\end{align*}

modulo a normalization factor and where $P_{z} = \ketbra{z}{z}{}, \, z \in \lbrace 0,1 \rbrace$ denotes the respective outcome of step \ref{enu.oxford.meas} of Protocol \ref{protocol.rd}.

The basic step is applied to all initial pairs, which comprises one distillation round. This distillation round is iterated where output states of the previous round are used as inputs for the next round. So we summarize the DEJMPS protocol as follows:

\begin{protocol}[h!]
\caption{DEJMPS protocol}
\label{protocol.rd.complete}
\begin{algorithmic}[1]
    \REQUIRE Input state of Alice and Bob: $\bigotimes^{2^n}_{i = 1} \rho^{(a_i,b_i)}$ where $F = \bra{B_{00}}{} \rho^{(a_i,b_i)} \ket{B_{00}}{} > 1/2$ for all $i \in \lbrace 1,..,2^n \rbrace$
    \WHILE{Pairs left for distillation}
		\STATE Apply Protocol \ref{protocol.rd} to all pairs
		\STATE Use the outputs of the previous step as input for the next distillation round
    \ENDWHILE
\end{algorithmic}
\end{protocol}

We remind the reader that the recurrence relations of the protocol (i.e. update functions of the coefficients of an ensemble) are central for the convergence analysis of the DEJMPS protocol. For Bell-diagonal states, i.e. states of the form
\begin{align*}
\rho = p_{00} \ketbra{B_{00}}{B_{00}}{} + p_{11} \ketbra{B_{11}}{B_{11}}{} + p_{01} \ketbra{B_{01}}{B_{01}}{} + p_{10} \ketbra{B_{10}}{B_{10}}{}
\end{align*}
where $\sum_{ij} p_{ij} = 1, \, p_{ij} \geq 0$, a straightforward computation yields the recurrence relations for the DEJMPS protocol to be
\begin{align}\label{eqn.recurrencedejmps}
\tilde{p}_{00} = \frac{p_{00}^2 + p_{11}^2}{N}, \qquad \tilde{p}_{11} = \frac{2 p_{01} p_{10}}{N}, \notag \\
\tilde{p}_{01} = \frac{p_{01}^2 + p_{10}^2}{N}, \qquad \tilde{p}_{10} = \frac{2 p_{00} p_{11}}{N}
\end{align}
where $N = (p_{00} + p_{11})^2 + (p_{01} + p_{10})^2$, see e.g. \cite{bib.dejmps}.

In \cite{Macchiavello} it has been shown analytically that the recurrence relations (\ref{eqn.recurrencedejmps}) converge towards a unique and attracting fixed point provided the initial fidelity with $\ket{B_{00}}$, $p_{00}$, is above $1/2$.

The recurrence relations of the DEJMPS protocol taking independent single qubit white noise, i.e. noise of the form $N \rho = f \rho + (1-f)/4 (\rho + \px \rho \px + \py \rho \py + \pz \rho \pz)$ acting on each qubit of Alice into account, read far more complex. In the presence of noise we have strong numerical evidence that the DEJMPS protocol converges towards a unique and attracting fixed point depending on the noise level $f$ only.

\begin{figure}[htb]
\centering
\hspace*{1.5em}\raisebox{\dimexpr-.5\height-1em}
  {\includegraphics{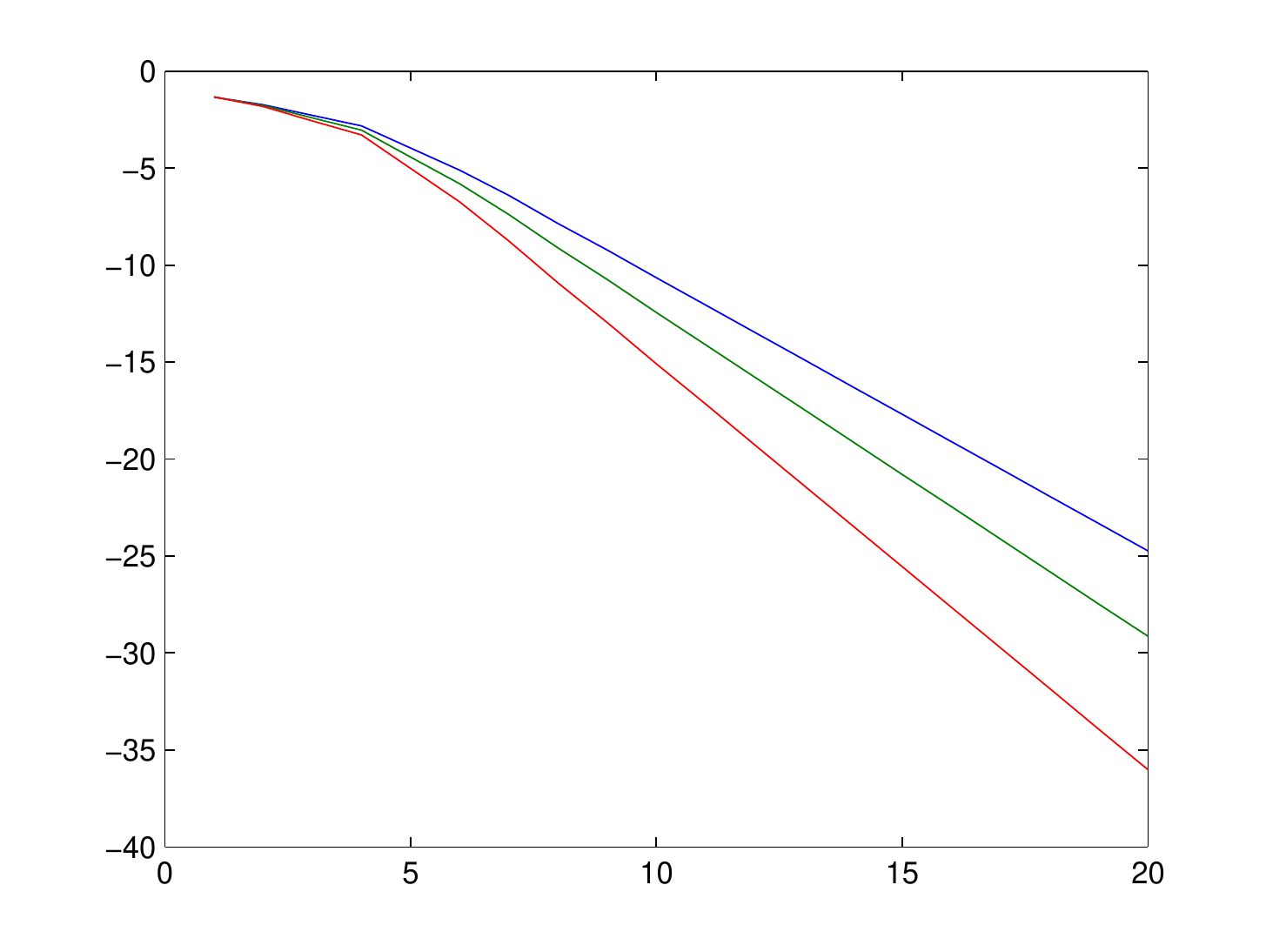}
}%

\leavevmode\smash{\makebox[0pt]{\hspace{-45em}
  \rotatebox[origin=l]{90}{\hspace{15em}
    $\log \|\rho_{fix} - \rho_{n} \|_{1}$}%
}}

\hspace{0pt plus 1filll}\null

$n$

\medskip

\caption[h!]{\label{fig:dejmps} The figure illustrates $\log \|\rho_{\text{fix}} - \rho_{n} \|_{1}$ for different noise parameters $f = 0.97$ (blue), $f = 0.98$ (green) and $f = 0.99$ (red). The fixed point $\rho_{\text{fix}}$ was evaluated for $500$ iterations of the DEJMPS protocol.}
\end{figure}
From figure \ref{fig:dejmps} we suggest a linear relationship between $\log \|\rho_{\text{fix}} - \rho_{n} \|_{1}$ (where $\rho_{\text{fix}}$ and $\rho_n$ denote the fixed point and the state after successfully completing $n$ distillation rounds respectively) and the number of successful distillation rounds $n$. We immediately observe that the slope only depends on the noise parameter $f$, i.e. we have that

\begin{align*}
\log \|\rho_{\text{fix}} - \rho_{n} \|_{1} = a(f) - n b(f).
\end{align*}

Using $\log_2 N = n$, where $N$ denotes the number of input pairs, this implies $\|\rho_{\text{fix}} - \rho_{n} \|_{1} = e^{a(f)} e^{-b(f) \log_2 N} = a'(f) N^{-b'(f)}$, i.e. $\|\rho_{\text{fix}} - \rho_{n} \|_{1}$ scales as $F(N) \in O(N^{-b'(f)})$ as mentioned in the main text. Furthermore we numerically find that the function $b'(f)$ monotonically grows for $f \to 1$.

For two qubit correlated noise, we refer the reader to the analysis including L, as the fixed point and the scaling can be recovered from that analysis by tracing out the system of L.

\subsubsection{Detailed analysis including L}\label{sec:det:l}

We outline the remainder of this section as follows: First we derive the recurrence relations of the DEJMPS protocol in the most general setting, taking the noise applied by L into account as well as assuming that Eve receives the leaked noise transcripts of L. We use those recurrence relations in the next subsection to provide analytical results regarding the fixed point of the recurrence relations, where the inputs are binary pairs and L only applies either $id$ or $\px$ operators. We close the section with numerical results for general i.i.d. Bell-diagonal pairs and the most general noise maps of L.


\vspace*{1cm}
\underline{The recurrence relations}
\vspace*{1cm}

For i.i.d. input states the state of each system subject to distillation at an intermediate distillation round of the DEJMPS protocol is of the form $\ket{\Psi}_{ABEL} = \sum_{i,j,k,l} P_{ijkl} \ket{B_{ij}}_{AB} \ket{kl}_{L} \ket{ijkl}_{E}$, where $P_{ijkl}$ are probability amplitudes, if we assume the noise is leaked to Eve after every distillation round. The system $AB$ models the pair of Alice and Bob, $L$ the system of L (where the content of the register corresponds to the effective noise introduced to $AB$) and $E$ the system of Eve. L applies the noise processes before a basic protocol step to the systems of Alice. Moreover, L keeps track of the effective noise introduced using its system in a sense we clarify later.

In the following we use the notation

\begin{align*}
\sigma_{0,0} = id, \quad \sigma_{0,1} = \px, \quad \sigma_{1,0} = \pz, \quad \sigma_{1,1} = \py
\end{align*}

for the four Pauli-operators. Furthermore we denote by superscripts in brackets particle labels and by superscripts without brackets the power of an operator.

L introduces the noise maps $U_{\alpha_1,\beta_1, \alpha_2, \beta_2} = U^{(a_1)}_{\alpha_1, \beta_1} \otimes U^{(a_2)}_{\alpha_2,\beta_2}$ where  $U^{(a_k)}_{\alpha,\beta} = \sigma^{(a_k)}_{\alpha,\beta} \otimes \left((\px^\alpha) \otimes (\px^\beta)\right)^{(L_k)}$. We observe that applying the noise map $U_{\alpha_1,\beta_1, \alpha_2, \beta_2}$ might flip the contents of the registers $L_1$ and $L_2$ depending on the values of $\alpha_1,\beta_1, \alpha_2$ and $\beta_2$. This enables L to keep track of the noise introduced to a pair.

There are two approaches how L can apply the noise maps $U_{\alpha_1,\beta_1, \alpha_2, \beta_2}$: stochastically in terms of CPTP maps, or coherently in terms of unitaries acting on an enlarged Hilbert space. Here we assume the latter approach, but provide the analysis of the noisy DEJMPS protocol in terms of CPTP maps and purifications.

To show that these are equivalent, first suppose that L owns a register $H$ set to the state $\sum_{\alpha_1, \beta_1, \alpha_2, \beta_2} \sqrt{\tilde{f}_{\alpha_1, \beta_1, \alpha_2, \beta_2}} \ket{\alpha_1 \beta_1 \alpha_2 \beta_2}_{H}$ where $\tilde{f}_{\alpha_1, \beta_1, \alpha_2, \beta_2}$ are the probabilities of applying the respective noise map $U_{\alpha_1,\beta_1, \alpha_2, \beta_2}$. L uses the register $H$ to apply the noise maps $U_{\alpha_1,\beta_1, \alpha_2, \beta_2}$ coherently controlled to the input state $\ket{\Psi}_{ABEL}$. We observe that tracing out $H$ after applying all the noise maps $U_{\alpha_1,\beta_1, \alpha_2, \beta_2}$ in a controlled fashion yields
\begin{align*}
\sum\limits_{\alpha_1, \beta_1, \alpha_2, \beta_2}  \tilde{f}_{\alpha_1, \beta_1, \alpha_2, \beta_2} U_{\alpha_1,\beta_1, \alpha_2, \beta_2} \left(\ketbra{\Psi}{\Psi}{} \otimes \ketbra{\Psi}{\Psi}{}\right) U^\dagger_{\alpha_1,\beta_1, \alpha_2, \beta_2}.
\end{align*}
On the other hand, assume that L applies the noise process in terms of a CPTP map $N$, i.e.
\begin{align*}
N \rho = \sum\limits_{\alpha_1, \beta_1, \alpha_2, \beta_2}  \tilde{f}_{\alpha_1, \beta_1, \alpha_2, \beta_2} U_{\alpha_1,\beta_1, \alpha_2, \beta_2} \left(\ketbra{\Psi}{\Psi}{} \otimes \ketbra{\Psi}{\Psi}{}\right) U^\dagger_{\alpha_1,\beta_1, \alpha_2, \beta_2}.
\end{align*}
We observe that $N \rho$ will be, in general, a mixed state, thus there exists a purification on a larger Hilbert space. As all purifications are unitarily equivalent, see e.g. \cite{Nielsen}, we choose the purification
\begin{align*}
\ket{\Phi} = \sum\limits_{\alpha_1, \beta_1, \alpha_2, \beta_2}  \sqrt{\tilde{f}_{\alpha_1, \beta_1, \alpha_2, \beta_2}} U_{\alpha_1,\beta_1, \alpha_2, \beta_2} \ket{\Psi} \otimes \ket{\Psi} \otimes \ket{\alpha_1 \beta_1 \alpha_2 \beta_2}_{H}.
\end{align*}
Hence $\ptr{H}{\dm{\Phi}} = N \rho$. Furthermore, we observe that the pure state $\ket{\Phi}$ can be generated by applying the unitaries $U_{\alpha_1,\beta_1, \alpha_2, \beta_2}$, coherently controlled by the register $H$, to $\ket{\Psi} \otimes \ket{\Psi} \otimes \left(\sum_{\alpha_1, \beta_1, \alpha_2, \beta_2} \sqrt{\tilde{f}_{\alpha_1, \beta_1, \alpha_2, \beta_2}} \ket{\alpha_1 \beta_1 \alpha_2 \beta_2}_{H} \right)$.

This equivalence allows us to assume that L introduces the noise as a CPTP map, applying $U_{\alpha_1,\beta_1, \alpha_2, \beta_2}$ with respective probabilities $f_{\alpha_1, \beta_1, \alpha_2, \beta_2}$ and purifying the state after the basic distillation step is executed by Alice and Bob.

Since the noise of L is applied before the basic distillation step is executed by Alice and Bob, the result of one noisy distillation step reads as
\begin{align}\label{eqn.overall}
\rho' & = \sum\limits_{\alpha_1, \beta_1, \alpha_2, \beta_2}  \tilde{f}_{\alpha_1, \beta_1, \alpha_2, \beta_2} U_u O'_{\text{2-EPP}} (U^{(a_1)}_{\alpha_1, \beta_1} \otimes  U^{(a_2)}_{\alpha_2, \beta_2}) \left(\ketbra{\Psi}{\Psi}{} \otimes \ketbra{\Psi}{\Psi}{}\right) (U^{(a_1)}_{\alpha_1, \beta_1} \otimes  U^{(a_2)}_{\alpha_2, \beta_2})^\dagger O'^\dagger_{\text{2-EPP}} U_u^\dagger
\end{align}
which needs finally to be purified.

In order to evaluate (\ref{eqn.overall}), we proceed as follows: \newline
\begin{itemize}
\item Step 1: We first compute
\begin{align*}
O'_{\text{2-EPP}} (U^{(a_1)}_{\alpha_1, \beta_1} \otimes  U^{(a_2)}_{\alpha_2, \beta_2}) \ket{\Psi} \otimes \ket{\Psi}.
\end{align*}
which corresponds to the state after the noise map $U^{(a_1)}_{\alpha_1, \beta_1} \otimes  U^{(a_2)}_{\alpha_2, \beta_2}$ is applied by L and the basic distillation step of the entanglement distillation protocol is executed by Alice and Bob.
\item Step 2: We apply the unitary $U_u$, which acts only on L's systems and whose purpose we clarify later, to the previous equality.
\item Step 3: We have to determine the purification held by Eve if the noise is leaked to her. In doing so, we trace out Eve and then provide her with the purification of the resulting state (which corresponds to leaking the noise transcripts to Eve).
\end{itemize}
\underline{Step 1:} We observe that applying the noise map $U^{(a_1)}_{\alpha,\beta}$ to $\ket{\Psi}{}$ yields
\begin{align}\label{eqn.proofrec1}
U^{(a_1)}_{\alpha,\beta} \ket{\Psi}{} & = U^{(a_1)}_{\alpha,\beta} \sum\limits_{i,j,k,l} P_{ijkl} \ket{B_{ij}}_{AB} \ket{kl}_{L} \ket{ijkl}_{E} \\ \notag
& = \sum\limits_{i,j,k,l} P_{ijkl} \ket{B_{(i \oplus \alpha)(j \oplus \beta)}}_{AB} \ket{(k \oplus \alpha)(l \oplus \beta)}_{L} \ket{ijkl}_{E} \\ \notag
& = \sum\limits_{i,j,k,l} P_{(i \oplus \alpha)(j \oplus \beta)(k \oplus \alpha)(l \oplus \beta)} \ket{B_{ij}}_{AB} \ket{kl}_{L} \ket{(i \oplus \alpha)(j \oplus \beta)(k \oplus \alpha)(l \oplus \beta)}_{E}.
\end{align}
This observation suggests the following notational simplifications:
\begin{align*}
P^{\alpha\beta}_{ijkl} = P_{(i \oplus \alpha)(j \oplus \beta)(k \oplus \alpha)(l \oplus \beta)} \quad \text{ and } \quad \ket{e^{\alpha\beta}_{ijkl}}_{E} = \ket{(i \oplus \alpha)(j \oplus \beta)(k \oplus \alpha)(l \oplus \beta)}_{E}.
\end{align*}
Using this notation we rewrite (\ref{eqn.proofrec1}) as $U^{(a_1)}_{\alpha,\beta} \ket{\Psi}{} = \sum_{i,j,k,l} P^{\alpha\beta}_{ijkl} \ket{B_{ij}}_{AB} \ket{kl}_{L} \ket{e^{\alpha\beta}_{ijkl}}_{E}$. This is the state of Alice, Bob, L, and Eve after the  noise map $U^{(a_1)}_{\alpha,\beta}$ is applied by L to the first pair. In order to compute (\ref{eqn.overall}) we define
\begin{align*}
\ket{\Psi''_{\alpha_1, \beta_1, \alpha_2, \beta_2}}{} & = (U^{(a_1)}_{\alpha_1, \beta_1} \otimes  U^{(a_2)}_{\alpha_2, \beta_2})  \ket{\Psi}{} \ket{\Psi}{} \\
& = \sum\limits_{i_1, j_1, i_2, j_2} \sum\limits_{k_1, l_1, k_2, l_2} A^{\alpha_1\beta_1}_{i_1 j_1 k_1 l_1} P^{\alpha_2 \beta_2}_{i_2 j_2 k_2 l_2} \ket{B_{i_1 j_1}}_{AB_1} \ket{B_{i_2 j_2}}_{AB_2} \ket{k_1 l_1}_{L_1} \ket{k_2 l_2}_{L_2} \\
& \hspace*{0.4cm} \otimes \ket{e^{\alpha_1 \beta_1}_{i_1 j_1 k_1 l_1}}_{E_1} \ket{e^{\alpha_2 \beta_2}_{i_2 j_2 k_2 l_2}}_{E_2}
\end{align*}
which corresponds to the state after the noise map $U^{(a_1)}_{\alpha_1, \beta_1} \otimes  U^{(a_2)}_{\alpha_2, \beta_2}$ is applied and
\begin{align}\label{eqn.proofrec00}
\ket{\Psi'_{\alpha_1, \beta_1, \alpha_2, \beta_2}}{} = U_u O_{\text{2-EPP}} \ket{\Psi''_{\alpha_1, \beta_1, \alpha_2, \beta_2}}{}
\end{align}
which is the state after the noise map $U^{(a_1)}_{\alpha_1, \beta_1} \otimes  U^{(a_2)}_{\alpha_2, \beta_2}$, one basic distillation step and the update of L's noise register by $U_u$. Thus we rewrite (\ref{eqn.overall}) as
\begin{align}\label{eqn.proofrec0}
\rho' = \sum\limits_{\alpha_1, \beta_1, \alpha_2, \beta_2}  \tilde{f}_{\alpha_1, \beta_1, \alpha_2, \beta_2} \ketbra{\Psi'_{\alpha_1, \beta_1, \alpha_2, \beta_2}}{\Psi'_{\alpha_1, \beta_1, \alpha_2, \beta_2}}{}.
\end{align}
According to (\ref{eqn.proofrec00}) Alice and Bob apply one basic distillation step of the DEJMPS protocol to the state $\ket{\Psi''_{\alpha_1, \beta_1, \alpha_2, \beta_2}}{}$. Recall that step \ref{enu.oxford.step1} of Protocol \ref{protocol.rd} maps $\ket{B_{ij}}{}$ to $\ket{B_{i (i \oplus j)}}{}$ and that step \ref{enu.oxford.bcnot} maps $\ket{B_{ij}}{} \ket{B_{i'j'}}{}$ to $\ket{B_{(i \oplus i')j}}{} \ket{B_{i'(j \oplus j')}}{}$. Thus we conclude that after step \ref{enu.oxford.step1} and \ref{enu.oxford.bcnot} of Protocol \ref{protocol.rd} the state of Alice, Bob, L, and Eve is
\begin{align}\label{eqn.proofrec2}
\sum\limits_{i_1, j_1, i_2, j_2} \sum\limits_{k_1, l_1, k_2, l_2} P^{\alpha_1 \beta_1}_{i_1 j_1 k_1 l_1} P^{\alpha_2 \beta_2}_{i_2 j_2 k_2 l_2} & \ket{B_{(i_1 \oplus i_2)(i_1 \oplus j_1)}}_{AB_1} \ket{B_{i_2 (i_1 \oplus j_1 \oplus i_2 \oplus j_2)}}_{AB_2} \ket{k_1 l_1}_{L_1} \ket{k_2 l_2}_{L_2} \notag \\
& \ket{e^{\alpha_1 \beta_1}_{i_1 j_1 k_1 l_1}}_{E_1} \ket{e^{\alpha_2 \beta_2}_{i_2 j_2 k_2 l_2}}_{E_2}
\end{align}
Following Protocol \ref{protocol.rd}, a $\pz$-measurement of the target pair of the BCNOT, i.e. the subsystem $AB_2$, is applied to (\ref{eqn.proofrec2}). Next Alice and Bob communicate their respective measurement outcomes over a classic authentic channel. If the measurement outcomes coincide, Alice and Bob keep the source pair, i.e. subsystem $AB_1$ of step \ref{enu.oxford.bcnot}, else they discard both subsystems $AB_1$ and $AB_2$. We assume that both measurements yield the outcome $1$. If both measurement outcomes yield $0$, no phase factor $(-1)^{i_2}$ would be required in the expression (\ref{eqn.proofrec3}). The coinciding measurement outcomes imply $i_1 \oplus j_1 \oplus i_2 \oplus j_2 = 0$. To summarize, the state post-selected on the measurement outcomes $1$ of Alice and Bob is
\begin{align}\label{eqn.proofrec3}
& \sum\limits_{i_1, j_1, i_2, j_2} \sum\limits_{k_1, l_1, k_2, l_2} (-1)^{i_2} P^{\alpha_1 \beta_1}_{i_1 j_1 k_1 l_1} P^{\alpha_2 \beta_2}_{i_2 (i_1 \oplus j_1 \oplus i_2) k_2 l_2} \ket{B_{(i_1 \oplus i_2)(i_1 \oplus j_1)}}_{AB_1} \ket{k_1 l_1}_{L_1} \ket{k_2 l_2}_{L_2} \ket{e^{\alpha_1 \beta_1}_{i_1 j_1 k_1 l_1}}_{E_1} \ket{e^{\alpha_2 \beta_2}_{i_2(i_1 \oplus j_1 \oplus i_2) k_2 l_2}}_{E_2}.
\end{align}
\underline{Step 2:} Recall that L stores in its register attached to the pair of Alice and Bob the effective noise introduced. For that purpose we introduce the unitary $U_u$ as well as an ancilla system $L_3$ set to the state $\ket{00}_{L_3}$. Applying $U_u$ to all three registers of L yields $U_u \ket{00}{}\ket{i}{}\ket{j}{} \ket{i'}{}\ket{j'}{} = \ket{u(i,j,i',j')}{}\ket{i}{}\ket{j}{} \ket{i'}{}\ket{j'}{}$ where $u$ is the so called flag update function defined in \cite{bib.aschauer}. The function $u$ returns the effective noise introduced on the source pair of step \ref{enu.oxford.bcnot} of Protocol \ref{protocol.rd}. Applying $U_u$ to (\ref{eqn.proofrec3}) gives
\begin{align*}
\ket{\Psi'_{\alpha_1, \beta_1, \alpha_2, \beta_2}}{} = \sum\limits_{i_1, j_1, i_2, j_2} \sum\limits_{k_1, l_1, k_2, l_2} (-1)^{i_2} P^{\alpha_1 \beta_1}_{i_1 j_1 k_1 l_1} P^{\alpha_2 \beta_2}_{i_2 (i_1 \oplus j_1 \oplus i_2) k_2 l_2} & \ket{B_{(i_1 \oplus i_2)(i_1 \oplus j_1)}}_{AB_1} \ket{k_1 l_1}_{L_1} \ket{k_2 l_2}_{L_2} \ket{u(k_1, l_1, k_2, l_2)}_{L_3} \\ \notag
& \otimes \ket{e^{\alpha_1 \beta_1}_{i_1 j_1 k_1 l_1}}_{E_1} \ket{e^{\alpha_2 \beta_2}_{i_2(i_1 \oplus j_1 \oplus i_2) k_2 l_2}}_{E_2}.
\end{align*}
We remind the reader that $\ket{\Psi'_{\alpha_1, \beta_1, \alpha_2, \beta_2}}{}$ is the state after the application of i) the noise map $U^{(a_1)}_{\alpha_1, \beta_1} \otimes  U^{(a_2)}_{\alpha_2, \beta_2}$, ii) a basic distillation step, and iii) the update of L's noise register by $U_u$.

\underline{Step 3:} Since the noise transcripts - by assumption for this analysis - leak to Eve, we attribute the systems $L_1$ and $L_2$ to Eve. In order to treat the most general situation, we assume that Eve holds a purification of $\ptr{L_1,L_2,E_1,E_2}{\rho'}$. We determine this purification by computing $\rho'_{1} = \ptr{L_1,L_2}{\rho'}$ and $\rho'_{2} = \ptr{E_1,E_2}{\rho'_{1}}$ and attribute the purification of $\rho'_{2}$ to Eve.

By the linearity of the partial trace we have
\begin{align*}
\rho'_{1} &= \ptr{L_1,L_2}{\rho'} = \sum\limits_{\alpha_1, \beta_1, \alpha_2, \beta_2}  \tilde{f}_{\alpha_1, \beta_1, \alpha_2, \beta_2} \ptr{L_1,L_2}{\ketbra{\Psi'_{\alpha_1, \beta_1, \alpha_2, \beta_2}}{\Psi'_{\alpha_1, \beta_1, \alpha_2, \beta_2}}{}}.
\end{align*}
It is useful to define $\rho'_{\alpha_1, \beta_1, \alpha_2, \beta_2} = \ptr{L_1,L_2}{\ketbra{\Psi'_{\alpha_1, \beta_1, \alpha_2, \beta_2}}{\Psi'_{\alpha_1, \beta_1, \alpha_2, \beta_2}}{}}$ which evaluates to
\begin{align*}
\rho'_{\alpha_1, \beta_1, \alpha_2, \beta_2} &= \ptr{L_1,L_2}{\ketbra{\Psi'_{\alpha_1, \beta_1, \alpha_2, \beta_2}}{\Psi'_{\alpha_1, \beta_1, \alpha_2, \beta_2}}{}} \\
& = \sum (-1)^{i_2 \oplus i'_2} P^{\alpha_1 \beta_1}_{i_1 j_1 k_1 l_1} P^{\alpha_2 \beta_2}_{i_2 (i_1 \oplus j_1 \oplus i_2) k_2 l_2} (P^{\alpha_1 \beta_1}_{i'_1 j'_1 k_1 l_1} P^{\alpha_2 \beta_2}_{i'_2 (i'_1 \oplus j'_1 \oplus i'_2) k_2 l_2})^{*} \ketbra{B_{(i_1 \oplus i_2)(i_1 \oplus j_1)}}{B_{(i'_1 \oplus i'_2)(i'_1 \oplus j'_1)}}{} \\
& \hspace*{0.4cm} \otimes \ketbra{u(k_1, l_1, k_2, l_2)}{u(k_1, l_1, k_2, l_2)}{} \otimes \ketbra{e^{\alpha_1 \beta_1}_{i_1 j_1 k_1 l_1}}{e^{\alpha_1 \beta_1}_{i'_1 j'_1 k_1 l_1}}{} \otimes \ketbra{e^{\alpha_2 \beta_2}_{i_2 (i_1 \oplus j_1 \oplus i_2) k_2 l_2}}{e^{\alpha_2 \beta_2}_{i'_2 (i'_1 \oplus j'_1 \oplus i'_2) k_2 l_2}}{}.
\end{align*}
In the previous expression we neglected the indices appearing in the sum for simplicity, but it is understood that the sum ranges over all indices except $\alpha_1, \beta_1, \alpha_2$ and $\beta_2$.

In order to determine the state of Alice, Bob, and L which Eve finally purifies we have to compute $\rho'_{2} = \ptr{E_1,E_2}{\rho'_{1}}$. Again, the linearity of the partial trace yields
\begin{align}\label{eqn.proofrec6}
\rho'_{2} = \ptr{E_1,E_2}{\rho'_{1}} = \sum\limits_{\alpha_1, \beta_1, \alpha_2, \beta_2}  \tilde{f}_{\alpha_1, \beta_1, \alpha_2, \beta_2} \ptr{E_1,E_2}{\rho'_{\alpha_1, \beta_1, \alpha_2, \beta_2}}.
\end{align}
We remind the reader that $\ket{e^{\alpha \beta}_{i j k l}}_{E_1} = \ket{(i \oplus \alpha)(j \oplus \beta)(k \oplus \alpha)(l \oplus \beta)}_{E_1}$. Hence, for fixed $\alpha_1$ and $\beta_1$, we have $\tr{\ketbra{e^{\alpha_1 \beta_1}_{i_1 j_1 k_1 l_1}}{e^{\alpha_1 \beta_1}_{i'_1 j'_1 k_1 l_1}}{}} = \delta_{i_1  i'_1} \delta_{j_1  j'_1}$, which implies that $i'_1 = i_1$ and $j'_1 = j_1$. Thus, we  also have
\begin{align*}
\tr{\ketbra{e^{\alpha_2 \beta_2}_{i_2 (i_1 \oplus j_1 \oplus i_2) k_2 l_2}}{e^{\alpha_2 \beta_2}_{i'_2 (i'_1 \oplus j'_1 \oplus i'_2) k_2 l_2}}{}} = \tr{\ketbra{e^{\alpha_2 \beta_2}_{i_2 (i_1 \oplus j_1 \oplus i_2) k_2 l_2}}{e^{\alpha_2 \beta_2}_{i'_2 (i_1 \oplus j_1 \oplus i'_2) k_2 l_2}}{}} = \delta_{i_2 i'_2}.
\end{align*}
Hence
\begin{align}\label{eqn.proofrec5}
\ptr{E_1,E_2}{\rho'_{\alpha_1, \beta_1, \alpha_2, \beta_2}} &= \notag \\[3ex]
&=  \sum\limits_{i_1, i_2,j_1} \sum\limits_{k_1, l_1, k_2, l_2} P^{\alpha_1 \beta_1}_{i_1 j_1 k_1 l_1} P^{\alpha_2 \beta_2}_{i_2 (i_1 \oplus j_1 \oplus i_2) k_2 l_2} (P^{\alpha_1 \beta_1}_{i_1 j_1 k_1 l_1} P^{\alpha_2 \beta_2}_{i_2 (i_1 \oplus j_1 \oplus i_2) k_2 l_2})^{*} \notag \\
& \hspace*{0.4cm} \ketbra{B_{(i_1 \oplus i_2) (i_1 \oplus j_1)}}{B_{(i_1 \oplus i_2) (i_1 \oplus j_1)}}{} \otimes \ketbra{u(k_1, l_1, k_2, l_2)}{u(k_1, l_1, k_2, l_2)}{} \notag \\[3ex]
&=  \sum\limits_{i_1, i_2,j_1} \sum\limits_{k_1, l_1, k_2, l_2} \left|P^{\alpha_1 \beta_1}_{i_1 j_1 k_1 l_1} P^{\alpha_2 \beta_2}_{i_2 (i_1 \oplus j_1 \oplus i_2) k_2 l_2} \right|^2 \notag \\
& \hspace*{0.4cm} \ketbra{B_{(i_1 \oplus i_2)(i_1 \oplus j_1)}}{B_{(i_1 \oplus i_2)(i_1 \oplus j_1)}}{} \otimes \ketbra{u(k_1, l_1, k_2, l_2)}{u(k_1, l_1, k_2, l_2)}{}.
\end{align}
By inserting (\ref{eqn.proofrec5}) in (\ref{eqn.proofrec6}) we get
\begin{align*}
\rho'_{2} & = \ptr{E_1,E_2}{\rho'_{1}} \\[3ex]
&= \sum\limits_{\alpha_1, \beta_1, \alpha_2, \beta_2}  \tilde{f}_{\alpha_1, \beta_1, \alpha_2, \beta_2} \ptr{E_1,E_2}{\rho'_{\alpha_1, \beta_1, \alpha_2, \beta_2}} \\[3ex]
& = \sum\limits_{\alpha_1, \beta_1, \alpha_2, \beta_2}  \tilde{f}_{\alpha_1, \beta_1, \alpha_2, \beta_2} \sum\limits_{i_1, i_2,j_1} \sum\limits_{k_1, l_1, k_2, l_2} \left|P^{\alpha_1 \beta_1}_{i_1 j_1 k_1 l_1} P^{\alpha_2 \beta_2}_{i_2 (i_1 \oplus j_1 \oplus i_2) k_2 l_2} \right|^2 \\
& \hspace*{0.4cm} \ketbra{B_{(i_1 \oplus i_2)(i_1 \oplus j_1)}}{B_{(i_1 \oplus i_2)(i_1 \oplus j_1)}}{} \otimes \ketbra{u(k_1, l_1, k_2, l_2)}{u(k_1, l_1, k_2, l_2)}{} \\[3ex]
& = \sum_{i_1, i_2, j_1}\ketbra{B_{(i_1 \oplus i_2)(i_1 \oplus j_1)}}{B_{(i_1 \oplus i_2)(i_1 \oplus j_1)}}{} \otimes \sum_{\gamma_0, \gamma_1} \left(\sum_{\stackrel{\alpha_1, \beta_1, \alpha_2, \beta_2, k_1, l_1, k_2, l_2}{u(k_1, l_1, k_2, l_2) = (\gamma_0, \gamma_1)}} \tilde{f}_{\alpha_1, \beta_1, \alpha_2, \beta_2} \left|P^{\alpha_1 \beta_1}_{i_1 j_1 k_1 l_1} P^{\alpha_2 \beta_2}_{i_2 (i_1 \oplus j_1 \oplus i_2) k_2 l_2}\right|^2 \right) \\
& \hspace*{0.4cm} \ketbra{\gamma_0 \gamma_1}{\gamma_0 \gamma_1}{}.
\end{align*}
Rearranging the sum over $i_1, i_2$ and $j_1$ in the previous equation gives
\begin{align}\label{eqn.proofrec7}
\sum_{\delta_0, \delta_1}\ketbra{B_{\delta_0 \delta_1}}{B_{\delta_0 \delta_1}}{} \otimes \sum_{\gamma_0, \gamma_1} & \left(\sum_{\stackrel{i_1, i_2, j_1}{i_1 \oplus i_2 = \delta_0, i_1 \oplus j_1 = \delta_1}}\sum_{\stackrel{\alpha_1, \beta_1, \alpha_2, \beta_2, k_1, l_1, k_2, l_2}{u(k_1, l_1, k_2, l_2) = (\gamma_0, \gamma_1)}} \tilde{f}_{\alpha_1, \beta_1, \alpha_2, \beta_2} \left|P^{\alpha_1 \beta_1}_{i_1 j_1 k_1 l_1} P^{\alpha_2 \beta_2}_{i_2 (i_1 \oplus j_1 \oplus i_2) k_2 l_2}\right|^2 \right) \\ \notag
& \ketbra{\gamma_0 \gamma_1}{\gamma_0 \gamma_1}{}.
\end{align}
Using the definition
\begin{align}\label{eq.proof.recurrencerelation}
|\tilde{P}_{\delta_0  \delta_1  \gamma_0  \gamma_1}|^2 = \sum_{\stackrel{i_1, i_2, j_1}{i_1 \oplus i_2 = \delta_0, i_1 \oplus j_1 = \delta_1}}\sum_{\stackrel{\alpha_1, \beta_1, \alpha_2, \beta_2, k_1, l_1, k_2, l_2}{u(k_1, l_1, k_2, l_2) = (\gamma_0, \gamma_1)}} \tilde{f}_{\alpha_1, \beta_1, \alpha_2, \beta_2} \left| P^{\alpha_1 \beta_1}_{i_1 j_1 k_1 l_1} P^{\alpha_2 \beta_2}_{i_2 (i_1 \oplus j_1 \oplus i_2) k_2 l_2}\right|^2
\end{align}
where $\delta_0, \delta_1, \gamma_0, \gamma_1 \in \lbrace 0,1 \rbrace$ and omitting the normalization factor for clarity, (\ref{eqn.proofrec7}) simplifies to
\begin{align*}
\sum_{\delta_0, \delta_1}\ketbra{B_{\delta_0 \delta_1}}{B_{\delta_0 \delta_1}}{} \otimes \sum_{\gamma_0, \gamma_1} |\tilde{P}_{\delta_0  \delta_1  \gamma_0  \gamma_1}|^2 \ketbra{\gamma_0 \gamma_1}{\gamma_0 \gamma_1}{}
\end{align*}
which is the state of Alice, Bob, and L after one noisy distillation step. Since this final state is purified by Eve with the leaked noise transcripts and all purifications are unitarily equivalent, the state of Alice, Bob, L, and Eve after one noisy distillation step can be written without loss of generality as
\begin{align*}
\ket{\psi^{DEJMPS}} = \sum\limits_{\delta_0, \delta_1, \gamma_0, \gamma_1} \tilde{P}_{\delta_0, \delta_1, \gamma_0, \gamma_1} \ket{B_{\delta_0,\delta_1}}_{AB} \ket{\gamma_0 \gamma_1}_{L} \ket{\delta_0 \delta_1 \gamma_0 \gamma_1}_{E}.
\end{align*}
This also implies that (\ref{eq.proof.recurrencerelation}) are the recurrence relations of the noisy DEJMPS protocol.

\vspace*{1cm}
\underline{Fixed point and convergence - Binary pairs}
\vspace*{1cm}

First we study the scaling of the systems of Alice, Bob, and L and extend those results then to the (possibly leaked) noise transcripts of Eve in terms of purifications.

Suppose that the initial i.i.d. pairs of Alice and Bob are mixtures of $\ket{B_{00}}$ and $\ket{B_{01}}$ and that L applies either the identity or a $\px$-operator with respective probabilities $\tilde{f}_0$ and  $\tilde{f}_1 = 1-\tilde{f}_0$ independently to each pair. We remind the reader that Eve purifies the state of Alice, Bob, and L with the leaked noise transcripts, i.e. each individual state taking Eve into account at an intermediate round of the DEJMPS protocol reads as $\sum_{i,j} P_{ij} \ket{B_{0i}}_{AB} \otimes \ket{\eta_j}_{L} \otimes \ket{\eta_{ij}}_{E}$. Using $p_{ij} = |P_{ij}|^2$, the recurrence relations (\ref{eq.proof.recurrencerelation}) for the setting we are concerned with here simplify to
\begin{align}\label{eq.recurrencerelation.bin1}
\tilde{p}_{00} & = 1/N(\tilde{f}^2_0 \left( p^2_{00} + 2 p_{00} p_{01}  \right) + \tilde{f}^2_1 \left(p^2_{11} + 2 p_{10} p_{11} \right) + 2 \tilde{f}_0 \tilde{f}_1 \left(p_{11} p_{00} + p_{10} p_{00} + p_{11} p_{01} \right)), \\
\tilde{p}_{01} & = 1/N(\tilde{f}^2_0 p^2_{01} + 2 \tilde{f}_0 \tilde{f}_1 p_{10} p_{01} + \tilde{f}^2_1 p^2_{10}), \\
\tilde{p}_{10} & = 1/N(\tilde{f}^2_0 \left( p^2_{10} + 2 p_{10} p_{11}  \right) + \tilde{f}^2_1 \left(p^2_{01} + 2 p_{00} p_{01} \right) + 2 \tilde{f}_0 \tilde{f}_1 \left(p_{01} p_{10} + p_{00} p_{10} + p_{01} p_{11} \right)), \\ \label{eq.recurrencerelation.bin4}
\tilde{p}_{11} & = 1/N(\tilde{f}^2_0 p^2_{11} + 2 \tilde{f}_0 \tilde{f}_1 p_{00} p_{11} + \tilde{f}^2_1 p^2_{00}).
\end{align}
where $N = (\tilde{f}^2_0 + \tilde{f}^2_1)((p_{00} + p_{01})^2 + (p_{10} + p_{11})^2) + 4 \tilde{f}_0 \tilde{f}_1 (p_{00}+p_{01})(p_{10}+p_{11})$. In the following we denote the recurrence relations (\ref{eq.recurrencerelation.bin1})--(\ref{eq.recurrencerelation.bin4}) by the vector-valued mapping $\mathbf{f}$, i.e. $\mathbf{p} \stackrel{\mathbf{f}}{\to} \mathbf{\tilde{p}},$ where $\mathbf{p}=(p_{00}, p_{01}, p_{10}, p_{11})$. A simple computation yields the following fixed points of $\mathbf{f}$:
\begin{align}\label{eq.fixedpoint.bin}
& p^\infty_{00} = 1/2 + \sqrt{4 \tilde{f}_0 - 3}/(4\tilde{f}_0 -2) \quad p^\infty_{01} = p^\infty_{10} = 0 \quad p^\infty_{11} = 1-p^\infty_{00}, \\
& p^\infty_{00} = 1/2 - \sqrt{4 \tilde{f}_0 - 3}/(4\tilde{f}_0 -2) \quad p^\infty_{01} = p^\infty_{10} = 0 \quad p^\infty_{11} = 1-p^\infty_{00}, \\
& p^\infty_{00}= p^\infty_{11} = 1/2 \quad p^\infty_{01} = p^\infty_{10} = 0.
\end{align}
The parameter estimation phase guarantees that the fidelity $F$ with $\ket{B_{00}}$ is sufficiently high for distillation. Hence the fixed point of interested is (\ref{eq.fixedpoint.bin}), i.e.
\begin{align}\label{eqn.bin.fixedpoints}
\mathbf{p}^\infty = (1/2 + \sqrt{4 \tilde{f}_0 - 3}/(4\tilde{f}_0 -2), 0, 0, 1/2 - \sqrt{4 \tilde{f}_0 - 3}/(4\tilde{f}_0 -2)).
\end{align}
From (\ref{eqn.bin.fixedpoints}) we observe that in the limit the `cross-probabilities' $p_{01}$ and $p_{10}$, vanish, hence $L$ is fully correlated to $AB$.

It is of central importance, regarding convergence that the fixed point $\mathbf{p}^\infty$ is an attractor, as only this ensures convergence towards that fixed point. Note that $\mathbf{p}^\infty$ is an attractor if and only if the largest eigenvalue $\lambda_{max}$ of $\mathbf{f'(p^\infty)}$ satisfies $\lambda_{max} < 1$. We easily find that $\lambda_{max} = (\tilde{f}_0 \sqrt{4 \tilde{f}_0 - 3} - \tilde{f}_0)/(2 \tilde{f}_0 - 1) < 1$ for $0.78 \leq \tilde{f}_0 \leq 1$. 

The fixed point $\mathbf{p}^\infty$ enables us to determine the rate of convergence. For that purpose, we expand $\mathbf{f}$ in terms of its Taylor series around the fixed point $\mathbf{p}^\infty$, i.e. $\mathbf{\tilde{p}} = \mathbf{f}(\mathbf{p}) \approx \mathbf{f}(\mathbf{p}^\infty) + \mathbf{f'}(\mathbf{p}^\infty)(\mathbf{p} - \mathbf{p}^\infty).$ Hence by defining $\mathbf{e} = \mathbf{p} - \mathbf{p^\infty}$ we find $\mathbf{\tilde{e}} = \mathbf{f'(p^\infty)} \mathbf{e}$, providing an estimate of the error propagation for one successful distillation round. The state of Alice, Bob, and L after $n$ successful distillation rounds and at the fixpoint read as $\rho_n = \sum_{ij} p^{(n)}_{ij} \dm{B_{0i}}_{AB} \otimes \dm{\eta_j}_{L}$ and $\rho_{\text{fix}} = \sum_{i} p^\infty_{ii} \dm{B_{0i}}_{AB} \otimes \dm{\eta_i}_{L}$ respectively, which implies for their distance induced by the $1$-norm
\begin{align}\label{eq:scaling:binary:local}
\epsilon_{n} = \| \rho_n - \rho_{\text{fix}} \|_{1} = \left\| \sum_{i,j} (p^{(n)}_{ij} - p^\infty_{ij}) \dm{B_{0i}}_{AB} \otimes \dm{\eta_j}_{L} \right\|_{1} = \underbrace{\sum_{i,j} |p^{(n)}_{ij} - p^\infty_{ij}|}_{\|\mathbf{e_{n}} \|_{1;\text{v}}} \leq \|\mathbf{f'}(\mathbf{p^\infty})^{n-1} \| \| \mathbf{e_{1}}\|_{1;\text{v}}.
\end{align}
where $\| \mathbf{x} \|_{1;\text{v}} = \sum^k_{i=1} |x_i|$ denotes the $1$-norm of vectors in $\C^k$.

Eq. (\ref{eq:scaling:binary:local}) only concerns the systems of Alice, Bob, and L. To complete the analysis we recall that Eve purifies $\rho_{n}$ and $\rho_{\text{fix}}$ with the leaked noise transcripts of L. If we take this purifying system, $E$, into account, i.e. consider $\|\dm{\psi^{n}}_{ABEL} - \dm{\psi^\alpha}_{ABEL}\|_{1}$ where $\rho_n = \ptr{E}{\dm{\psi^{n}}_{ABEL}}$, $\ket{\psi^\alpha}_{ABEL} = \sum_{i,j} P^{\infty}_{ij} \ket{B_{0i}}_{AB} \otimes \ket{\eta_j}_{L} \otimes \ket{\eta_{ij}}_{E}$ with $|P^{\infty}_{ij}|^2 = p^\infty_{ij}$ and $\rho_{\text{fix}} = \ptr{E}{\dm{\psi^\alpha}_{ABEL}}$, we find
\begin{align}
\|\dm{\psi^{n}}_{ABEL} - \dm{\psi^\alpha}_{ABEL}\|_{1} \leq \sqrt{\epsilon_{n}} \label{eq:scaling:bin:1}
\end{align}
since purifications scale with a square root.

In order to apply the post-selection-based reduction, we need to relate the previously obtained results for i.i.d. input pairs to general ensembles. As stated in the main text, we exclude the parameter estimation step on $\sqrt{n}$ initial states for simplicity. We remind the reader, as we have stated in the main text, that for all purifications $\ket{\psi}_{ABE'}$ of a $n$-partite input state $\rho_{AB}$ we have
\begin{align}\label{eq:postselection:bin}
\| (\E \otimes id_{E'})(\ketbra{\psi}{\psi}{ABE'}) &- (\F \otimes id_{E'})(\ketbra{\psi}{\psi}{ABE'})\|_1 \leq 4 g_{n,d} \sqrt{\max_{\sigma_{AB}} \left\|(\E_{\text{L}}-\F_{\text{L}})\left(\sigma^{\otimes n}_{AB} \right) \right\|_1}
\end{align}
where $g_{n,d} = {n + d^2 -1 \choose n}$. Thus, inserting the previous result for $2^n$ i.i.d. input states (necessary to achieve $n$ rounds of distillation) in (\ref{eq:postselection:bin}) yields
\begin{align*}
\|(\E \otimes id_{E'})(\ketbra{\psi}{\psi}{ABE'}) &- (\F \otimes id_{E'})(\ketbra{\psi}{\psi}{ABE'})\|_1 \leq 4 g_{2^n,d} \epsilon_{n}^{1/4}.
\end{align*}
One square root in the expression above arises from inequality (\ref{eq:scaling:bin:1}) and the other square root appears from inequality $(\ref{eq:postselection:bin})$.

Hence, for confidentiality we necessarily need $g_{2^n,d} \epsilon_{n}^{1/4} \to 0$ for $n \to \infty$. Thus $\epsilon_{n}^{1/4}$ should decay faster than $g_{2^n,d}$ grows in $n$. Numerical simulations suggest that, for $\tilde{f}_0 = 1-10^{-19}$, this turns out to be true, i.e. the post-selection-based reduction is applicable (see Figure \ref{fig:binpairspostselect}). As stated in the main text such rates are unlikely to be achievable on the physical level, but they are, at least in principle, possible through fault-tolerant constructions.

\begin{figure}[htb]
\centering
\hspace{0pt plus 1filll}\null

$n$

\medskip
\hspace*{1.5em}\raisebox{\dimexpr-.5\height-1em}
  {\includegraphics[scale=0.8]{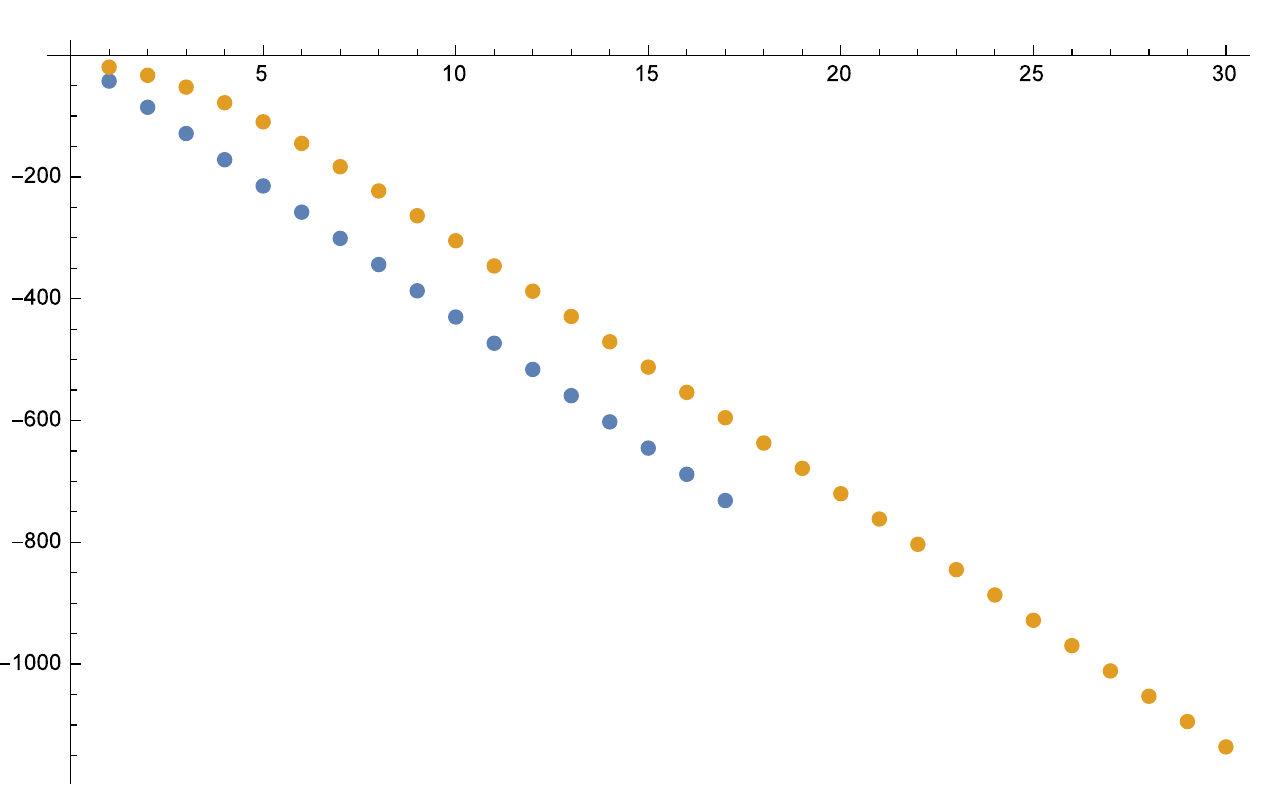}}%

\leavevmode\smash{\makebox[0pt]{\hspace{-35em}
  \rotatebox[origin=l]{90}{\hspace{4em}
    $\log \| \mathbf{f'}(\mathbf{p^\infty})^{n}\|, -4 \log g_{2^n,d}$}%
}}

\caption[h!]{\label{fig:binpairspostselect} The figure illustrates $\log \| \mathbf{f'}(\mathbf{p^\infty})^{n}\|$ (blue) and $-4 \log g_{2^n,d}$ (yellow) for the binary pairs and $\tilde{f}_0 = 1-10^{-19}$.}
\end{figure}

\vspace*{1cm}
\underline{Fixed point and convergence - General pairs}
\vspace*{1cm}

In the following we show that the previous established results also hold true for the general i.i.d. setting where L applies all four Pauli operators and each individual pair is arbitrary. We remind the reader that the recurrence relations for states $\sum_{i,j,k,l} P_{ijkl} \ket{B_{ij}}_{AB} \otimes \ket{\eta_{kl}}_{L} \otimes \ket{\eta_{ijkl}}_{E}$ (i.e. Eve purifies $\rho_n = \sum_{i,j,k,l} |P_{ijkl}|^2 \ketbra{B_{ij}}{B_{ij}}{AB} \otimes \ketbra{\eta_{kl}}{\eta_{kl}}{L}$ with the leaked noise transcripts) read (by denoting $|P_{ijkl}|^2 = p_{ijkl}$) as
\begin{align*}
\tilde{p}_{\delta_0  \delta_1  \gamma_0  \gamma_1} = \sum_{\stackrel{i_1, i_2, j_1}{i_1 \oplus i_2 = \delta_0, i_1 \oplus j_1 = \delta_1}}\sum_{\stackrel{\alpha_1, \beta_1, \alpha_2, \beta_2, k_1, l_1, k_2, l_2}{u(k_1, l_1, k_2, l_2) = (\gamma_0, \gamma_1)}} \tilde{f}_{\alpha_1, \beta_1, \alpha_2, \beta_2} p_{(i_1 \oplus \alpha_1) (j_1\oplus \beta_1) (k_1 \oplus \alpha_1) (l_1 \oplus \beta_1)} p_{(i_2 \oplus \alpha_2) (i_1 \oplus j_1 \oplus i_2 \oplus \beta_2) (k_2 \oplus \alpha_2) (l_2 \oplus \beta_2)}
\end{align*}
modulo the normalization factor $\sum_{\delta_0  \delta_1  \gamma_0  \gamma_1} \tilde{p}_{\delta_0  \delta_1  \gamma_0  \gamma_1}$.

For simplicity we assume independent single qubit white noise, i.e. $\tilde{f}_{\alpha_1, \beta_1, \alpha_2, \beta_2} = \tilde{f}_{\alpha_1, \beta_1} \tilde{f}_{\alpha_2, \beta_2}$ as well as $\tilde{f}_{\alpha_1, \beta_1} = f$ if $\alpha_1=\beta_1=0$ and $(1-f)/3$ otherwise. Furthermore, we assume that the initial fidelity $F$ with $\ket{B_{00}}$ is sufficiently high for distillation. Numerically iterating the recurrence relations (which we again denote by $\mathbf{p} \stackrel{\mathbf{f}}{\to} \mathbf{\tilde{p}}$) reveal that, for a sufficiently large number of iterations, the `cross-probabilities' vanish, i.e. $p^{\infty}_{ijkl} = 0 \Leftrightarrow i \neq k$ or $j \neq l$. Hence, to obtain a fixed point $\mathbf{p^{\infty}} = (p^{\infty}_{ijkl})^1_{i,j,k,l = 0}$ of $\mathbf{f}$, it is reasonable to assume that $p^{\infty}_{ijkl} =0 \Leftrightarrow i \neq k$ or $j \neq l$.

Thus the fixed point $\mathbf{p^{\infty}}$ is determined by four equations in four unknowns, namely the equations

\begin{align*}
p_{\delta_0  \delta_1  \delta_0  \delta_1} = \frac{1}{N} \sum_{\stackrel{i_1, i_2, j_1}{i_1 \oplus i_2 = \delta_0, i_1 \oplus j_1 = \delta_1}}\sum_{\stackrel{\alpha_1, \beta_1, \alpha_2, \beta_2}{u(i_1, j_1, i_2, i_1 \oplus j_1 \oplus i_2) = (\delta_0, \delta_1)}} \tilde{f}_{\alpha_1, \beta_1} \tilde{f}_{\alpha_2, \beta_2} & p_{(i_1 \oplus \alpha_1) (j_1 \oplus \beta_1) (i_1 \oplus \alpha_1) (j_1 \oplus \beta_1)} \\
& \cdot p_{(i_2 \oplus \alpha_2) (i_1 \oplus j_1 \oplus i_2 \oplus \beta_2) (i_2 \oplus \alpha_2) (i_1 \oplus j_1 \oplus i_2 \oplus \beta_2)}.
\end{align*}

where $\delta_0, \delta_1 \in \lbrace 0,1 \rbrace$ and $N = \sum_{\delta_0,  \delta_1} p_{\delta_0  \delta_1  \delta_0  \delta_1 }$. Figure \ref{fig:p0000:fixed} illustrates the numerical estimate of $p^{\infty}_{0000}$ as a function of $f$.

\begin{figure}[htb]
\centering
\hspace*{1.5em}\raisebox{\dimexpr-.5\height-1em}
  {\includegraphics{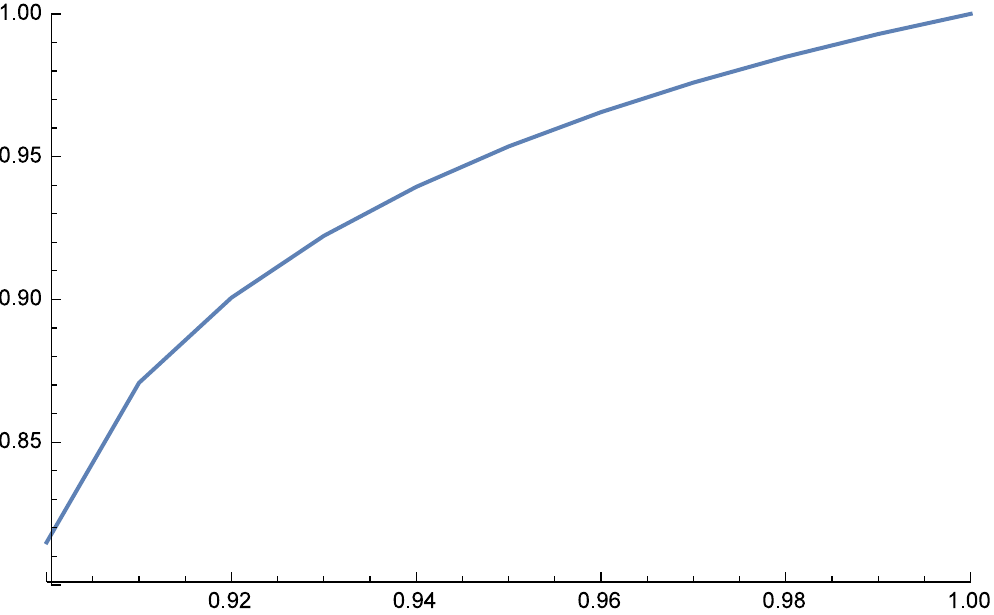}
}%

\leavevmode\smash{\makebox[0pt]{\hspace{-35em}
  \rotatebox[origin=l]{90}{\hspace{10em}
    $p^\infty_{0000}$}%
}}

\hspace{0pt plus 1filll}\null

$f$

\medskip
\caption[h!]{\label{fig:p0000:fixed} The figure illustrates $p^{\infty}_{0000}$ as a function of $f$. The fidelity with $\ket{B_{00}}$ of the asymptotic state is equal to unity for a perfect apparatus.}
\end{figure}

Similar to the case of binary pairs, we can write the recurrence relations $\mathbf{f}$ in terms of its Taylor series expansion around the fixed point $\mathbf{p}^\infty$, i.e. $\mathbf{\tilde{p}} = \mathbf{f}(\mathbf{p}) \approx \mathbf{f}(\mathbf{p}^\infty) + \mathbf{f'}(\mathbf{p}^\infty)(\mathbf{p} - \mathbf{p}^\infty).$ Hence by defining $\mathbf{e} = \mathbf{p} - \mathbf{p^\infty}$ we have $\mathbf{\tilde{e}} = \mathbf{f'(p^\infty)} \mathbf{e}$, i.e. as for binary pairs, the error induced by the $1-$norm of the state of Alice, Bob, and L after $n$ successful distillation rounds satisfies
\begin{align}
\| \rho_n - \rho_{\text{fix}} \|_{1} = \left\| \sum_{i,j,k,l} \left(p^{(n)}_{ijkl} - p^\infty_{ijkl} \right) \dm{B_{ij}}_{AB} \otimes \dm{\eta_{kl}}_{L} \right\|_{1} \leq  \sum_{i,j,k,l} \left|p^{(n)}_{ijkl} - p^\infty_{ijkl} \right| \leq \|\mathbf{f'}(\mathbf{p^\infty})^{n-1} \| \| \mathbf{e_{1}}\|_{1;\text{v}}. \label{eq:scaling:general}
\end{align}

\begin{figure}[htb]
\centering
\hspace{0pt plus 1filll}\null

$n$

\medskip
\hspace*{1.5em}\raisebox{\dimexpr-.5\height-1em}
  {\includegraphics{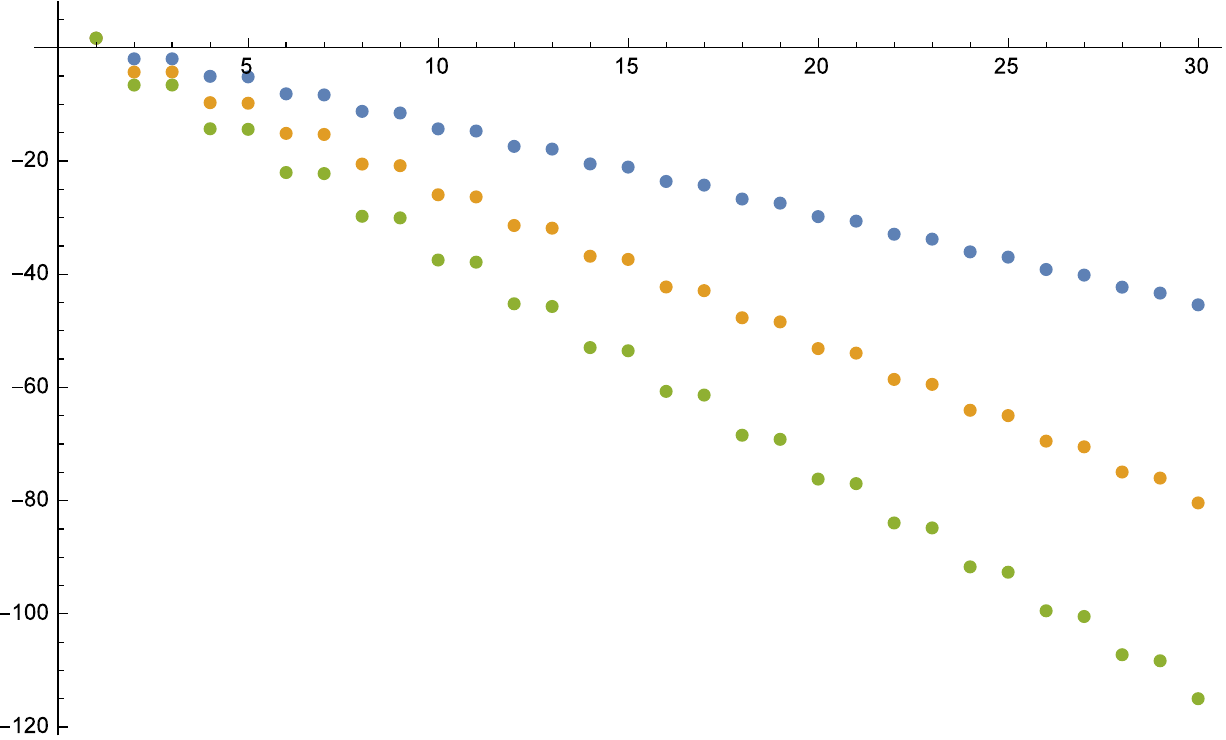}}%

\leavevmode\smash{\makebox[0pt]{\hspace{-40em}
  \rotatebox[origin=l]{90}{\hspace{8em}
    $\log \|\mathbf{f'}(\mathbf{p^\infty})^{n} \|$}%
}}

\caption[h!]{\label{fig:specnormjac} The figure illustrates the value of $\log \|\mathbf{f'}(\mathbf{p^\infty})^{n} \|$ as a function of successful distillation rounds for single qubit white noise $10^{-2}$ (blue), $10^{-3}$ (yellow) and $10^{-4}$ (green).}
\end{figure}

Figure \ref{fig:specnormjac} suggests a linear relationship between the number of successful distillation rounds $n$ and $\log \|\mathbf{f'}(\mathbf{p^\infty})^{n-1} \|$ for each noise level $f$, i.e. $b(f) n + a(f) = \log \|\mathbf{f'}(\mathbf{p^\infty})^{n-1} \|$. As the number $N$ of pairs necessary to achieve $n$ distillation rounds is $N = 2^n$ ($\Leftrightarrow n = \log_2 N$) we have $b(f) \log_2 N + a(f) = \log \|\mathbf{f'}(\mathbf{p^\infty})^{n-1} \|$, which is equivalent to
\begin{align*}
\|\mathbf{f'}(\mathbf{p^\infty})^{n-1} \| = e^{a(f)} e^{b(f) \log_2 N} = a'(f) N^{b'(f)}.
\end{align*}
Hence, $\|\mathbf{f'}(\mathbf{p^\infty})^{n-1} \|$ scales as $F(N) \in O(N^{b'(f)})$ where $b'(f) < 0$ and $b'(f)$ decays for $f \to 1$.

What is left to show, is that the fixed point $\mathbf{p^\infty}$ is an attracting fixed point. For that purpose we numerically compute the largest eigenvalue of $\mathbf{f'}(\mathbf{p^\infty})$, see Fig. \ref{fig:lambdamaxjac}, and observe that, for noise below $10^{-1}$, i.e. $1-f < 10^{-1}$, the largest eigenvalue $\lambda_{max}$ of $\mathbf{f'}(\mathbf{p^\infty})$ fulfills $\lambda_{max} < 1$, proving that $\mathbf{p^\infty}$ is an attracting fixed point.

\begin{figure}[htb]
\centering
\hspace*{1.5em}\raisebox{\dimexpr-.5\height-1em}
  {\includegraphics{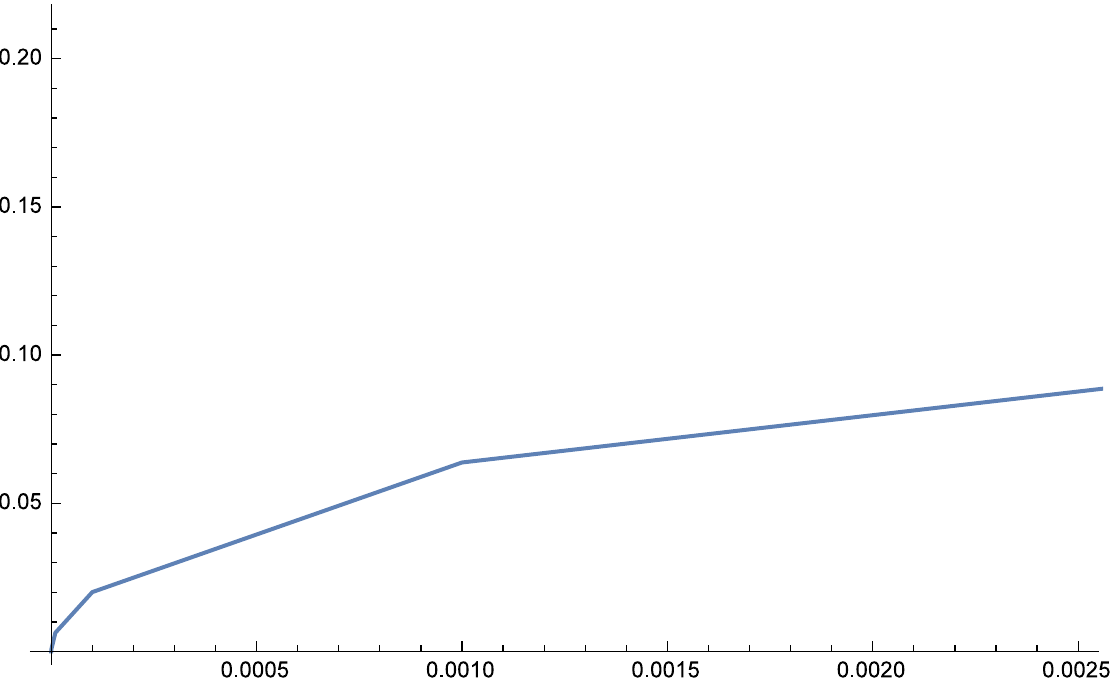}
}%

\leavevmode\smash{\makebox[0pt]{\hspace{-38em}
  \rotatebox[origin=l]{90}{\hspace{10em}
    $\lambda_{max}$}%
}}

\hspace{0pt plus 1filll}\null

$1-f$

\medskip
\caption[h!]{\label{fig:lambdamaxjac} The figure shows the largest eigenvalue of $\mathbf{f'}(\mathbf{p^\infty})$ (y-axis) for single qubit white noise (x-axis)}
\end{figure}

This implies that, if the initial fidelity $F$ with $\ket{B_{00}}$ is sufficiently large for distillation, the DEJMPS protocol necessarily converges towards the fixed point $\mathbf{p^\infty}$ where the `cross-probabilities' vanish.

The analysis so far still lacks Eve's system $E$ for the leaked noise transcripts. Suppose $\ket{\psi^{n}}_{ABEL}$ and $\ket{\psi^f}_{ABEL}$ are purifications of $\rho_n$ and $\rho_{\text{fix}}$, i.e. $\rho_n = \ptr{E}{\dm{\psi^{n}}}$ and $\rho_{\text{fix}} = \ptr{E}{\dm{\psi^f}}$ respectively. This implies $\epsilon_{n} = \|\dm{\psi^{n}} - \dm{\psi^f}\|_{1} \leq \sqrt{F(N)}$, i.e. $\epsilon_{n} \in O(N^{b'(f)/2})$ which we also confirmed with our numeric results.

It is straightforward to extend the analysis above to two-qubit correlated noise introduced by L on the system of Alice and Bob. For that purpose we assume that $\tilde{f}_{\alpha_1, \beta_1, \alpha_2, \beta_2} = \tilde{f} + (1-\tilde{f})/16$ if $\alpha_1=\beta_1=\alpha_2=\beta_2=0$ and $(1-\tilde{f})/16$ otherwise. Also in that case we numerically observe that $p^{\infty}_{ijkl} = 0 \Leftrightarrow i \neq k$ or $j \neq l$. Hence it is reasonable to assume that $p^{\infty}_{ijkl} =0 \Leftrightarrow i \neq k$ or $j \neq l$ in order to obtain a fixed point $\mathbf{p^{\infty}} = (p^{\infty}_{ijkl})^1_{i,j,k,l = 0}$ of $\mathbf{f}$.

The fixed point $\mathbf{p^{\infty}}$ is determined by four equations in four unknowns, namely the equations

\begin{align*}
p_{\delta_0  \delta_1  \delta_0  \delta_1} = \frac{1}{N} \sum_{\stackrel{i_1, i_2, j_1}{i_1 \oplus i_2 = \delta_0, i_1 \oplus j_1 = \delta_1}}\sum_{\stackrel{\alpha_1, \beta_1, \alpha_2, \beta_2}{u(i_1, j_1, i_2, i_1 \oplus j_1 \oplus i_2) = (\delta_0, \delta_1)}} \tilde{f}_{\alpha_1, \beta_1, \alpha_2, \beta_2} & p_{(i_1 \oplus \alpha_1) (j_1 \oplus \beta_1) (i_1 \oplus \alpha_1) (j_1 \oplus \beta_1)} \\
& \cdot p_{(i_2 \oplus \alpha_2) (i_1 \oplus j_1 \oplus i_2 \oplus \beta_2) (i_2 \oplus \alpha_2) (i_1 \oplus j_1 \oplus i_2 \oplus \beta_2)}.
\end{align*}

where $\delta_0, \delta_1 \in \lbrace 0,1 \rbrace$ and $N = \sum_{\delta_0,  \delta_1} p_{\delta_0  \delta_1  \delta_0  \delta_1 }$. Figure \ref{fig:p0000:fixed:2qu} illustrates the numerical estimate of $p^{\infty}_{0000}$ as a function of $\tilde{f}$.

\begin{figure}[htb]
\centering
\hspace*{1.5em}\raisebox{\dimexpr-.5\height-1em}
  {\includegraphics[scale=0.7]{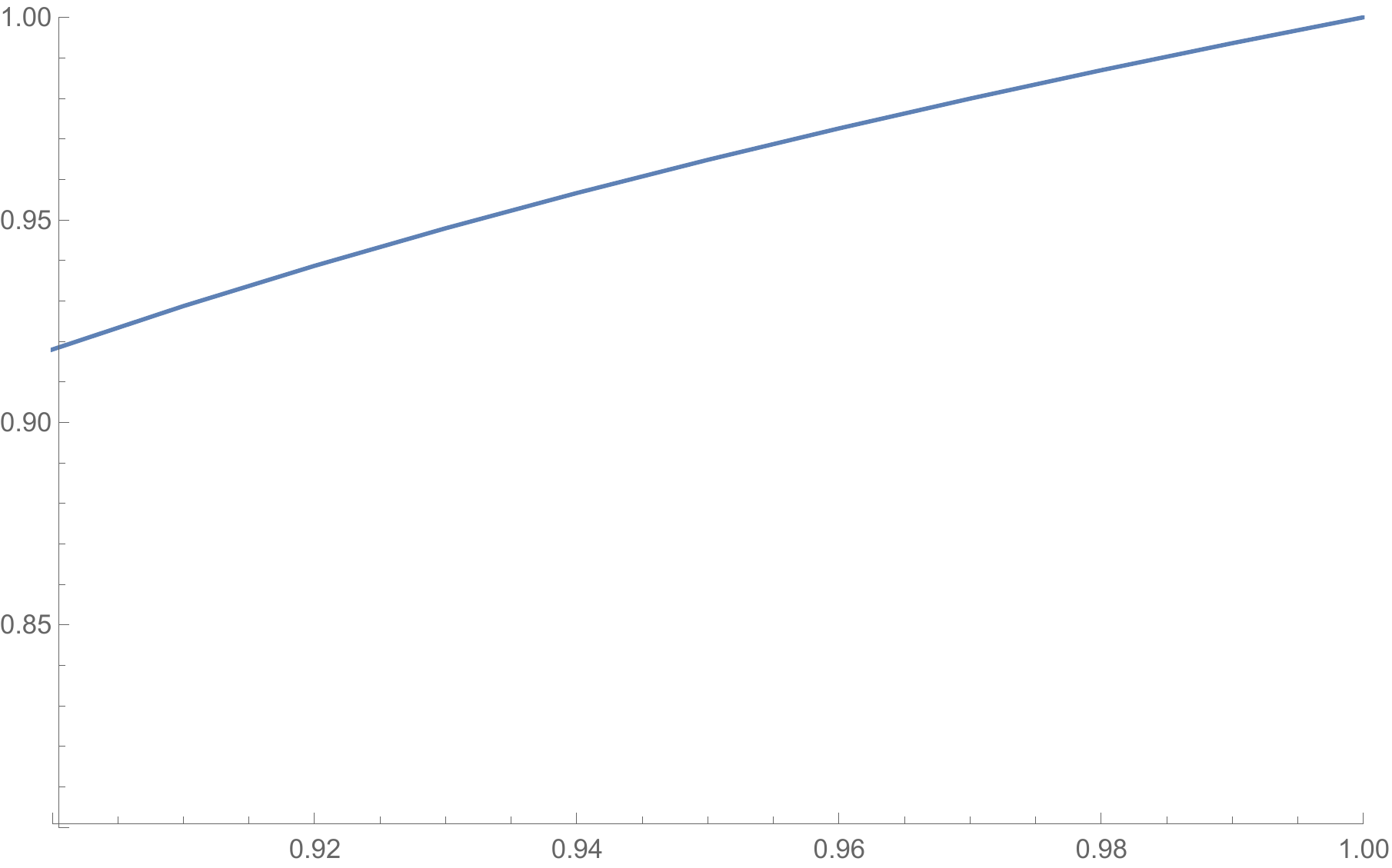}
}%

\leavevmode\smash{\makebox[0pt]{\hspace{-45em}
  \rotatebox[origin=l]{90}{\hspace{15em}
    $p^\infty_{0000}$}%
}}

\hspace{0pt plus 1filll}\null

$\tilde{f}$

\medskip
\caption[h!]{\label{fig:p0000:fixed:2qu} The figure illustrates $p^{\infty}_{0000}$ as a function of $\tilde{f}$ for two qubit correlated noise. The fidelity with $\ket{B_{00}}$ of the asymptotic state is equal to unity for a perfect apparatus.}
\end{figure}

Furthermore we numerically compute the largest eigenvalue of $\mathbf{f'}(\mathbf{p^\infty})$ and observe that if $\tilde{f} > 0.8284$, the largest eigenvalue $\lambda_{max}$ of $\mathbf{f'}(\mathbf{p^\infty})$ fulfills $\lambda_{max} < 1$, hence $\mathbf{p^\infty}$ is an attracting fixed point, see Fig. \ref{fig:lambdamaxjac_2qu}.

\begin{figure}[htb]
\centering
\hspace*{1.5em}\raisebox{\dimexpr-.5\height-1em}
  {\includegraphics[scale=0.7]{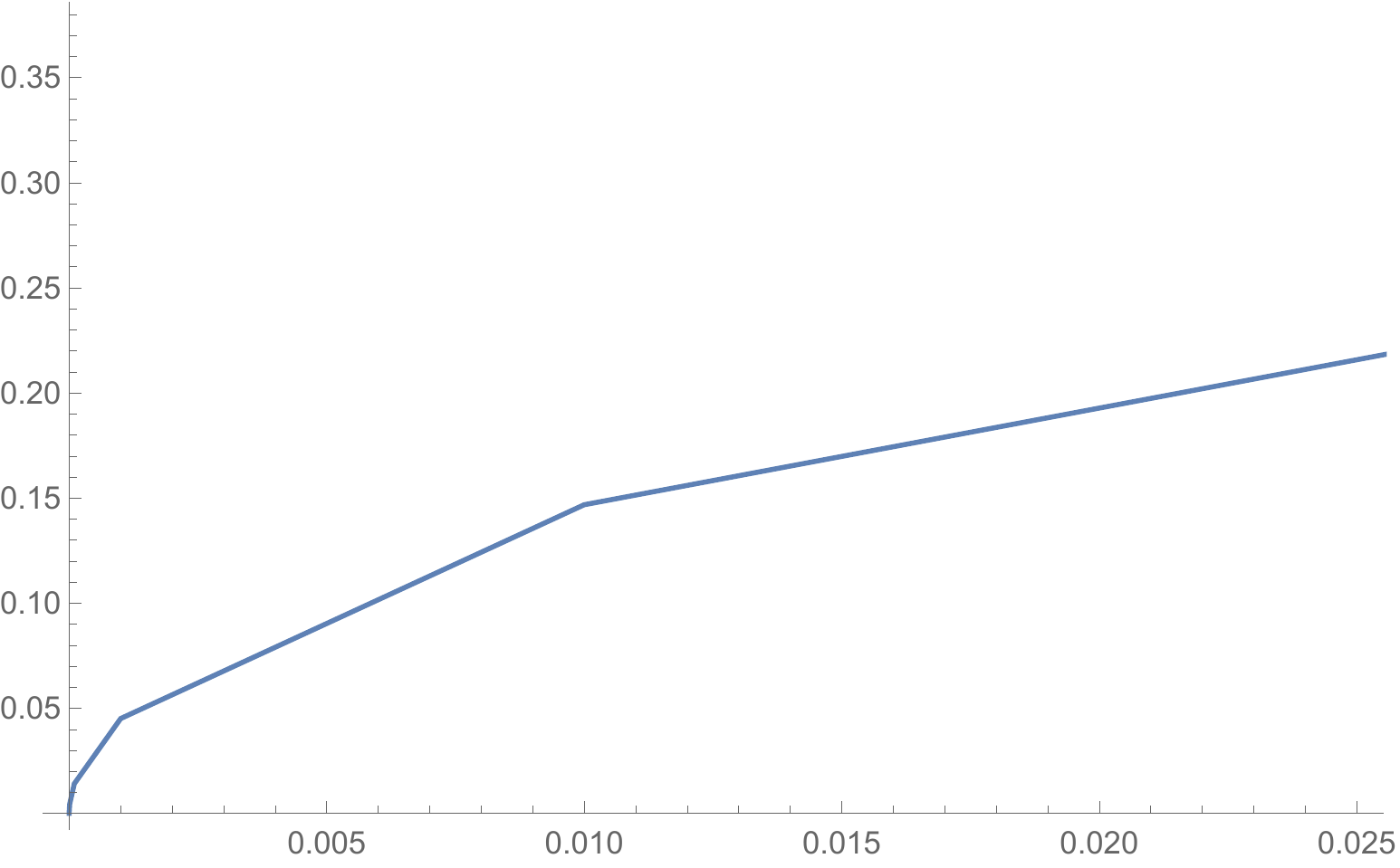}
}%

\leavevmode\smash{\makebox[0pt]{\hspace{-43em}
  \rotatebox[origin=l]{90}{\hspace{10em}
    $\lambda_{max}$}%
}}

\hspace{0pt plus 1filll}\null

$1-\tilde{f}$

\medskip
\caption[h!]{\label{fig:lambdamaxjac_2qu} The figure shows the largest eigenvalue of $\mathbf{f'}(\mathbf{p^\infty})$ (y-axis) for correlated two qubit noise (x-axis)}
\end{figure}

Finally, we obtain again a linear relationship between the number of successful distillation rounds $n$ and $\log \|\mathbf{f'}(\mathbf{p^\infty})^{n-1} \|$ for each noise level $\tilde{f}$, i.e. $b_2(\tilde{f}) n + a_2(\tilde{f}) = \log \|\mathbf{f'}(\mathbf{p^\infty})^{n-1} \|$, see Fig. \ref{fig:specnormjac:2qu}. This implies, similar to the case of single qubit white noise, that the right-hand-side of (\ref{eq:scaling:general}) converges polynomial fast towards zero in terms of initial states. The rate of convergence is governed by $\tilde{f}$, i.e. $\| \rho_n - \rho_{\text{fix}} \|_{1} \leq F_2(N)$ where $F_2(N) \in O(N^{b_2(\tilde{f})})$ and $b_2(\tilde{f}) < 0$ with $b_2(\tilde{f})$ decays for $\tilde{f} \to 1$.

\begin{figure}[htb]
\centering
\hspace{0pt plus 1filll}\null

$n$

\medskip
\hspace*{1.5em}\raisebox{\dimexpr-.5\height-1em}
  {\includegraphics[scale=0.7]{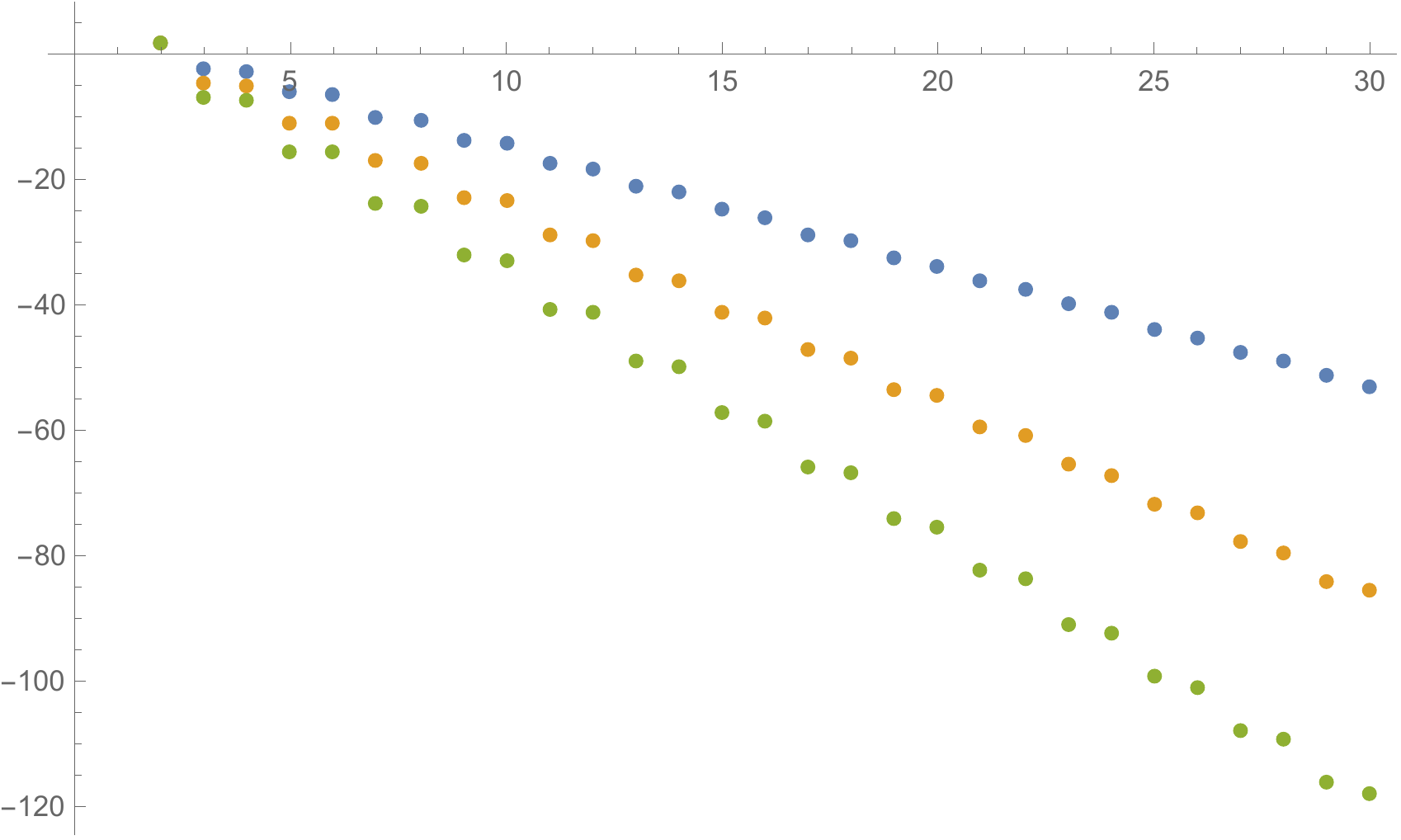}}%

\leavevmode\smash{\makebox[0pt]{\hspace{-40em}
  \rotatebox[origin=l]{90}{\hspace{8em}
    $\log \|\mathbf{f'}(\mathbf{p^\infty})^{n} \|$}%
}}

\caption[h!]{\label{fig:specnormjac:2qu} The figure illustrates the value of $\log \|\mathbf{f'}(\mathbf{p^\infty})^{n} \|$ as a function of successful distillation rounds for two qubit correlated noise $10^{-2}$ (blue), $10^{-3}$ (yellow) and $10^{-4}$ (green).}
\end{figure}

Taking the system of leaking noise transcripts into account,this implies that $\epsilon_{n} = \left\|\dm{\psi^{n}} - \dm{\psi^{\tilde{f}}} \right\|_{1} \leq \sqrt{F_2(N)}$, i.e. $\epsilon_{n} \in O(N^{b_2(\tilde{f})/2})$.

To conclude the analysis, we now show that the noise model of two-qubit depolarizing noise is actually sufficient to cover {\em any} noise process for two-qubit operations. This is the case because for any CNOT-type gate (which we need to apply in the case of both recurrence-type entanglement distillation protocols we consider), one can depolarize these gates to a standard form \cite{DuerStd}. This is done by randomly applying single-qubit operations before and after the application of the gate, which allows one to reduce any noise characteristics to a specific form with 8 parameters without altering the fidelity of the gate. A further simplification is possible if the noise characteristic of the apparatus is known \cite{DuerStd}, which could in some cases be achieved through quantum process tomography. In this case, one can add additional (local) noise by randomly choosing to apply the gate, or some other (separable) operation. This allows one to bring any CNOT-type gate (i.e. any two-qubit gate that is equivalent to a CNOT gate up to single qubit unitary operations that are applied before and after the gate) to the standard form
\begin{align}
\E(\rho) = \tilde{f} U \rho U^{\dagger} + \frac{1-\tilde{f}}{16}\sum\limits^1_{\alpha_1, \beta_1, \alpha_2, \beta_2 = 0} \sigma_{\alpha_1, \beta_1} \sigma_{\alpha_2, \beta_2} \rho \sigma_{\alpha_1, \beta_1} \sigma_{\alpha_2, \beta_2} \label{eq:ox:gen:1}
\end{align}
As outlined in \cite{DuerStd} this depolarization procedure causes a change in the gate fidelity of the utilized quantum gates. More precisely, if the fidelity of the quantum gate before the depolarization was $F_{g} = 1-x$ then the gate fidelity after the depolarization is $F'_{g} > 1- 17x$. Thus one reduces the quality of the gate by about an order of magnitude in the worst case by depolarizing to this standard form. \newline
We observe that (\ref{eq:ox:gen:1}) can be rewritten as
\begin{align}
\E(\rho) &= \tilde{f} U \rho U^{\dagger} + \frac{1-\tilde{f}}{16}\sum\limits^1_{\alpha_1, \beta_1, \alpha_2, \beta_2 = 0} \sigma_{\alpha_1, \beta_1} \sigma_{\alpha_2, \beta_2} \rho \sigma_{\alpha_1, \beta_1} \sigma_{\alpha_2, \beta_2} \notag \\
&= U\left(\tilde{f} \rho + \frac{1-\tilde{f}}{16}\sum\limits^1_{\alpha_1, \beta_1, \alpha_2, \beta_2 = 0} \sigma_{\alpha_1, \beta_1} \sigma_{\alpha_2, \beta_2} \rho \sigma_{\alpha_1, \beta_1} \sigma_{\alpha_2, \beta_2} \right) U^\dagger \notag \\
&= U \left(\sum\limits^1_{\alpha_1, \beta_1, \alpha_2, \beta_2 = 0} \tilde{f}_{\alpha_1, \beta_1, \alpha_2, \beta_2} \sigma_{\alpha_1, \beta_1} \sigma_{\alpha_2, \beta_2} \rho \sigma_{\alpha_1, \beta_1} \sigma_{\alpha_2, \beta_2} \right) U^\dagger \label{eq:ox:gen:2}
\end{align}
where $\tilde{f}_{0,0,0,0} = \tilde{f} + (1-\tilde{f})/16$ and $\tilde{f}_{\alpha_1, \beta_1, \alpha_2, \beta_2} = (1-\tilde{f})/16$ otherwise. Recall, that one noisy distillation step of the DEJMPS protocol including L is given by (\ref{eqn.overall}). By introducing $O_{\mathrm{D}} = U_{u} O'_{\mathrm{2-EPP}}$ we rewrite (\ref{eqn.overall}) as
\begin{align}
\rho' & = \sum\limits_{\alpha_1, \beta_1, \alpha_2, \beta_2}  \tilde{f}_{\alpha_1, \beta_1, \alpha_2, \beta_2} O_{\mathrm{D}} (U^{(a_1)}_{\alpha_1, \beta_1} \otimes  U^{(a_2)}_{\alpha_2, \beta_2}) \left(\ketbra{\Psi}{\Psi}{} \otimes \ketbra{\Psi}{\Psi}{}\right) (U^{(a_1)}_{\alpha_1, \beta_1} \otimes  U^{(a_2)}_{\alpha_2, \beta_2})^\dagger O^\dagger_{\mathrm{D}}. \label{eq:ox:gen:3}
\end{align}
We observe that the noise maps $U^{(a_1)}_{\alpha_1, \beta_1} \otimes  U^{(a_2)}_{\alpha_2, \beta_2}$ in (\ref{eq:ox:gen:3}) act on Alice's part of the systems only. But this is sufficient due to the symmetry of Bell-states - noise on Bobs side can be moved to the other side. Furthermore the additional $\px$-flips introduced on the system(s) of L by the unitaries $U^{(a_1)}_{\alpha_1, \beta_1} \otimes  U^{(a_2)}_{\alpha_2, \beta_2}$ are used to keep track of the noise map applied. Because Alice and Bob apply the depolarization procedure as described in \cite{DuerStd} and L keeps track of the effective error introduced, we can safely assume that the additional $\px$-flips will be introduced \emph{after} Alice and Bob complete the depolarization procedure, hence it is sufficient to consider two qubit correlated noise introduced at Alice's part of the systems.

\subsection{The BBPSSW protocol}

The protocol proposed in \cite{Bennett} (also referred to as BBPSSW protocol) is very similar to the DEJMPS protocol. Instead of step \ref{enu.oxford.step1} of Protocol \ref{protocol.rd} Alice and Bob apply a correlated depolarization procedure (twirl) to their input states which brings them to Werner form.

For the subsequent analysis, suppose that each pair of Alice and Bob is of the form $\rho(p) = p \ketbra{B_{00}}{B_{00}}{} + (1-p)\frac{1}{4} id$. We assume that the apparatus applies independent and identical noise of the form $N \rho(p) = f \rho(p) + (1-f)/4 (\rho(p) + \px \rho(p) \px + \py \rho(p) \py + \pz \rho(p) \pz)$ before each distillation step. In similar fashion to the DEJMPS protocol one easily obtains the recurrence relation for the noisy BBPSSW protocol:
\begin{align*}
\tilde{p} = \frac{4 p^2 f^2 + 2 p f}{3 p^2 f^2 + 3} = b(p).
\end{align*}
The fixed point $p^\infty$ of the protocol is obtained by solving the equation $b(p^\infty) = p^\infty$. A straightforward computation gives the fixed point $p^\infty = 2/3 + 1/3 \sqrt{4 - 9/f^2 + 6/f}$ (which depends on the noise parameter $f$). It was shown in \cite{Duer} that this fixed point is an attractor assuming sufficiently high initial fidelity with $\ket{B_{00}}$ per input pair. Expressing the recurrence relation $b$ in terms of its Taylor series around $p^\infty$ leads to
\begin{align}\label{eqn.bbpswtaylor}
\tilde{p} = b(p) \approx b(p^\infty) + b'(p^\infty)(p-p^\infty).
\end{align}
Hence, (\ref{eqn.bbpswtaylor}) provides an approximation of the error in terms of fidelity with $\ket{B_{00}}$ after $n+1$ successful distillation rounds, i.e. $\epsilon_{n+1} = \left(b'(p^\infty)\right)^{n} \epsilon_{1}$, see also the plots within Fig. \ref{fig:ibm}. Moreover, we compute the first derivative of $b$ by
\begin{align*}
b'(p) = \frac{2 f (1 + 4 f p - f^2 p^2)}{3 (1 + f^2 p^2)^2}.
\end{align*}
Evaluating $b'$ at $p^\infty$ yields
\begin{align}\label{eqn.bbpsswderivativefix}
b'(p^\infty) = \frac{9 - 3 f}{f (3 + 2 (2 + \sqrt{4 - 9/f^2 + 6/f}) f)}.
\end{align}

\begin{figure}[htb]
\centering
\hspace*{1.5em}\raisebox{\dimexpr-.5\height-1em}
  {\includegraphics{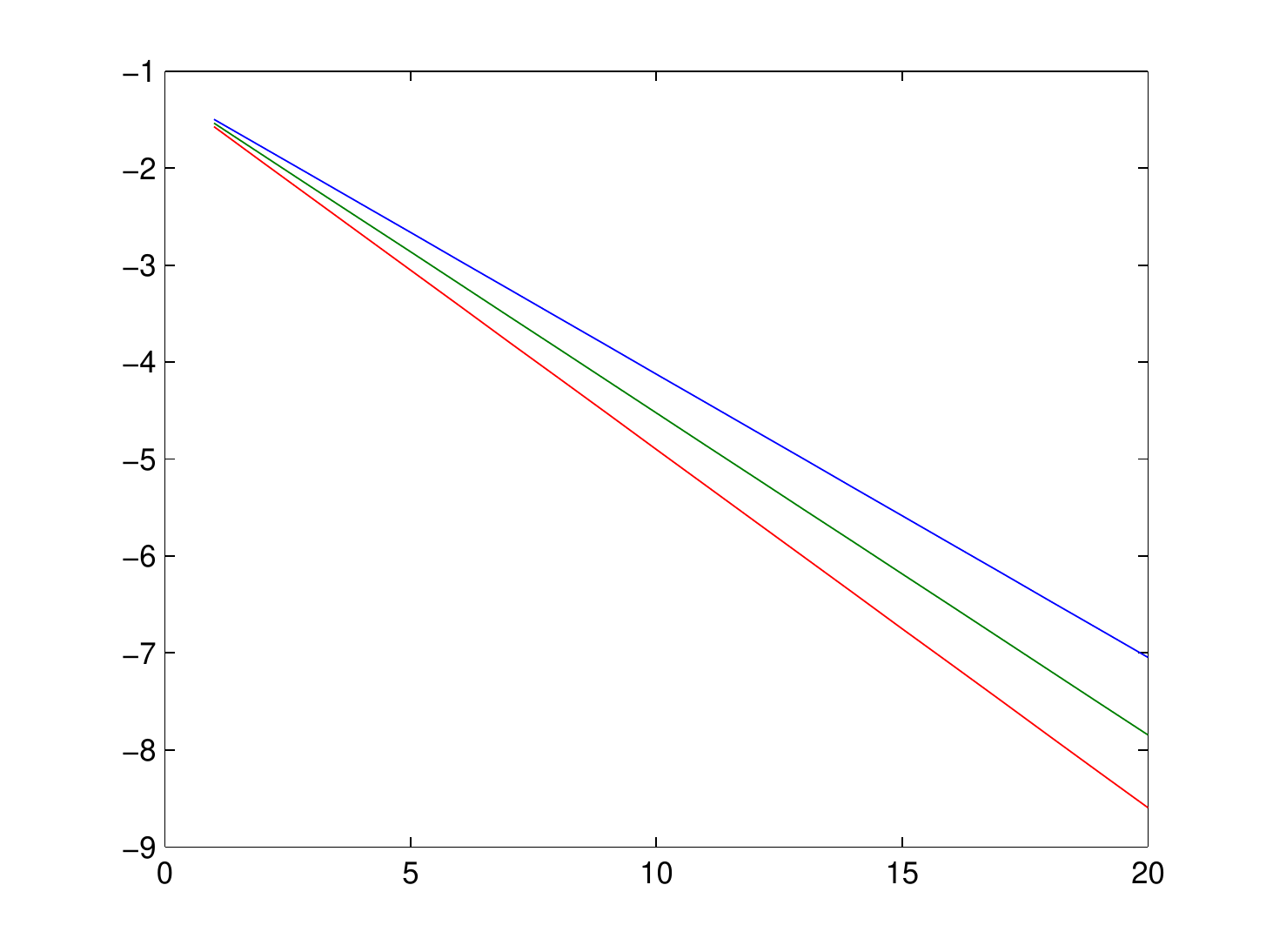}
}%

\leavevmode\smash{\makebox[0pt]{\hspace{-45em}
  \rotatebox[origin=l]{90}{\hspace{12em}
    $\log \epsilon_{n+1} = \log\left(\left(b'(p^\infty)\right)^{n}\epsilon_1 \right)$}%
}}

\hspace{0pt plus 1filll}\null

$n$

\medskip

\caption[h!]{\label{fig:ibm} The figure illustrates $\log \epsilon_{n+1}$ for the BBPSSW protocol for different noise parameters $f = 0.97$ (blue), $f = 0.98$ (green) and $f = 0.99$ (red).}
\end{figure}
From this we conclude that, if the apparatus is perfect, i.e. $f=1$ in (\ref{eqn.bbpsswderivativefix}), the error in terms of fidelity with $\ket{B_{00}}$ after $n+1$ successful distillation rounds scales as $\epsilon_{n+1} = (2/3)^{n} \epsilon_{1}$.

Using $\log_2 N = n$, where $N$ denotes the number of initial states, we infer for $\epsilon_{n+1}$ that
\begin{align*}
\epsilon_{n+1} = \epsilon_{1} b'(p^\infty)^{\log_2 N} = \epsilon_{1} \left(2^{\log_2 b'(p^\infty)} \right)^{\log_2 N} = \epsilon_{1} N^{\log_2 b'(p^\infty)}.
\end{align*}
This implies that $\epsilon_{n+1}$ scales as $F(N) \in O(N^{\log_2 b'(p^\infty)})$ and thus $\|\rho_{\text{fix}} - \rho_{n} \|_{1}$, where  $\rho_{\text{fix}}$ and $\rho_{n}$ denote the fixed point and the state after $n$ successful distillation rounds respectively, scales also as $F(N) \in O(N^{\log_2 b'(p^\infty)})$ as mentioned in the main text. \newline
For the analysis of two qubit correlated noise we assume that the noisy operations used by the BBPSSW protocol are of the form
\begin{align}
O_{12} \rho = \tilde{f} O^{\mathrm{ideal}}_{12} \rho + \frac{1-\tilde{f}}{4} \ptr{12}{\rho} \otimes \id_{12} \label{eq:ibm:noisyop}
\end{align}
where $\rho$ is a two qubit density operator and $O^{\mathrm{ideal}}_{12}$ denotes the ideal two qubit quantum gate. Observe that (\ref{eq:ibm:noisyop}) coincides with the standard form of \cite{DuerStd}. If the noisy quantum gates are not of the form (\ref{eq:ibm:noisyop}) we bring them to that standard form via the same depolarization procedure mentioned in the analysis of the DEJMPS protocol. Hence the following anaylsis is not restricted to this specific noise model, but actually applies to arbitrary noise processes describing noisy two qubit gates. \newline
It has been shown in \cite{Duer} that the BBPSSW protocol converges for noisy CNOT gates of the form (\ref{eq:ibm:noisyop}) to a unique and attracting fixed point if $\tilde{f}$ is sufficiently high. The recurrence relation for the fidelity relative to $\ket{B_{00}}$ obtained in \cite{Duer} is given by the formula
\begin{align}
F' = \frac{\tilde{f}^2 (F^2 + (\frac{1-F}{3})^2) + \frac{1-\tilde{f}^2}{8}}{\tilde{f}^2(F^2 + \frac{2F(1-F)}{3} + 5 (\frac{1-F}{3})^2) + \frac{1-\tilde{f}^2}{2}}. \label{eq:ibm:recgeneral}
\end{align}
Hence one obtains as in \cite{Duer} the respective fixed points of (\ref{eq:ibm:recgeneral}) to be
\begin{align*}
F_{\mathrm{min,max}} = \frac{3 \pm \sqrt{10 - 9/\tilde{f}^2}}{4}.
\end{align*}
For $F \in (F_{\mathrm{min}}, F_{\mathrm{max}})$ we have that $F' > F$ which shows that $F_{\mathrm{max}}$ is an attracting fixed point. By replacing $F'$ in (\ref{eq:ibm:recgeneral}) with $\tilde{b}(F)$ we observe similar to (\ref{eqn.bbpswtaylor}) that the error after $n+1$ successful distillation rounds scales for two qubit correlated noise as $F(N) \in O(N^{\log_2 \tilde{b}'(F_{\mathrm{max}})})$ where $N$ denotes the number of initial states. \newline
Finally we provide a worst case analysis of the BBPSSW protocol. For that purpose assume the following scenario: The noisy apparatus performs with probability $f_{\mathrm{I}}$ the ideal distillation step $\E_{\mathrm{I}}$ and introduces with probability $1-f_{\mathrm{I}}$ an arbitrary noise map $\E_{\perp}$. More precisely, we decompose the distillation step taken by Alice and Bob before the measurement of the target system as the CP map
\begin{align*}
\E(\rho) = f_{\mathrm{I}} \E_{\mathrm{I}}(\rho) + (1-f_{\mathrm{I}}) \E_{\perp} (\rho)
\end{align*}
where $\rho$ is a four qubit density operator. Notice that one can always decompose a noisy map in this form, where both maps are completely positive and trace preserving. We remark, however, that the map $\E_{\mathrm{I}}$ denotes the ideal protocol which includes an abort option, i.e. we only keep the first pair if the results of the measurements on the second pair coincide. The map $\E_{\perp}$ may similarly contain such an abort branch. The noise parameter $f_{\mathrm{I}}$ describes the quality of the overall map \footnote{We remark that a similar analysis can be performed by modelling local operations of Alice and Bob sepearetely in this way.}, i.e. one can think of the process that with probability $f_{\mathrm{I}}$ the desired procedure (including gates and measurements) is performed, while with probability $(1-f_{\mathrm{I}})$ something else happens (described by the map $\E_{\perp}$). 
\newline
We will now consider the worst case for the map $\E_{\perp}$ w.r.t. entanglement distillation. The worst case for the BBPSSW protocol is that the apparatus introduces a state orthogonal to $\ket{B_{00}}$ on the source system and the state $\ket{B_{00}}$ on the target system as this will always contribute to the overall success probability of a distillation step of the BBPSSW protocol but lead at the same time to a lower fidelity relative to $\ket{B_{00}}$ after the measurement of the target system compared to the ideal distillation step. One example for such a map is given by $\E_{\perp} (\rho) = \dm{B_{01}} \otimes \dm{B_{00}}$. Any other map will lead to a larger fidelity after the distillation step followed by depolarization to Werner form. We thus have 
\begin{align}
F ' \geq  \frac{f_{\mathrm{I}}(F^2 + (\frac{1-F}{3})^2)}{f_{\mathrm{I}}(F^2 + \frac{2F(1-F)}{3} + 5 (\frac{1-F}{3})^2) + 1 - f_{\mathrm{I}}} \label{eq:ibm:gen:rec}
\end{align}
for the fidelity relative to $\ket{B_{00}}$. This formula can be understood as follows: The ideal protocol is applied with probability $f_{\mathrm{I}}$, and succeeds with probability $f_{\rm suc}$, thereby producing a fidelity $\tilde F$. The map $\E_{\perp}$ is applied with probability $(1-f_{\mathrm{I}})$, does never abort and does not contribute to the final fidelity (which is clearly the worst case). We thus have $F' \geq f_{\mathrm{I}} f_{\rm suc} \tilde F/[f_{\mathrm{I}} f_{\rm suc} + (1-f_{\mathrm{I}})]$

We now analyze the worst case scenario, i.e. assuming equality in (\ref{eq:ibm:gen:rec}). Since we know that at each step the actual noise map produces an output density operator with a larger fidelity than the worst-case map, we can conclude that the resulting fidelity of any noise map will be larger than the fixed point which is achieved by the worst-case map. We remark, however, that this does not constitute a full confidentiality proof for arbitrary noise maps, as it is not evident from this analysis that for any fixed noise map a unique fixed point is reached. 
Assuming equality in (\ref{eq:ibm:gen:rec}), one can compute that the fixed points of the noisy BBPSSW protocol are in this case given by the solutions of
\begin{align}
-f_{\mathrm{I}} + (9-2f_{\mathrm{I}})F_{\infty} - 14f_{\mathrm{I}} F^2_{\infty} + 8f_{\mathrm{I}} F^3_{\infty} = 0 \label{eq:ibm:gen:fixeq}
\end{align}
which only depend on the noise parameter $f_{\mathrm{I}}$. We define $g_{\mathrm{fix}}(x,f_{\mathrm{I}}) = -f_{\mathrm{I}} + (9-2f_{\mathrm{I}})x - 14f_{\mathrm{I}} x^2 + 8f_{\mathrm{I}} x^3$ which implies that (\ref{eq:ibm:gen:fixeq}) reads as $g_{\mathrm{fix}}(F_{\infty},f_{\mathrm{I}}) = 0$. The question how many solutions of (\ref{eq:ibm:gen:fixeq}) are real we easily answer by the discriminant of $g_{\mathrm{fix}}$. We obtain for the discriminant of $g_{\mathrm{fix}}$
\begin{align}
\Delta(f_{\mathrm{I}}) = -36 (648 f_{\mathrm{I}} - 873 f_{\mathrm{I}}^2 - 212 f_{\mathrm{I}}^3 + 436 f_{\mathrm{I}}^4). \label{eq:ibm:gen:disc}
\end{align}
Hence if $\Delta(f_{\mathrm{I}}) > 0$ then all three solutions of (\ref{eq:ibm:gen:fixeq}) are real. We numerically estimate that $\Delta(f_{\mathrm{I}_{\mathrm{crit}}}) = 0$ for $f_{{\mathrm{I}}_{\mathrm{crit}}} \approx 0.9641$, hence for $f_{\mathrm{I}} > f_{{\mathrm{I}}_{\mathrm{crit}}}$ there exist three real solutions of (\ref{eq:ibm:gen:fixeq}) because $\Delta(f_{\mathrm{I}}) > 0$ for $f_{\mathrm{I}} > f_{{\mathrm{I}}_{\mathrm{crit}}}$, see Fig. \ref{fig:ibm:gen:disc}.
\begin{figure}[htb]
\centering
\hspace{0pt plus 1filll}\null

$f_{\mathrm{I}}$

\medskip
\hspace*{1.5em}\raisebox{\dimexpr-.5\height-1em}
  {\includegraphics[scale=0.92]{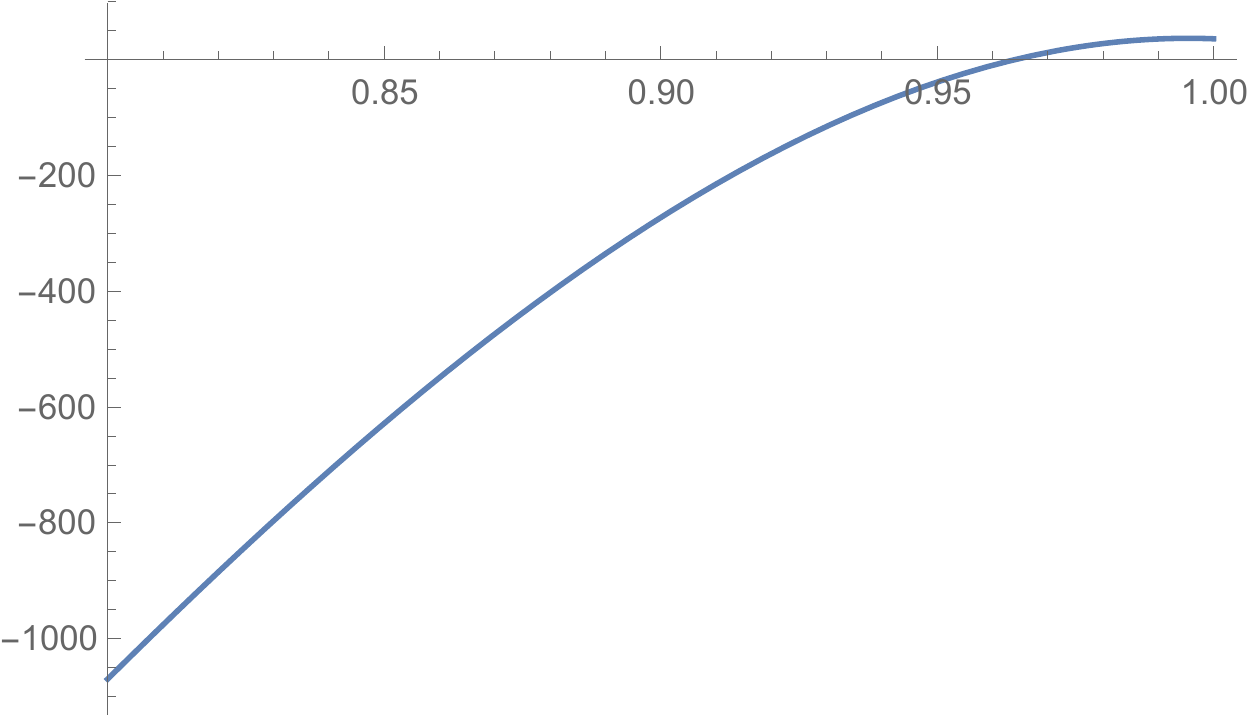}}%

\leavevmode\smash{\makebox[0pt]{\hspace{-40em}
  \rotatebox[origin=l]{90}{\hspace{8em}
    $\Delta(f_{\mathrm{I}})$}%
}}

\caption[h!]{\label{fig:ibm:gen:disc} The figure illustrates the discriminant $\Delta(f_{\mathrm{I}})$ of (\ref{eq:ibm:gen:fixeq}). For $f_{\mathrm{I}} > 0.9641$ we have $\Delta(f_{\mathrm{I}}) > 0$.}
\end{figure}
Thus, for $f_{\mathrm{I}} > f_{{\mathrm{I}}_{\mathrm{crit}}}$, we compute the fixed points of the noisy BBPSSW protocol via solving (\ref{eq:ibm:gen:fixeq}). Fig. \ref{fig:ibm:gen:fixpoint} shows the function $g_{\mathrm{fix}}$ for different values of $f_{\mathrm{I}}$.
\begin{figure}[htb]
\centering
\hspace{0pt plus 1filll}\null

$F$

\medskip
\hspace*{1.5em}\raisebox{\dimexpr-.5\height-1em}
  {\includegraphics[scale=0.92]{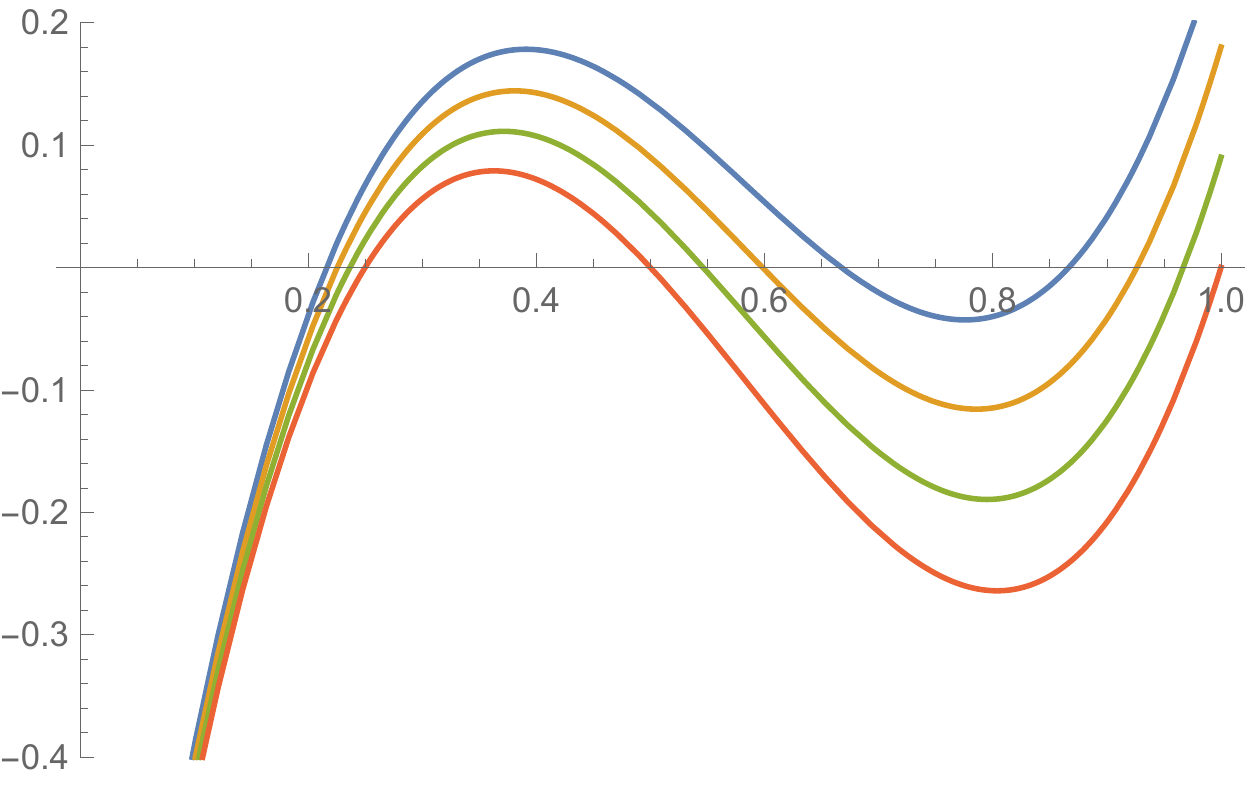}}%

\leavevmode\smash{\makebox[0pt]{\hspace{-40em}
  \rotatebox[origin=l]{90}{\hspace{8em}
    $g_{\mathrm{fix}}(F,f_{\mathrm{I}})$}%
}}

\caption[h!]{\label{fig:ibm:gen:fixpoint} The figure illustrates the function $g_{\mathrm{fix}}$ for $f_{\mathrm{I}}=0.97$ (blue), $f_{\mathrm{I}}=0.98$ (yellow), $f_{\mathrm{I}}=0.99$ (green) and $f_{\mathrm{I}}=1$ (red). The fixed points of the BBPSSW protocol correspond to the zero's of $g_{\mathrm{fix}}(F,f_{\mathrm{I}})$.}
\end{figure}
From Fig. \ref{fig:ibm:gen:disc} and \ref{fig:ibm:gen:fixpoint} we infer that we have three possible fixed points for $f_{\mathrm{I}} > f_{{\mathrm{I}}_{\mathrm{crit}}}$. Hence we need to show that the fixed point with the highest fidelity relative to $\ket{B_{00}}$ obtained via (\ref{eq:ibm:gen:fixeq}) is an attracting fixed point. We solve this issue by showing that $F' > F$ for $F \in (F_\text{min}, F_{\max})$ (where $F_\text{min}$ denotes the second, and $F_{\max}$ the third fixed point in Fig. \ref{fig:ibm:gen:fixpoint}). From Fig. \ref{fig:ibm:gen:attr} we find that $F' - F > 0$ for $f_{\mathrm{I}} > f_{{\mathrm{I}}_{\mathrm{crit}}}$, hence $F' > F$ which shows that $F_{\max}$ is an attracting fixed point whenever starting with initial fidelity $F > F_\text{min}$.
\begin{figure}[htb]
\centering
\hspace{0pt plus 1filll}\null

$F$

\medskip
\hspace*{1.5em}\raisebox{\dimexpr-.5\height-1em}
  {\includegraphics[scale=0.92]{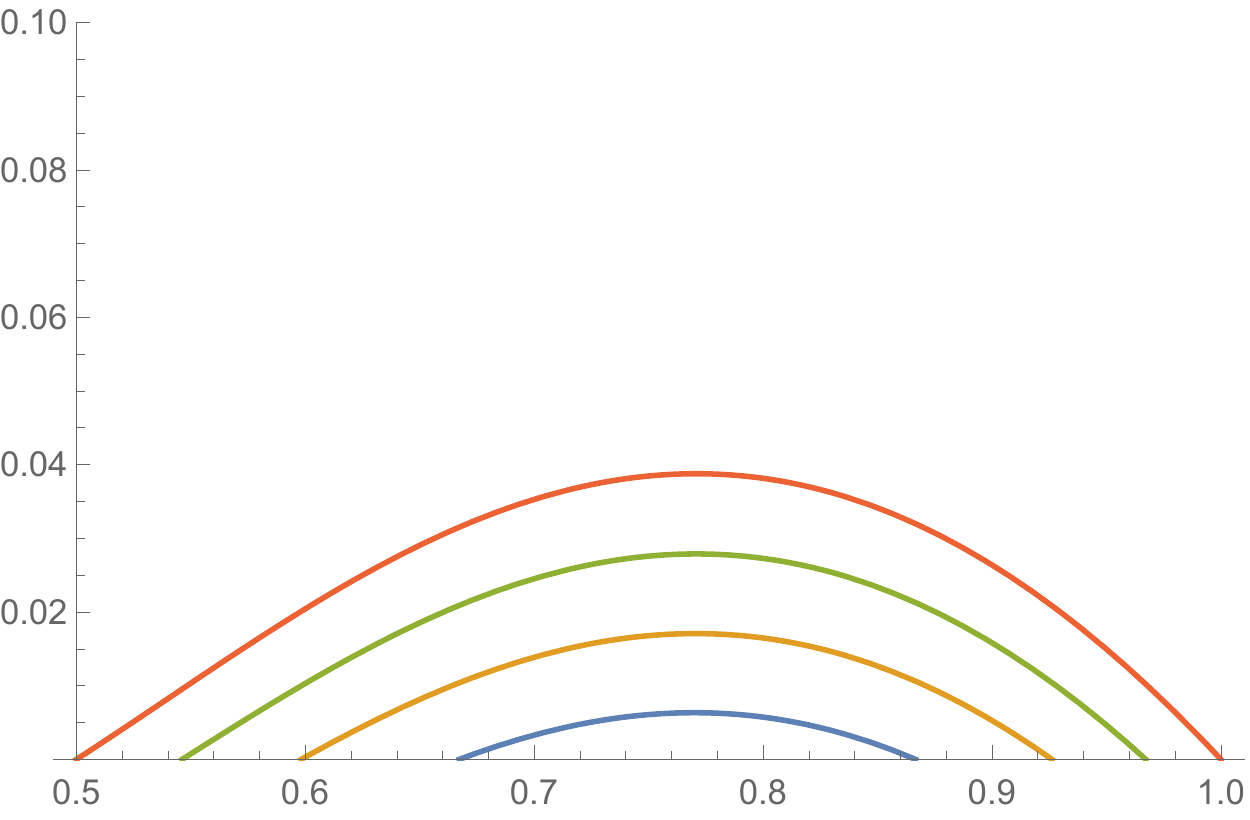}}%

\leavevmode\smash{\makebox[0pt]{\hspace{-40em}
  \rotatebox[origin=l]{90}{\hspace{8em}
    $F'-F$}%
}}

\caption[h!]{\label{fig:ibm:gen:attr} The figure illustrates $F' - F$ for $f_{\mathrm{I}}=0.97$ (blue), $f_{\mathrm{I}}=0.98$ (yellow), $f_{\mathrm{I}}=0.99$ (green) and $f_{\mathrm{I}}=1$ (red).}
\end{figure}

Furthermore, by assuming equality in (\ref{eq:ibm:gen:rec}) and replacing $F'$ with  $b_{\perp}(F)$, we observe similar to (\ref{eqn.bbpswtaylor}) that the error after $n+1$ successful distillation rounds scales in this worst case analysis as $F(N) \in O(N^{\log_2 b'_{\perp}(F_{\mathrm{max}})})$ where $N$ denotes the number of initial states.

\section{Confidentiality of entanglement distillation protocols}\label{sec.supp.noleak}

In this section we provide the proofs of Lemma \ref{thm.globalclosenessfactor} and Lemma \ref{lem:forpost} of the main text, crucial for the de-Finetti-based and post-selection-based reduction techniques. Both proofs require only one specific property of the real protocol $\E^\alpha$: after passing the parameter estimation phase the entanglement distillation protocol always converges to \emph{one} fixed point, i.e. the fixed point is \emph{unique}, an \emph{attractor} for all the states which pass the parameter estimation and depends on the noise parameters \emph{only}, as this implies that the distance with respect to the $1-$norm within the $ok-$branch of the protocol is bounded and converges towards zero.

\subsection{Proof of Lemma \ref{thm.globalclosenessfactor}}\label{app:sec:defin}

We first state the following lemma which establishes a connection between measurements on one subsystem of a bipartite state and tensor product states.

\begin{lemma}\label{lem.steering}[Steering of local states]
Let $\rho_{AB}$ be a bipartite (in general, mixed) state and let $\rho_A = \ptr{B}{\rho_{AB}}$ and $\rho_B = \ptr{A}{\rho_{AB}}$. Furthermore let $\rho^{\phi}_{B}$ be defined as
\begin{align*}
\rho^{\phi}_{B} = \frac{\ptr{A}{(\dm{\phi} \otimes I) \rho_{AB}}}{p_{A}(\phi)}
\end{align*}
where $\ket{\phi} \in \H_A$ and $p_{A}(\phi) = \tr{\left(\dm{\phi} \rho_{A} \right)}$. If $\|\rho^{\phi}_{B} - \rho_B \|_1 \leq \epsilon$ for all $\ket{\phi} \in \H_A$, then
\begin{align}\label{inequ.steering}
\|\rho_{AB} - \rho_A \otimes \rho_B \|_1 \leq 2 C \epsilon
\end{align}
where $C$ only depends on the dimensions of $A$ and $B$. In particular, if we fix the number of qubits of $A$ and $B$ to $2$ respectively, then we have $C = 4^8$.
\end{lemma}
\begin{proof}
In the following we denote the four Pauli operators by
\begin{align*}
\sigma_{0} = id, \quad \sigma_{1} = \px, \quad \sigma_{2} = \pz, \quad \sigma_{3} = \py.
\end{align*}
First we decompose $\rho_{AB}$ in the Pauli basis, i.e. we have
\begin{align}\label{equ.lem.meas.rhoab}
\rho_{AB} = \frac{1}{2^{n+m}} \sum\limits_{\mathbf{i},\mathbf{j}} \alpha_{\mathbf{ij}} \sigma_{\mathbf{i}} \otimes \sigma_{\mathbf{j}}
\end{align}
where $n$ and $m$ denote the number of qubits of $A$ and $B$ respectively and we use the notations $\mathbf{i} = (i_1,..,i_{n})$ and $\mathbf{j}=(j_1,..,j_{m})$ where each $i_k$ and $j_k$ are in $\lbrace 0,..,3 \rbrace$ as well as $\sigma_\mathbf{i} = \bigotimes^{n}_{k = 1} \sigma_{i_k}$ and $\sigma_\mathbf{j} = \bigotimes^{m}_{k = 1} \sigma_{j_k}$. Recall that $\tr{(\sigma_0)} = 2$ and $\tr{(\sigma_1)} = \tr{(\sigma_2)} = \tr{(\sigma_3)} = 0$. From this one easily computes $\rho_A$ and $\rho_B$ by
\begin{align}
\rho_A &= \ptr{B}{\rho_{AB}} =  \frac{1}{2^{n+m}} \sum\limits_{\mathbf{i,j}} \alpha_{\mathbf{ij}} \sigma_{\mathbf{i}} \tr{(\sigma_{\mathbf{j}})} =  \frac{1}{2^{n}} \sum\limits_{\mathbf{i}} \alpha_{\mathbf{i0}} \sigma_{\mathbf{i}}, \label{equ.lem.meas.partials.1} \\
\rho_B &= \ptr{A}{\rho_{AB}} = \frac{1}{2^{n+m}} \sum\limits_{\mathbf{i,j}} \alpha_{\mathbf{ij}} \tr{(\sigma_{\mathbf{i}})} \sigma_{\mathbf{j}} = \frac{1}{2^{m}} \sum\limits_{\mathbf{j}} \alpha_{\mathbf{0j}} \sigma_{\mathbf{j}}. \label{equ.lem.meas.partials.2}
\end{align}
Using (\ref{equ.lem.meas.rhoab}), (\ref{equ.lem.meas.partials.1}) and (\ref{equ.lem.meas.partials.2}) we obtain for (\ref{inequ.steering})
\begin{align}\label{inequ.lem.meas.1}
\|\rho_{AB} - \rho_A \otimes \rho_B \|_1 & \leq \frac{1}{2^{n+m}} \sum\limits_{\mathbf{i,j}} \|(\alpha_{\mathbf{ij}} - \alpha_{\mathbf{i0}} \alpha_{\mathbf{0j}}) \sigma_{\mathbf{i}} \otimes \sigma_{\mathbf{j}} \|_1 = \frac{1}{2^{n+m}} \sum\limits_{\mathbf{i,j}} |\alpha_{\mathbf{ij}} - \alpha_{\mathbf{i0}} \alpha_{\mathbf{0j}}| \cdot \|\sigma_{\mathbf{i}} \otimes \sigma_{\mathbf{j}}\|_1 \notag \\
& = \frac{2^{n+m}}{2^{n+m}} \sum\limits_{\mathbf{i,j}} |\alpha_{\mathbf{ij}} - \alpha_{\mathbf{i0}} \alpha_{\mathbf{0j}}| = \sum\limits_{\mathbf{i,j}} |\alpha_{\mathbf{ij}} - \alpha_{\mathbf{i0}} \alpha_{\mathbf{0j}}| = \|\mathbf{a} - \mathbf{a'} \|_{1;\C^{4^{n+m}}}
\end{align}
where $\mathbf{a} = (\alpha_{00},..,\alpha_{3^n 3^m})$, $\mathbf{a'} = (\alpha_{00} \alpha_{00},..,\alpha_{3^n 0} \alpha_{0 3^m})$ and $\| \cdot \|_{1;\C^{4^{n+m}}}$ denotes the $1-$norm of vectors in $\C^{4^{n+m}}$. Hence in order to prove (\ref{inequ.steering}) it is sufficient to prove $\|\mathbf{a} - \mathbf{a'} \|_{1;\C^{4^{n+m}}} \leq 2C \epsilon$. By assumption we have for $\rho^{\phi}_{B}$ where $\ket{\phi} \in \H_A$ and $p_{A}(\phi) = \tr{\left((\dm{\phi} \otimes I) \rho_{AB}\right)}$ that $\|\rho^{\phi}_{B} - \rho_B \|_1 \leq \epsilon$ for all $\ket{\phi} \in \H_A$. Moreover, according to Theorem 9.1 in \cite{Nielsen} we have for all $\ket{\xi} \in \H_B$
\begin{align}\label{inequ.lem.meas.nc}
\frac{1}{2} |p_B(\xi | \phi) - q_B(\xi)| = \frac{1}{2}\left|\tr{\left(\dm{\xi} \rho^{\phi}_{B} \right)} - \tr{\left(\dm{\xi} \rho_{B} \right)} \right| \leq \max_{E_m} \frac{1}{2} \sum\limits_{m} \left|\tr{\left(E_m \rho^{\phi}_{B} \right)} - \tr{\left(E_m \rho_{B} \right)} \right| = \|\rho^{\phi}_{B} - \rho_B \|_1 \leq \epsilon
\end{align}
where $p_B(\xi | \phi)$ denotes the conditional probability of obtaining the outcome $\phi$ on system $A$ and the outcome $\xi$ on system $B$ and $\{E_m \}$ denotes a POVM on the subsystem of $B$. Suppose we perform a projective measurement on the systems of $A$ and $B$ denoted by $\{\ket{\psi_k}_{AB} \} = \{\ket{\phi_k}_{A} \otimes \ket{\xi_k}_{B} \}$ where $k \in \lbrace 1,.., 4^{n+m}\rbrace$ on $\rho_{AB}$ and $\rho_A \otimes \rho_B$. This yields for the respective probabilities $p_{AB}(\psi_k)$ and $q_{AB}(\psi_k)$ of observing outcome $k$ for $\rho_{AB}$ and $\rho_A \otimes \rho_B$
\begin{align*}
p_{AB}(\psi_k) &= \tr{(\dm{\psi_k} \rho_{AB})} = \tr{\left(\dm{\phi_k}_A \otimes \dm{\xi_k}_B \rho_{AB} \right)} = \tr{\left(\dm{\xi_k}_B \ptr{A}{(\dm{\phi_k}_A \otimes I) \rho_{AB}}\right)} \\
& = \tr{\left(\dm{\xi_k}_B p_A(\phi_k) \rho^{\phi_k}_{B} \right)} = p_A(\phi_k) \tr{\left(\dm{\xi_k}_B \rho^{\phi_k}_{B} \right)} = p_A(\phi_k) p_B(\xi_k | \phi_k),\\
q_{AB}(\psi_k) &= \tr{(\dm{\psi_k} \rho_{A} \otimes \rho_{B})} =  \tr{(\dm{\phi_k} \rho_{A})} \tr{(\dm{\xi_k} \rho_{B})} = q_A(\phi_k) q_B(\xi_k)
\end{align*}
where $p_B(\xi_k | \phi_k)$ denotes the conditional probability of obtaining outcome $\phi_k$ on system $A$ first and obtaining outcome $\xi_k$ on system $B$. We observe $p_A(\phi_k) = q_A(\phi_k)$. Thus we obtain
\begin{align*}
|p_{AB}(\psi_k) - q_{AB}(\psi_k)| = p_A(\phi_k) |p_B(\xi_k | \phi_k) - q_B(\xi_k)| \leq 2\epsilon p_A(\phi_k)
\end{align*}
using (\ref{inequ.lem.meas.nc}). In order to compute a bound for (\ref{inequ.lem.meas.1}) we use quantum state tomography, see e.g. \cite{tomo}. For that purpose we perform an informationally complete POVM induced by different separable bases on $\H_A \otimes \H_B$. More precisely, we choose that many POVMs such that we have in total $4^{n+m}$ different outcomes. We observe for $\ket{\psi_k}_{AB} = \ket{\phi_k}_{A} \otimes \ket{\xi_k}_{B}$ that
\begin{align}\label{equ.lem.meas.pandqtomo}
p_{AB}(\psi_k) = \frac{1}{2^{n+m}} \sum\limits_{\mathbf{i,j}} \bra{\phi_k} \sigma_\mathbf{i} \ket{\phi_k} \bra{\xi_k} \sigma_\mathbf{j} \ket{\xi_k} \alpha_{\mathbf{ij}} \quad \text{and} \quad q_{AB}(\psi_k) = \frac{1}{2^{n+m}} \sum\limits_{\mathbf{i,j}} \bra{\phi_k} \sigma_\mathbf{i} \ket{\phi_k} \bra{\xi_k} \sigma_\mathbf{j} \ket{\xi_k} \alpha_{\mathbf{i0}} \alpha_{\mathbf{0j}}.
\end{align}
Enumerating (\ref{equ.lem.meas.pandqtomo}) for $1 \leq k \leq 4^{n+m}$ yields $4^{n+m}$ equations for $\mathbf{a}$, i.e.
\begin{align}\label{inequ.lem.meas.20}
p_{AB}(\psi_1) &= \frac{1}{2^{n+m}} \sum\limits_{\mathbf{i,j}} \bra{\phi_1} \sigma_\mathbf{i} \ket{\phi_1} \bra{\xi_1} \sigma_\mathbf{j} \ket{\xi_1}  \alpha_{\mathbf{ij}},\\
 &... \notag \\
p_{AB}(\psi_{4^{n+m}}) &= \frac{1}{2^{n+m}} \sum\limits_{\mathbf{i,j}}  \bra{\phi_{4^{n+m}}} \sigma_\mathbf{i} \ket{\phi_{4^{n+m}}} \bra{\xi_{4^{n+m}}} \sigma_\mathbf{j} \ket{\xi_{4^{n+m}}} \alpha_{\mathbf{ij}} \label{inequ.lem.meas.21}
\end{align}
as well as $4^{n+m}$ equations for $\mathbf{a'}$
\begin{align}\label{inequ.lem.meas.30}
q_{AB}(\psi_1) &= \frac{1}{2^{n+m}} \sum\limits_{\mathbf{i,j}}  \bra{\phi_1} \sigma_\mathbf{i} \ket{\phi_1} \bra{\xi_1} \sigma_\mathbf{j} \ket{\xi_1}\alpha_{\mathbf{i0}} \alpha_{\mathbf{0j}}, \\
& ... \notag \\
 q_{AB}(\psi_{4^{n+m}}) &= \frac{1}{2^{n+m}} \sum\limits_{\mathbf{i,j}}  \bra{\phi_{4^{n+m}}} \sigma_\mathbf{i} \ket{\phi_{4^{n+m}}} \bra{\xi_{4^{n+m}}} \sigma_\mathbf{j} \ket{\xi_{4^{n+m}}} \alpha_{\mathbf{i0}} \alpha_{\mathbf{0j}} \label{inequ.lem.meas.31}.
\end{align}
We can rewrite the systems of equations (\ref{inequ.lem.meas.20})-(\ref{inequ.lem.meas.21}) and (\ref{inequ.lem.meas.30})-(\ref{inequ.lem.meas.31}) using
\begin{align*}
T = \begin{pmatrix}
\bra{\phi_1} \sigma_0 \ket{\phi_1} \bra{\xi_1} \sigma_0 \ket{\xi_1} &  ... & \bra{\phi_1} \sigma_{3^n} \ket{\phi_1} \bra{\xi_1} \sigma_{3^m} \ket{\xi_1} \\
... & ... & ... \\
\bra{\phi_{4^{n+m}}} \sigma_0 \ket{\phi_{4^{n+m}}} \bra{\xi_{4^{n+m}}} \sigma_0 \ket{\xi_{4^{n+m}}} & ... & \bra{\phi_{4^{n+m}}} \sigma_{3^n} \ket{\phi_{4^{n+m}}} \bra{\xi_{4^{n+m}}} \sigma_{3^m} \ket{\xi_{4^{n+m}}}
\end{pmatrix}
\end{align*}
and $\mathbf{p} = (p_{AB}(\psi_1),..,p_{AB}(\psi_{4^{n+m}}))$ and $\mathbf{q} = (q_{AB}(\psi_1),..,q_{AB}(\psi_{4^{n+m}}))$ as
\begin{align*}
\mathbf{p} = \frac{1}{2^{n+m}} T \mathbf{a} \quad \text{and} \quad \mathbf{q} = \frac{1}{2^{n+m}}  T \mathbf{a'}
\end{align*}
respectively. Hence $2^{n+m}(\mathbf{p} - \mathbf{q}) = T\left(\mathbf{a} - \mathbf{a'} \right)$. Moreover we observe that $T$ is invertible if the POVM is informationally complete, see \cite{tomo} for details. Thus, inverting $T$ and taking norms on both sides yields
\begin{align*}
\| \mathbf{a} - \mathbf{a'} \|_{1;\C^{4^{n+m}}} & \leq 2^{n+m} \|T^{-1}\| \| \mathbf{p} - \mathbf{q}\|_{1;\C^{4^{n+m}}} = 2^{n+m} \|T^{-1}\| \sum_{k} |p_{AB}(\psi_k) - q_{AB}(\psi_k)| \\
& \leq 2^{n+m} \|T^{-1}\| \sum_{k} 2 \epsilon p_A(\phi_k) \leq 2 \|T^{-1}\| 4^{n+m} 2^{n+m} \epsilon
\end{align*}
which completes the proof for the general case with $C = \|T^{-1}\| 4^{n+m} 2^{n+m}$. \newline
Before we complete the Lemma we need to determine $C$ for the case of $n=m=2$. 
We choose $\ket{\phi_{4^3(j_1-1) + 4^2(j_2-1) + 4(j_3-1) + j_4}} = \ket{\phi'_{j_1}} \otimes \ket{\phi'_{j_2}} \otimes \ket{\phi'_{j_4}} \otimes \ket{\phi'_{j_4}}$ where $j_1,j_2,j_3,j_4 \in \lbrace 1,2,3,4 \rbrace$ and
\begin{align}
\ket{\phi'_{1}} &= (\ket{0} + \ket{1})/\sqrt{2} , \label{eq:tomoset:1} \\
\ket{\phi'_{2}} &= (\ket{0} + i\ket{1})/\sqrt{2}, \label{eq:tomoset:2} \\
\ket{\phi'_{3}} &= \ket{0}, \label{eq:tomoset:3} \\
\ket{\phi'_{4}} &= (\ket{0} - \ket{1})/\sqrt{2}, \label{eq:tomoset:4} 
\end{align}
which is informationally complete and thus a valid choice. This choice of $\ket{\phi'_{l}}$ corresponds to a Pauli tomography on a single qubit. We observe that the matrix $T$ is invertible and compute $\|T^{-1}\| = 16$. Thus $C = 4^8$ which completes the proof.
\end{proof}

Roughly speaking Lemma \ref{lem.steering} states that if all post-selected reduced states of a bipartite state, where each partition consists of two qubits, are $\eta-$close then the overall state is $2 \cdot 4^8 \eta$ close to a product state. \newline
We gave the lemma in a more general form as it may have utility beyond the scope of this paper. However for our purposes we need a stronger, but more specific result. In the following lemma we will show that we can achieve the same result even if the measurements must succeed above a threshold, which is important in the application of the lemma.

\begin{lemma}\label{lem:steering:prob}
In the situation of Lemma \ref{lem.steering} for $n=m=2$ it suffice to consider measurements on the subsystem $A$ which have a probability greater than or equal to $1/16$. \newline
More precisely, for every state $\rho_{AB}$ there exists a unitary $U$ acting on system $A$ and a state $\rho'_{AB} = (U \otimes I_{B}) \rho_{AB} (U \otimes I_{B})^\dagger$, such that if the state $\rho'_{AB}$ 
meets the conditions of Lemma \ref{lem.steering}, i.e. subsystem $B$ is $\epsilon-$non-steerable via measurements on subsystem $A$ for all measurements with probability greater than or equal to $1/16$, then
\begin{align}
\| \rho_{AB} - \rho_{A} \otimes \rho_{B} \|_1 \leq 2C \epsilon.
\end{align}
\end{lemma}
\begin{proof}
First we construct the state $\rho'_{AB}$ associated with $\rho_{AB}$ and show that it suffice to consider measurements of probability greater than or equal to $1/16$. Recall the situation of Lemma \ref{lem.steering}. Let $\rho_{AB}$ be a bipartite (in general, mixed) state and let $\rho_A = \ptr{B}{\rho_{AB}}$ and $\rho_B = \ptr{A}{\rho_{AB}}$. Furthermore let $\rho^{\phi}_{B}$ be defined as
\begin{align*}
\rho^{\phi}_{B} = \frac{\ptr{A}{(\dm{\phi} \otimes I) \rho_{AB}}}{p_{A}(\phi)}
\end{align*}
where $\ket{\phi} \in \H_A$ and $p_{A}(\phi) = \tr{\left(\dm{\phi} \rho_{A} \right)}$. Then the claim of Lemma \ref{lem.steering} was: If $\|\rho^{\phi}_{B} - \rho_B \|_1 \leq \epsilon$ for all $\ket{\phi} \in \H_A$, then
\begin{align}
\|\rho_{AB} - \rho_A \otimes \rho_B \|_1 \leq 2 C \epsilon
\end{align}
where $C$ only depends on the dimensions of $A$ and $B$. In particular, if we fix the number of qubits of $A$ and $B$ to $2$ respectively, then we have $C = 4^8$. \newline 
Further recall that the set $\ket{\phi_{4^3(j_1-1) + 4^2(j_2-1) + 4(j_3-1) + j_4}} = \ket{\phi'_{j_1}} \otimes \ket{\phi'_{j_2}} \otimes \ket{\phi'_{j_4}} \otimes \ket{\phi'_{j_4}}$ where $j_1,j_2,j_3,j_4 \in \lbrace 1,2,3,4 \rbrace$ of Lemma \ref{lem.steering}, i.e. (\ref{eq:tomoset:1})--(\ref{eq:tomoset:4}), is informationally complete and thus suffice to reconstruct any $4$ qubit quantum state where
\begin{align}
\ket{\phi'_{1}} &= (\ket{0} + \ket{1})/\sqrt{2} , \label{eq:lem:prob:tom:1}\\
\ket{\phi'_{2}} &= (\ket{0} + i\ket{1})/\sqrt{2}, \label{eq:lem:prob:tom:2}\\
\ket{\phi'_{3}} &= \ket{0}, \label{eq:lem:prob:tom:3}\\
\ket{\phi'_{4}} &= (\ket{0} - \ket{1})/\sqrt{2}. \label{eq:lem:prob:tom:4}
\end{align}
In order to prove the claim, we use the following observation: The state $\rho_{A} = \text{tr}_{B}[\rho_{AB}]$ is a two qubit quantum state, so it can be written as 
\begin{align}
\rho_{A} = \sum^3_{j=0} \lambda_j \dm{\Psi_j} \label{eq:lem:prob:1}
\end{align}
where the states $\ket{\Psi_j}$ correspond to the (orthogonal) eigenstates of $\rho_{A}$ for the real non-negative eigenvalues $\lambda_{j}$. Hence there exists at least one $j' \in \lbrace 0,1,2,3 \rbrace$ such that $\lambda_{j'} \geq 1/4 $, which corresponds to the maximum of the eigenvalues $\lambda_j$. Now we choose a local unitary $U$ such that $U \ket{\Psi_{j'}} = \ket{0} \otimes \ket{0}$. Applying this unitary to (\ref{eq:lem:prob:1}) therefore leads to the state
\begin{align}
\rho'_{A} = U \left(\sum^3_{j=0} \lambda_j \dm{\Psi_j} \right) U^\dagger = \lambda_{j'} \dm{00} + \sum^3_{j \neq j'} \lambda_j \dm{\varphi_{j}} \label{eq:lem:prob:2}
\end{align}
where $\ket{\varphi_j} = U \ket{\psi_j}$. We compute the probability for  any projector applied on $\rho'_{A}$ which is taken from the set (\ref{eq:lem:prob:tom:1})--(\ref{eq:lem:prob:tom:4}) and of the form $\dm{\phi'} = \dm{\phi'_{k}} \otimes \dm{\phi'_{l}}$ by
\begin{align}
\text{tr} \left(\dm{\phi'} \rho'_{A} \right) &= \text{tr} \left(\dm{\phi'_{k}} \otimes \dm{\phi'_{l}} U \rho_{A} U^\dagger \right) = \sum^3_{j=0} \lambda_j \text{tr} \left(\dm{\phi'_{k}} \otimes \dm{\phi'_{l}} U \dm{\Psi_j} U^\dagger \right) \notag \\ 
& \geq \frac{1}{4} \text{tr} \left(\dm{\phi'_{k}} \otimes \dm{\phi'_{l}} U \dm{\Psi_{j'}} U^\dagger \right) = \frac{1}{4} \text{tr} \left(\dm{\phi'_{k}} \otimes \dm{\phi'_{l}} \dm{00} \right) \notag \\
& = \frac{1}{4} \text{tr} \left(\dm{\phi'_{k}} \dm{0} \right) \text{tr} \left(\dm{\phi'_{l}} \dm{0} \right) \geq \frac{1}{4} \frac{1}{2} \frac{1}{2} = \frac{1}{16} \label{eq:lem:prob:4}
\end{align}
where we have used that $\text{tr}(AB) = \text{tr}(BA)$ for matrices $A$ and $B$ and that $\text{tr}\left(\dm{\phi'_{k}} \dm{0} \right) \geq 1/2$  for all $k \in \lbrace 0,1,2,3 \rbrace$. \newline
So we define the state $\rho'_{AB}$ as  $\rho'_{AB} = (U \otimes I_{B}) \rho_{AB} (U \otimes I_{B})^\dagger$. Observe that the probabilities of all projectors within the tomographic set (\ref{eq:lem:prob:tom:1})--(\ref{eq:lem:prob:tom:4}) are greater than or equal to $1/16$ for the state $\rho'_{A}$. \newline  
Now suppose we perform a measurement from the tomographic set (\ref{eq:lem:prob:tom:1})--(\ref{eq:lem:prob:tom:4}) on the subsystem $A$ of $\rho'_{AB}$ yielding outcome $\ket{\phi}$. The post-selected state conditioned on $\ket{\phi}$ reads as 
\begin{align*}
\rho'^{\phi}_{B} = \frac{\ptr{A}{(\dm{\phi} \otimes I) \rho'_{AB}}}{p_{A}(\phi)}
\end{align*}
where $p_{A}(\phi) \geq 1/16$. Furthermore assume as in Lemma \ref{lem.steering} that $\|\rho'^{\phi}_{B} - \rho_{B} \|_1 \leq \epsilon$ for all such $\ket{\phi} \in \mathcal{H}_A$. Then Lemma \ref{lem.steering} implies that 
\begin{align}
\|\rho'_{AB} - \rho'_{A} \otimes \rho_{B} \|_1 \leq 2C \epsilon.
\end{align}
The proof completes by observing that $\rho'_{AB}$ and $\rho'_{A} \otimes \rho_{B}$ are related by the local unitary $U$ to $\rho_{AB}$ and $\rho_{A} \otimes \rho_{B}$ and the unitary equivalence of the trace distance, i.e.
\begin{align}
\|\rho_{AB} - \rho_{A} \otimes \rho_{B} \|_1 &= \| (U \otimes I_B)(\rho_{AB} - \rho_{A} \otimes \rho_{B})(U \otimes I_B)^\dagger \|_1 \\
&= \| \rho'_{AB} - \rho'_{A} \otimes \rho_{B} \|_1 \leq 2C \epsilon.
\end{align}
\end{proof}

We observe that, due to the proof of Lemma \ref{lem:steering:prob}, which relies on the informationally complete set (\ref{eq:lem:prob:tom:1})--(\ref{eq:lem:prob:tom:4}), it suffices to be non-steerable with respect to the measurements within that set for a probability of measurement above or equal to $1/16$. We actually have proven a stronger result, as the actual choice of measurements does not matter, provided the probability of success is above or equal to the threshold $1/16$.

\begin{lemmarep}[Lemma \ref{thm.globalclosenessfactor} in main text - Product Form Lemma]
Let $\rho$ be an arbitrary mixed state shared by Alice and Bob and let $\ket{\psi}_{ABE}$ be a purification thereof held by Eve. Furthermore, let $\mathcal{P}_1$ correspond to a (distillation-type) real protocol and $\mathcal{P}_2$ correspond to the associated (distillation-type) ideal protocol, i.e. 
\begin{align*}
\mathcal{P}_1 (\rho) &= p_{\rho} \sigma_{AB} \otimes \dm{ok} + (1-p_\rho) \sigma^{\perp}_{AB} \otimes \dm{fail}, \\
\mathcal{P}_2(\rho) &= p_{\rho} \sigma^{\alpha}_{AB} \otimes \dm{ok} + (1-p_\rho) \sigma^{\perp}_{AB} \otimes \dm{fail}. \notag
\end{align*}
where $\alpha$ characterizes the level of the noise, $\sigma^{\alpha}_{AB}$, and $\sigma^{\perp}_{AB}$ are two fixed two qubit states. Furthermore, let $\mathcal{P}_1$ and $\mathcal{P}_2$ satisfy the following properties:
\begin{enumerate}
	\item The noise transcripts do not leak to Eve.
	\item The protocol $\mathcal{P}_1$ guarantees to converge towards some state $\sigma^{\alpha}_{AB}$ within the ok-branch of the protocol and $\max_{\mu_{AB}} \|(\mathcal{P}_1 - \mathcal{P}_2)(\mu_{AB})\|_1 \leq \varepsilon$.
\end{enumerate}
Then it holds that
\begin{align}
\| (\mathcal{P}_1 \otimes id_{E} -\mathcal{P}_2 \otimes id_{E} )(\ketbra{\psi}{\psi}{ABE})\|_1 \leq (34 \cdot 4^8 +1) \varepsilon. \label{eq.thm.localsteer}
\end{align}
\end{lemmarep}
\begin{proof}
The proof relies on Lemma \ref{lem.steering} and \ref{lem:steering:prob}. Suppose Eve prepares the pure state $\ket{\psi}_{ABE}$ and let $\ptr{E}{\dm{\psi}} = \rho_{AB}$ be the state received by Alice and Bob. Then we have
\begin{align}
(\mathcal{P}_1 \otimes \id_E)(\dm{\psi}) &= p_{\rho} \sigma_{ABE} \otimes \dm{ok} + (1-p_\rho) \sigma^{\perp}_{AB} \otimes \sigma_E \otimes \dm{fail}, \label{eq.thm.localsteer.1} \\
(\mathcal{P}_2 \otimes \id_E)(\dm{\psi}) &= p_{\rho} \sigma^{\alpha}_{AB} \otimes \sigma_E \otimes \dm{ok} + (1-p_\rho) \sigma^{\perp}_{AB} \otimes \sigma_E \otimes \dm{fail}. \notag
\end{align}
If we post-select Eq. (\ref{eq.thm.localsteer.1}) on the $ok-$branch we have after normalization
\begin{align}\label{eq.thm.localsteer.2}
\frac{1}{p_{\rho}}(\id_{ABE} \otimes \dm{ok}) (\mathcal{P}_1 \otimes \id_E)(\dm{\psi}) = \sigma_{ABE} \otimes \dm{ok}.
\end{align}
It is obvious from the fact that the protocol is performed by Alice and Bob per definition that any measurement of Eve in the $ok-$branch can be commuted to the beginning of the protocol $\mathcal{P}_1$ because Eve is not part of the protocol. Hence her measurement only changes the input of the protocol $\mathcal{P}_1$ and thus either cause an abort or not.

We call the final state of Alice and Bob $\eta-$Eve-non steerable if for all $\phi \in \H_E$ we have $\|\sigma_{AB}^{\phi} - \sigma_{AB} \|_1 \leq \eta$ where $\sigma_{AB}^{\phi} = \ptr{E}{\frac{1}{p_E(\phi)}(\id_{AB} \otimes \dm{\phi}_E)\sigma_{ABE}}$. We sketch the remainder of this proof as follows: We show, that the final state of Alice and Bob is Eve-non steerable in the sense of Lemma \ref{lem.steering} by making use of the bounded distance of the protocols $\mathcal{P}_1$ and $\mathcal{P}_2$. Furthermore, Lemma \ref{lem:steering:prob} implies that it suffice to consider measurements of Eve of having probability greater than or equal to $1/16$. Therefore Lemma \ref{lem.steering} and Lemma \ref{lem:steering:prob} completes the proof. \newline
Because the output of Alice and Bob are $2$ qubits the purifying system  that Eve holds is without loss of generality also a two-qubit system. Hence, according to Lemma \ref{lem:steering:prob}, there exists a state $\sigma'_{ABE}$, which is unitarly related to $\sigma_{ABE}$ via an unitary $U$ on Eve's system only (which is not part of the protocol) and for which it suffice to consider measurements of Eve having probability greater than or equal to $1/16$. Furthermore observe that this local unitary of Eve can not change the success probability of the overall protocol as unitaries are CPTP. In other words, the success probabilities associated with $\sigma_{ABE}$ and $\sigma'_{ABE}$ are identical.
\newline
More formally, suppose Eve performs a projective measurement on this state $\sigma'_{ABE}$ (which stems from a purification $\ket{\psi'}_{ABE}$ of $\rho'_{AB}$ which is unitarly related to the purification $\ket{\psi}_{ABE}$ of $\rho_{AB}$ and both having the same success probability, see paragraph above) and observes outcome $\ket{\phi} \in \H_E$ having probability greater than or equal to $1/16$. Then the post-selected state of Alice, Bob, and Eve conditioned on that particular outcome $\phi$ reads as
\begin{align*}
& \frac{1}{p_E(\phi)}(\id_{AB} \otimes \dm{\phi}_E)(\sigma'_{ABE} \otimes \dm{ok}) = \frac{1}{p_E(\phi)}(\id_{AB} \otimes \dm{\phi}_E) \frac{1}{p_{\rho}}(\id_{ABE} \otimes \dm{ok}) (\mathcal{P}_1 \otimes \id_E)(\dm{\psi'}_{ABE}) \\
& =  \frac{1}{p_{\rho^\phi}} (\id_{ABE} \otimes \dm{ok}) (\mathcal{P}_1 \otimes \id_E)\underbrace{\left(\frac{\id_{AB} \otimes \dm{\phi}_E}{p'_E(\phi)} \dm{\psi'}_{ABE}\right)}_{=: \rho^{\phi}_{ABE}} = \frac{1}{p_{\rho^\phi}} (\id_{ABE} \otimes \dm{ok}) (\mathcal{P}_1 \otimes \id_E) (\rho^{\phi}_{ABE}) \\
& = \sigma^{\phi}_{ABE} \otimes \dm{ok}.
\end{align*}
More importantly, we relate the probability of the protocol succeeding for initial state $\rho'$, $p_{\rho}$, and the probability of measuring $\phi$ after the protocol, $p_E(\phi)$, to the probability of the protocol succeeding for the initial state $\rho^{\phi}_{ABE}$ (measurement of Eve commuted to the beginning of the protocol), $p_{\rho^\phi}$, and the probability of measuring $\phi$ before the protocol has started, $p'_E(\phi)$, via
\begin{align}
p_{\rho} p_E(\phi) = p_{\rho^\phi} p'_E(\phi). \label{eq:lemma:prod:1}
\end{align}
Observe that (\ref{eq:lemma:prod:1}) is equivalent to
\begin{align}
\frac{p_{\rho} p_E(\phi)}{p'_E(\phi)} = p_{\rho^\phi}  \label{eq:lemma:prod:2}
\end{align}
We note that the state $\sigma^{\phi}_{ABE}$ is in the $ok-$branch of the protocol $\mathcal{P}_1$. The next step is to apply Lemma \ref{lem.steering} which relates the distances $\| \sigma'_{ABE} - \sigma_{AB} \otimes \sigma'_E \|_1$ and $\|\sigma'_{AB} - \sigma^{\phi}_{AB}\|_1$. In particular we show that for all measurements of Eve with outcome $\ket{\phi} \in \H_E$ having a probability greater than or equal to $1/16$ we have that $\|\sigma'_{AB} - \sigma^{\phi}_{AB}\|_1 \leq 17 \varepsilon / p_{\rho}$. This then implies using Lemma \ref{lem:steering:prob} that $\| \sigma_{ABE} - \sigma_{AB} \otimes \sigma_E \|_1 \leq 34C \varepsilon / p_{\rho}$. In detail, using the triangle inequality we compute for the distance between $\sigma'_{AB}$ and $\sigma^{\phi}_{AB}$
\begin{align}
\|\sigma'_{AB} - \sigma^{\phi}_{AB}\|_1 & \leq \|\sigma'_{AB} - \sigma^{\alpha}_{AB}\|_1 + \|\sigma^{\alpha}_{AB} - \sigma^{\phi}_{AB}\|_1 = \frac{1}{p_{\rho}} \|(\mathcal{P}_1 - \mathcal{P}_2)(\rho'_{AB})\|_1 + \frac{1}{p_{\rho^\phi}} \|(\mathcal{P}_1 - \mathcal{P}_2)(\rho^\phi_{AB})\|_1 \notag \\
& \leq \left(\frac{1}{p_{\rho}} + \frac{1}{p_{\rho^\phi}} \right) \max_{\mu_{AB}} \|(\mathcal{P}_1 - \mathcal{P}_2)(\mu_{AB})\|_1. \label{eq:lemma:prod:3}
\end{align}
Now we employ (\ref{eq:lemma:prod:2}) in (\ref{eq:lemma:prod:3}) which yields
\begin{align}
\|\sigma'_{AB} - \sigma^{\phi}_{AB}\|_1 & \leq \left(\frac{1}{p_{\rho}} + \frac{p'_E(\phi)}{p_{\rho} p_E(\phi)} \right) \max_{\mu_{AB}} \|(\mathcal{P}_1 - \mathcal{P}_2)(\mu_{AB})\|_1 \leq \left(\frac{1}{p_{\rho}} + \frac{1}{p_{\rho} p_E(\phi)} \right) \max_{\mu_{AB}} \|(\mathcal{P}_1 - \mathcal{P}_2)(\mu_{AB})\|_1 \notag \\
& = \frac{1}{p_{\rho}} \left(1 + \frac{1}{p_E(\phi)} \right) \max_{\mu_{AB}} \|(\mathcal{P}_1 - \mathcal{P}_2)(\mu_{AB})\|_1 \leq  \frac{1}{p_{\rho}} \left(1 + 16 \right) \max_{\mu_{AB}} \|(\mathcal{P}_1 - \mathcal{P}_2)(\mu_{AB})\|_1 \notag \\
&\leq \frac{17}{p_{\rho}} \varepsilon \label{eq:lemma:prod:4}
\end{align}
because $p_E(\phi) \geq 1/16$ and $\max_{\mu_{AB}} \|(\mathcal{P}_1 - \mathcal{P}_2)(\mu_{AB})\|_1$ is bounded by $\varepsilon$ by assumption. Hence we apply Lemma \ref{lem.steering} to $\sigma'_{ABE}$ with $\epsilon = \frac{17}{p_{\rho}} \varepsilon$ which implies for the distance between $\sigma'_{ABE}$ and $\sigma_{AB} \otimes \sigma'_{E}$ that
\begin{align} \label{eq.thm.localsteer.4}
\|\sigma'_{ABE} - \sigma_{AB} \otimes \sigma'_{E} \|_1 \leq \frac{34 \cdot 4^8}{p_{\rho}} \varepsilon
\end{align}
where the factor $4^8$ is the constant $C$ of Lemma \ref{lem.steering} depending on the dimensions of the systems of Alice/Bob and Eve, for which we have $n=m=2$. Furthermore, this implies via Lemma \ref{lem:steering:prob} that
\begin{align}
\|\sigma_{ABE} - \sigma_{AB} \otimes \sigma_{E} \|_1 \leq \frac{34 \cdot 4^8}{p_{\rho}} \varepsilon \label{eq.thm.localsteer.5}
\end{align}
because $\sigma_{ABE}$ and $\sigma_{AB} \otimes \sigma_{E}$ are unitarly related to $\sigma'_{ABE}$ and $\sigma_{AB} \otimes \sigma'_{E}$ via the unitary $U$ on Eve's system. \newline
Finally, employing (\ref{eq.thm.localsteer.5}) in (\ref{eq.thm.localsteer}) yields
\begin{align*}
\|(\mathcal{P}_1 \otimes \id_E)(\dm{\psi}) &- (\mathcal{P}_2 \otimes \id_E)(\dm{\psi})\|_1 = p_{\rho} \|\sigma_{ABE} - \sigma^{\alpha}_{AB} \otimes \sigma_E \|_1 \\
& \leq p_{\rho} (\|\sigma_{ABE} - \sigma_{AB} \otimes \sigma_E \|_1 + \|\sigma_{AB} \otimes \sigma_E - \sigma^{\alpha}_{AB} \otimes \sigma_E \|_1) \\
& \leq 34 \cdot 4^8 \varepsilon + \varepsilon = (34 \cdot 4^8 +1) \varepsilon.
\end{align*}
\end{proof}

\subsection{Proof of Lemma \ref{lem:forpost}}\label{app:sec:post}

Now we turn to the proof of Lemma \ref{lem:forpost} of the main text. For that purpose we remind the reader that the final state after the distillation protocol including the system of L is pure. Thus, the following Lemma will turn out to be very useful.

\begin{lemma}\label{lem.localstates}
Let $\rho_{AB}$ and $\varphi_{AB} = \ketbra{\varphi}{\varphi}{A} \otimes \mu_{B}$ be two mixed states. Furthermore, assume that $\rho_A = \ptr{B}{\rho_{AB}} $ satisfies $\|\rho_A - \ketbra{\varphi}{\varphi}{A} \|_1 \leq \varepsilon$ and $\rho_B = \ptr{A}{\rho_{AB}}  = \mu_{B}$. Then $\|\rho_{AB} - \varphi_{AB} \|_{1} \leq 4 \sqrt{\varepsilon}$.
\end{lemma}
\begin{proof}
By assumption we have $\| \rho_{A} - \ketbra{\varphi}{\varphi}{A} \|_1 \leq \varepsilon$. Moreover, let $\ket{\psi}_{ABR}$ be a purification of $\rho_{AB}$. According to Lemma A.2.7 in \cite{bib.renner.diss} there exists a purification $\ket{\varphi}_{A} \otimes \ket{\xi}_{BR}$ of $\varphi_{AB}$ such that $\|\ket{\psi}_{ABR} - \ket{\varphi}_{A} \otimes \ket{\xi}_{BR}\|_{\text{vec}} \leq \sqrt{\| \rho_{A} - \ketbra{\varphi}{\varphi}{A} \|_1 } = \sqrt{\varepsilon}$ where $\|\ket{\psi}{}\|_{\text{vec}} = \sqrt{\langle \psi | \psi \rangle}$ and $_{ABR} \langle \psi | \varphi \rangle_{A} | \xi \rangle_{BR}$ is real and non-negative. Moreover, Lemma A.2.3 of \cite{bib.renner.diss} gives
\begin{align*}
\| \ketbra{\psi}{\psi}{ABR} - \ketbra{\varphi}{\varphi}{A} \otimes \ketbra{\xi}{\xi}{BR} \|_1 \leq 2 \|\ket{\psi}_{ABR} - \ket{\varphi}_{A} \otimes \ket{\xi}_{BR}\|_{\text{vec}} \leq 2 \sqrt{\epsilon}.
\end{align*}
We define $\xi_{B} = \ptr{R}{\ketbra{\xi}{\xi}{BR}}$. As the $1$-norm does not increase under the partial trace we have
\begin{align*}
\|\rho_{B} - \xi_{B}\|_1 \leq \|\rho_{AB} - \ketbra{\varphi}{\varphi}{A} \otimes \xi_{B}\|_1 \leq \|\ketbra{\psi}{\psi}{ABR} - \ketbra{\varphi}{\varphi}{A} \otimes \ketbra{\xi}{\xi}{BR}\|_1 \leq 2 \sqrt{\epsilon}
\end{align*}
by construction. Moreover, the assumption $\rho_{B} = \mu_{B}$ implies $\|\mu_{B} - \xi_{B}\|_1 = \|\rho_{B} - \xi_{B}\|_1 \leq 2 \sqrt{\varepsilon}$. This gives us $\|\ketbra{\varphi}{\varphi}{A} \otimes \mu_{B} - \ketbra{\varphi}{\varphi}{A} \otimes \xi_{B} \|_1 = \|\mu_{B} - \xi_{B}\|_1 \leq 2 \sqrt{\varepsilon}$. If we combine these results we obtain
\begin{align*}
\|\rho_{AB} - \varphi_{AB} \|_1 = \|\rho_{AB} - \ketbra{\varphi}{\varphi}{A} \otimes \mu_{B} \|_1 \leq \|\rho_{AB} - \ketbra{\varphi}{\varphi}{A} \otimes \xi_{B} \|_1 + \|\ketbra{\varphi}{\varphi}{A} \otimes \xi_{B} - \ketbra{\varphi}{\varphi}{A} \otimes \mu_{B} \|_1 \leq 4 \sqrt{\varepsilon}
\end{align*}
which proves the claim.
\end{proof}
Lemma \ref{lem.localstates} enables us to prove Lemma \ref{lem:forpost} of the main text.
\begin{lemmarep}[Lemma \ref{lem:forpost} of the main text]
Let $\E$ be the real protocol which guarantees to converge towards a unique and attracting fixed point depending on the noise parameter only. Let $\F$ be the ideal protocol as defined in the main text. Furthermore let $\rho$ be a mixed state (consisting of $n$ systems) shared by Alice and Bob. If the extension of $\E$ and $\F$ to the system of L satisfies $\|\E_{\text{L}}(\rho) - \F_{\text{L}}(\rho)\|_1 \leq \varepsilon(n)$, then
\begin{align*}
\|(\E \otimes id_{E'})(\ketbra{\psi}{\psi}{ABE'}) &- (\F \otimes id_{E'})(\ketbra{\psi}{\psi}{ABE'})\|_1 \leq 4 \sqrt{\varepsilon(n)}
\end{align*}
for all purifications $\ket{\psi}_{ABE'}$ of $\rho$.
\end{lemmarep}
\begin{proof}
As mentioned in the main text, we introduce a two-level flag system held by Alice which indicates whether they aborted the protocol or not. So we observe
\begin{align*}
\E_{\text{L}}(\rho) &= p_\rho \sigma_{ABEL} \otimes \ketbra{ok}{ok}{} + (1-p_\rho) \sigma^{\perp}_{ABEL} \otimes \ketbra{fail}{fail}{}, \\
\F_{\text{L}}(\rho) &= p_\rho \dm{\psi^{f}}_{ABEL} \otimes \ketbra{ok}{ok}{} + (1-p_\rho) \sigma^{\perp}_{ABEL} \otimes \ketbra{fail}{fail}{},
\end{align*}
where $E$ denotes the system of leaked noise transcripts to Eve. By assumption we have $\|\E_{\text{L}}(\rho) - \F_{\text{L}}(\rho)\|_1 \leq \varepsilon(n)$. This is equivalent to $p_{\rho} \|\sigma_{ABEL} - \ketbra{\psi_{f}}{\psi_{f}}{ABEL}\|_1 \leq \varepsilon(n)$ since $\E_{\text{L}}(\rho)$ and $\F_{\text{L}}(\rho)$ are equal on the fail branch. This we can rewrite to $\|\sigma_{ABEL} - \ketbra{\psi_{f}}{\psi_{f}}{ABEL}\|_1 \leq \varepsilon(n) / p_{\rho}$.

Moreover, applying the real and ideal protocol to the purification $\ket{\psi}_{ABE'}$ results in
\begin{align*}
(\E \otimes id_{E'})(\ketbra{\psi}{\psi}{ABE'}) &= p_{\rho} \sigma_{ABEE'} \otimes \ketbra{ok}{ok}{} + (1-p_{\rho}) \sigma^{\perp}_{ABEE'} \otimes \ketbra{fail}{fail}{}, \\
(\F \otimes id_{E'})(\ketbra{\psi}{\psi}{ABE'}) &= p_{\rho} \sigma^{f}_{ABE} \otimes \rho_{E'} \otimes \ketbra{ok}{ok}{} + (1-p_{\rho}) \sigma^{\perp}_{ABEE'} \otimes \ketbra{fail}{fail}{}.
\end{align*}
Again, both expression are equal in the fail branch, thus the $1$-norm simplifies to
\begin{align}\label{equ.global.1}
\|(\E \otimes id_{E'})(\ketbra{\psi}{\psi}{ABE'}) - (\F \otimes id_{E'})(\ketbra{\psi}{\psi}{ABE'})\|_1 = p_\rho \|\sigma_{ABEE'} -  \sigma^{f}_{ABE} \otimes \rho_{E'} \|_1.
\end{align}
Hence it is sufficient to show $p_\rho \|\sigma_{ABEE'} - \sigma^{f}_{ABE} \otimes \rho_{E'}\|_1 \leq 4 \sqrt{\varepsilon(n)}$. We observe that by introducing the system $L$ held by L that
\begin{align}\label{inequ.global.extension}
p_\rho \|\sigma_{ABEE'} - \sigma^{f}_{ABE} \otimes \rho_{E'}\|_1 \leq p_\rho \|\sigma_{ABELE'} - \dm{\psi^\alpha}_{ABEL} \otimes \rho_{E'}\|_1.
\end{align}
One easily verifies $\ptr{E'}{\sigma_{ABELE'}} = \sigma_{ABEL}$ and $\ptr{ABEL}{\sigma_{ABELE'}} = \rho_{E'}$ because the system $E'$ is not changed by the protocol $\E$. Moreover, by assumption we have $\|\sigma_{ABEL} - \ketbra{\psi_{f}}{\psi_{f}}{ABEL}\|_1 \leq \varepsilon(n) / p_{\rho}$. Thus we apply Lemma \ref{lem.localstates} to $\rho_{A'B'} := \sigma_{ABELE'}$ and $\varphi_{A'B'} = \ketbra{\psi_f}{\psi_f}{ABEL} \otimes \rho_{E'}$ where $A' := ABEL$ and $B' := E'$ which implies
\begin{align}\label{inequ.global.2}
\|\sigma_{ABELE'} - \ketbra{\psi_{f}}{\psi_{f}}{ABEL} \otimes \rho_{E'}\|_1 \leq 4 \sqrt{\varepsilon(n) / p_{\rho}}.
\end{align}
Employing (\ref{inequ.global.extension}) and (\ref{inequ.global.2}) in (\ref{equ.global.1}) yields
\begin{align*}
\|(\E \otimes id_{E'})(\ketbra{\psi}{\psi}{ABE'}) - (\F \otimes id_{E'})(\ketbra{\psi}{\psi}{ABE'})\|_1 \leq p_\rho 4 \sqrt{\varepsilon(n) / p_{\rho}} = 4 \sqrt{p_{\rho} \varepsilon(n)} \leq 4 \sqrt{\varepsilon(n)}
\end{align*}
which completes the proof.
\end{proof}

\section{Confidentiality of entanglement distillation protocols whenever the noise transcripts leak}\label{app:sec:leak}

In this section we show how the confidentiality guarantees regarding an entanglement distillation protocol can be extended to the case whenever the noise transcripts leak to Eve. \newline
We remind the reader that it is not necessary to leak the noise transcripts to Eve after every single distillation round. It is sufficient to copy all noise transcripts at the very end to Eve's register, as L is not accessible and Eve is not part of the protocol being executed by Alice and Bob.
\begin{theorep}[Theorem \ref{lem.leaklesstoleak} in main text]
Let $\E$ be the real protocol and $\F$ be the ideal protocol. Furthermore, let $\E^l$ be the real and $\F^l$ be the ideal protocol when the noise transcripts leak to Eve. Then
\begin{align*}
\|(\E \otimes \id_E) (\dm{\psi}_{ABE}) - (\F \otimes \id_E) (\dm{\psi}_{ABE})\|_1 \leq \varepsilon(n)
\end{align*}
implies that
\begin{align}\label{eq.supp.leakreduction}
\|(\E^l \otimes \id_E) (\dm{\psi}_{ABE}) - (\F^l \otimes \id_E) (\dm{\psi}_{ABE})\|_1 \leq 2 \sqrt{\varepsilon(n)}.
\end{align}
for all purifications $\ket{\psi}_{ABE}$ of initial state $\rho_{AB}$ consisting of $n$ systems
\end{theorep}
\begin{proof}
We observe that
\begin{align*}
(\E \otimes \id_E)(\dm{\psi}) &= p_{\rho} \sigma_{ABE} \otimes \dm{ok} + (1-p_\rho) \sigma^{\perp}_{AB} \otimes \sigma_E \otimes \dm{fail}, \\
(\F \otimes \id_E)(\dm{\psi}) &= p_{\rho} \sigma^{\alpha}_{AB} \otimes \sigma_E \otimes \dm{ok} + (1-p_\rho) \sigma^{\perp}_{AB} \otimes \sigma_E \otimes \dm{fail}.
\end{align*}
So by assumption we have
\begin{align*}
\| (\E \otimes \id_E)(\dm{\psi}) - (\F \otimes \id_E)(\dm{\psi}) \|_1 = p_{\rho} \| \sigma_{ABE} - \sigma^{\alpha}_{AB} \otimes \sigma_E\|_1 \leq \varepsilon(n),
\end{align*}
i.e. $\| \sigma_{ABE} - \sigma^{\alpha}_{AB} \otimes \sigma_E\|_1 \leq \varepsilon(n) / p_{\rho}$. \newline
As outlined in the main text we model L in terms of purifications. Because purifications are unitarly equivalent we choose a particular purification of $\sigma^{\alpha}_{AB} \otimes \sigma_E$. Thus we fix $\ket{\psi_{\F}}_{ABL_1L_2E} = \ket{\psi'}_{ABL_1} \otimes \ket{\psi''}_{L_2E}$ where $\ket{\psi'}_{ABL_1} = \sum_{i,j} \omega_{ij}(\alpha)\ket{B_{ij}}_{AB} \ket{ij}_{L_1}$. The purifying systems $L_1$ and $L_2$ we attribute to the Lab Demon. \newline
Moreover, according to Lemma A.2.7 in \cite{bib.renner.diss} there exists a purification $\ket{\psi_{\E}}$ of $\sigma_{ABE}$ such that $\|\ket{\psi_{\F}}_{ABL_1L_2E} - \ket{\psi_{\E}}_{ABL_1L_2E} \|_{\text{vec}} \leq \sqrt{\varepsilon(n) / p_{\rho}}$ where $\|\ket{\psi}{}\|_{\text{vec}} = \sqrt{\langle \psi | \psi \rangle}$ and $_{ABL_1L_2E} \langle \psi_{\F} | \psi_{\E} \rangle_{ABL_1L_2E}$ is real and non-negative. Furthermore, Lemma A.2.3 of \cite{bib.renner.diss} gives
\begin{align}\label{eq.leakreduction.0}
\| \dm{\psi_{\E}}_{ABL_1L_2E} - \dm{\psi_{\F}}_{ABL_1L_2E} \|_1 \leq 2 \|\ket{\psi_{\E}}_{ABL_1L_2E} - \ket{\psi_{\F}}_{ABL_1L_2E} \|_{\text{vec}} \leq 2 \sqrt{\varepsilon(n) / p_{\rho}}.
\end{align}
When the noise transcripts leak to Eve, L effectively copies the noise transcripts $\ket{ij}_{L_1}$ to Eve, resulting in the pure state $\ket{\phi}_{ABL_1L_2EE'} = \left(\sum_{i,j} \ket{B_{ij}}_{AB} \ket{ij}_{L_1} \ket{ij}_{E'}\right) \otimes \ket{\psi}_{L_2E}.$ Hence we can model the leakage of the noise transcripts to Eve by a unitary $U_{M}$ such that $U_{M} \ket{\psi_{\F}}_{ABL_1L_2E} \ket{0}_{E'} = \ket{\phi}_{ABL_1L_2EE'}$. \newline
For the protocol when the noise transcripts leak to Eve we have
\begin{align*}
(\E^l \otimes \id_E)(\dm{\psi}) &= p_{\rho} \sigma'_{ABE} \otimes \dm{ok} + (1-p_\rho) \sigma^{\perp}_{AB} \otimes \sigma_E \otimes \dm{fail} \\
& = p_{\rho} \ptr{L_1 L_2}{U_{M} \dm{\psi_{\E}} U^\dagger_{M}} \otimes \dm{ok} + (1-p_\rho) \sigma^{\perp}_{AB} \otimes \sigma_E \otimes \dm{fail} \\
(\F^l \otimes \id_E)(\dm{\psi}) &= p_{\rho} \sigma'^{\alpha}_{ABE}  \otimes \dm{ok} + (1-p_\rho) \sigma^{\perp}_{AB} \otimes \sigma_E \otimes \dm{fail} \\
& = p_{\rho} \ptr{L_1 L_2}{U_{M} \dm{\psi_{\F}} U^\dagger_{M}} \otimes \dm{ok} + (1-p_\rho) \sigma^{\perp}_{AB} \otimes \sigma_E \otimes \dm{fail}.
\end{align*}
Because the real and the ideal protocol are equal in the fail-branch we obtain by using (\ref{eq.leakreduction.0})
\begin{align*}
\|(\E^l \otimes \id_E) (\dm{\psi}_{ABE}) &- (\F^l \otimes \id_E) (\dm{\psi}_{ABE})\|_1 \\
& = p_{\rho} \left\| \sigma'_{ABE} \otimes \dm{ok} - \sigma'^\alpha_{ABE} \otimes \dm{ok} \right\|_1 \\
& = p_{\rho} \left\|\ptr{L_1 L_2}{U_{M} \dm{\psi_{\E}} U^\dagger_{M}} - \ptr{L_1 L_2}{U_{M} \dm{\psi_{\F}} U^\dagger_{M}} \right\|_1 \\
& \leq p_{\rho} \left\|U_{M} \dm{\psi_{\E}} U^\dagger_{M} - U_{M} \dm{\psi_{\F}} U^\dagger_{M} \right\|_1 \\
& = p_{\rho} \left\| \dm{\psi_{\E}} - \dm{\psi_{\F}} \right\|_1 \leq 2 \sqrt{\varepsilon(n) p_{\rho}} \leq 2 \sqrt{\varepsilon(n)}
\end{align*}
which proves (\ref{eq.supp.leakreduction}).
\end{proof}

Thus the confidentiality of a protocol where the noise transcripts leak to Eve is bounded by the confidentiality of the same protocol when they do not.

\section{Quantum one-time padding after the real protocol}
\label{sec.supp.decouple}

In this section we show that a final secret twirl applied to the pair of Alice and Bob decouples Eve completely from the remaining state. Keep in mind that for this Alice and Bob require two classical bits unknown to Eve.

Recall that the state of Alice, Bob, Eve, and L after $n$ distillation rounds is pure and of the form $\ket{\psi}{} = \sum_{i,j,k,l} P_{ijkl} \ket{B_{ij}}_{AB} \ket{\eta_{kl}}_{L} \ket{\eta_{ijkl}}_{E}$. Tracing over L yields the mixed state
\begin{align}\label{eqn.decouple.1}
\rho_{ABE} = \sum\limits_{i_1,i_2,j_1,j_2}\sum\limits_{k,l} P_{i_1 j_1 k l} P^{*}_{i_2 j_2 k l}\ketbra{B_{i_1 j_1}}{B_{i_2 j_2}}{} \otimes \ketbra{\eta_{i_1 j_1 k l}}{\eta_{i_2 j_2 k l}}{}.
\end{align}
Suppose Alice and Bob apply a secret twirl $\T$ to (\ref{eqn.decouple.1}), i.e. they apply stochastically the family of operators $\lbrace id, K_1, K_2, K_1 K_2 \rbrace$ where $K_1 = \px \otimes \px$ and $K_2 = \pz \otimes \pz$. These are two stabilizers of the Bell state, i.e.,
\begin{align*}
K^{r_1}_1 \ket{B_{i_1 j_1}}{} &= (-1)^{i_1 r_1} \ket{B_{i_1 j_1}}{}, \\
K^{r_2}_2 \ket{B_{i_1 j_1}}{} &= (-1)^{j_1 r_2} \ket{B_{i_1 j_1}}{}.
\end{align*}
Hence, applying the secret twirl $\T$ to (\ref{eqn.decouple.1}) gives
\begin{align*}
\T \rho_{ABE} &= \sum\limits_{\stackrel{r_1,r_2}{i_1,i_2,j_1,j_2,k,l}} \frac{1}{4} P_{i_1 j_1 k l} P^{*}_{i_2 j_2 k l} K^{r_1}_1 K^{r_2}_2 \ketbra{B_{i_1 j_1}}{B_{i_2 j_2}}{} K^{r_1}_1 K^{r_2}_2 \otimes \ketbra{\eta_{i_1 j_1 k l}}{\eta_{i_2 j_2 k l}}{} \\
& = \sum\limits_{\stackrel{r_1,r_2}{i_1,i_2,j_1,j_2,k,l}}(-1)^{i_1 r_1} (-1)^{j_1 r_2} (-1)^{i_2 r_1} (-1)^{j_2 r_2} \frac{1}{4} P_{i_1 j_1 k l} P^{*}_{i_2 j_2 k l}  \ketbra{B_{i_1 j_1}}{B_{i_2 j_2}}{} \otimes \ketbra{\eta_{i_1 j_1 k l}}{\eta_{i_2 j_2 k l}}{} \\
& = \sum\limits_{i_1,i_2,j_1,j_2} \ketbra{B_{i_1 j_1}}{B_{i_2 j_2}}{} \otimes \frac{1}{4} \sum\limits_{k,l} P_{i_1 j_1 k l} P^{*}_{i_2 j_2 k l} \ketbra{\eta_{i_1 j_1 k l}}{\eta_{i_2 j_2 k l}}{}   \sum\limits_{r_1,r_2} (-1)^{(i_1 + i_2) r_1} (-1)^{(j_1 + j_2) r_2} \\
& = \sum\limits_{i_1,j_1} \ketbra{B_{i_1 j_1}}{B_{i_1 j_1}}{} \otimes \sum\limits_{k,l} |P_{i_1 j_1 k l}|^2 \ketbra{\eta_{i_1 j_1 k l}}{\eta_{i_1 j_1 k l}}{}.
\end{align*}
Note that in the resulting state $\sum\limits_{i_1,j_1} \ketbra{B_{i_1 j_1}}{B_{i_1 j_1}}{} \otimes \sum\limits_{k,l} |P_{i_1 j_1 k l}|^2 \ketbra{\eta_{i_1 j_1 k l}}{\eta_{i_1 j_1 k l}}{}$ Eve decouples, i.e. Alice/Bob and Eve have a separable state. The obtained resource state can be used to establish a confidential quantum channel by means of quantum teleportation.

\section{Robustness of recurrence-type entanglement distillation protocol}\label{app:sec:robust}

To complete the security characterization of entanglement distillation protocols we also consider the robustness of an entanglement distillation protocol. To define this term precisely we first need the definition of a honest eavesdropper.
\begin{defi}
We call an eavesdropper honest, if the states sent by the eavesdropper are of the form $\ket{B_{00}}{}^{\otimes 2^n}$.
\end{defi}
It is obvious that a honest eavesdropper is not entangled with the ensemble delivered to Alice and Bob via the noisy quantum channel. Moreover we formally define the robustness of a protocol by:
\begin{defi}[Robustness of a protocol]
We call a protocol $\E^\alpha$ $\varepsilon_R$-robust, if for a honest eavesdropper the probability of aborting the protocol is at most $\varepsilon_R$.
\end{defi}
Now we show that we can tune the robustness of a recurrence-type entanglement distillation protocol to be exponentially small in terms of necessary number of input pairs.
\begin{theo}
Let $M \in \N$ such that Alice and Bob achieve $\varepsilon$-confidentiality by succeeding $M$ rounds of a recurrence-type entanglement distillation protocol. Furthermore assume that Alice and Bob receive $n$ pairs from a honest eavesdropper over the quantum channel $\Phi^{\otimes n}$ (where $\Phi(\rho) = \beta \rho + (1-\beta)/4 \left(\sum_{i,j} \sigma_{i,j} \rho \sigma_{i,j} \right)$) such that, after the parameter estimation step of the proposed protocol, $k - \sqrt{k}$ pairs (where $k - \sqrt{k} = c 2^M$ and $c = \xi 2^{M+2}$) are left for entanglement distillation. Then, the robustness $\varepsilon_{R}$ of the protocol is bounded by
\begin{align*}
\varepsilon_{R} \leq \exp\left(- (3\beta - 4 F_{min}(\alpha) - 1)^2 \sqrt{k}/128 \right) + M \exp\left( - \xi \right).
\end{align*}
\end{theo}
\begin{proof}
The basic idea of the proof is to request sufficiently many pairs from Eve such that the probabilities of abort during the protocol to be exponentially small while still having enough pairs left to achieve $M$ rounds of a recurrence-type entanglement distillation protocol. We divide the proof into two parts:
\begin{itemize}
\item Part 1: We prove that the probability of aborting the recurrence-type entanglement distillation protocol due to parameter estimation is exponentially small.
\item Part 2: We prove the same holds true for aborting the protocol during entanglement distillation.
\end{itemize}
\underline{Part 1:} Suppose Eve sends the state $\ket{B_{00}}{}^{\otimes n}$ through the noisy quantum channel $\Phi^{\otimes n}$ to Alice and Bob. Applying $\Phi$ to $\ketbra{B_{00}}{B_{00}}{}$ yields
\begin{align}\label{eqn.robust.qchannel}
\rho_{AB} = \Phi\left(\ketbra{B_{00}}{B_{00}}{} \right) = (3\beta + 1)/4 \ketbra{B_{00}}{B_{00}}{} + (1-\beta)/4 \left(\ketbra{B_{10}}{B_{10}}{} + \ketbra{B_{01}}{B_{01}}{} + \ketbra{B_{11}}{B_{11}}{} \right).
\end{align}
Thus the state Alice and Bob receive is $\rho_{AB}^{\otimes n}.$ According to the preceding protocols proposed in the main text, Alice and Bob apply a symmetrization to $\rho_{AB}^{\otimes n}$, and, depending on the noise level of the apparatus, they might have to trace out $n-k$ pairs or not. For the subsequent analysis we assume that this tracing out step is necessary, i.e. the de-Finetti-based reduction needs to be applied. Hence, Alice and Bob continue by applying a twirl to each remaining pair. Since $\rho_{AB}^{\otimes k}$ is invariant under permutations and $\rho_{AB}$ is Bell-diagonal, the remaining state after twirling is equal to $\rho_{AB}^{\otimes k}$.

Next, they apply to $\sqrt{k}$ of the remaining $k$ pairs the parameter estimation for estimating the fidelity of each pair. Necessary for convergence of all recurrence-type entanglement distillation protocols is that the fidelity $F$ of $\rho_{AB}$ with $\ket{B_{00}}{}$ satisfies $F > F_{min}(\alpha)$. Hence this step is crucial in order to guarantee successful distillation.

For that purpose, we measure $\lfloor \sqrt{k} \rfloor$ of $k$ pairs by applying two-qubit measurements. To be more precise, we apply a $\px \otimes \px$ to the first and $\pz \otimes \pz$ measurement to the second pair. We refer to this measurements by $M_1$ and $M_2$ respectively. We observe that the state $\ket{B_{00}}{}$ is a common eigenstate of $M_1$ and $M_2$ with eigenvalue $1$. We define to each pair of pairs a random variable $X_i$ for $i \in \lbrace 1,.., \lfloor \sqrt{k} \rfloor /2 \rbrace$ with $X_i = 1$ whenever both measurements $M_1$ and $M_2$ yield outcome $1$ and $X_i = 0$ else.

Furthermore we assume for the expected value $\mathbb{E}(X)$ of the fidelity with $\ket{B_{00}}$ that $\mathbb{E}(X) = F_{min}(\alpha) + \delta$, where $\delta > 0$ will be fixed below. The protocol will be aborted if the estimate is below $F_{min}(\alpha) + \delta$.

From (\ref{eqn.robust.qchannel}) we observe that, whenever $(3\beta + 1)/4 \leq F_{min}(\alpha)$, the entanglement distillation protocol will not distill any entanglement. This implies for the quantum channel $\Phi$ that, if $\beta \leq (4F_{min}(\alpha) - 1)/3$ the parameter estimation step will abort, independent of the input provided by Eve. Thus we assume for the subsequent analysis that $\beta > (4F_{min}(\alpha) - 1)/3$.

Moreover we define $\eta = \delta/2$. Hence we get by the Hoeffdings inequality \cite{Hoeffding} for the probability of an error larger than $\eta$ in our measured estimate $\overline{X}$ for the fidelity the following expression:
\begin{align*}
\mathbb{P}(|\mathbb{E}(X) - \overline{X}| \geq \eta) \leq \exp\left( - \eta^2 \sqrt{k}/2 \right) =: p_{\text{pe-abort}}.
\end{align*}
Thus the probability of aborting the protocol due to an error in the parameter estimation is exponentially small in number of necessary input pairs. In order to fix $\delta$ we recognize that Alice and Bob abort the protocol whenever $(3\beta + 1)/4 < F_{min}(\alpha) + \delta$. This is equivalent to $\delta > (3\beta - 4 F_{min}(\alpha) - 1)/4$. Inserting the definition of $\eta$ yields $\eta > (3\beta - 4 F_{min}(\alpha) - 1)/8$ and thus $p_{\text{pe-abort}} < \exp\left(- (3\beta - 4 F_{min}(\alpha) - 1)^2 \sqrt{k}/128 \right)$. \newline

\underline{Part 2:} What remains to be shown is that the probability of aborting the protocol in the distillation phase is also exponentially small in the number of input pairs. For that purpose, we assume that the noise level $\alpha$ of the apparatus is such that distillation is feasible. In the following we show that we can force the probability of abort due to entanglement distillation to be exponentially small in terms of requested input pairs.

We assume that Alice and Bob are left with $c 2^M$ pairs after parameter estimation. Recall that the Chernoff inequality for a sequence of independent Bernoulli random variables $X_1, ... , X_n$ where $\mathbb{P}\left(X_i = 1 \right) = p$ and $d \in [0,1]$ reads as
\begin{align*}
\mathbb{P}\left(\sum_i X_i \leq (1-d) p n \right) \leq \exp\left( - \frac{d^2}{2} p n \right).
\end{align*}
Moreover, we observe that a basic distillation step can be modelled by a Bernoulli random variable $X_i$ where $\mathbb{P}\left(X_i = 1 \right) = p$ is the probability of succeeding (measurement outcomes coincide).

Suppose we perform $m$ rounds of entanglement distillation. Let $N_m$ denote the number of input pairs to the $m$-th round and let $d \in [0,1]$.
\begin{figure}[h]
\scalebox{0.9}{
\includegraphics{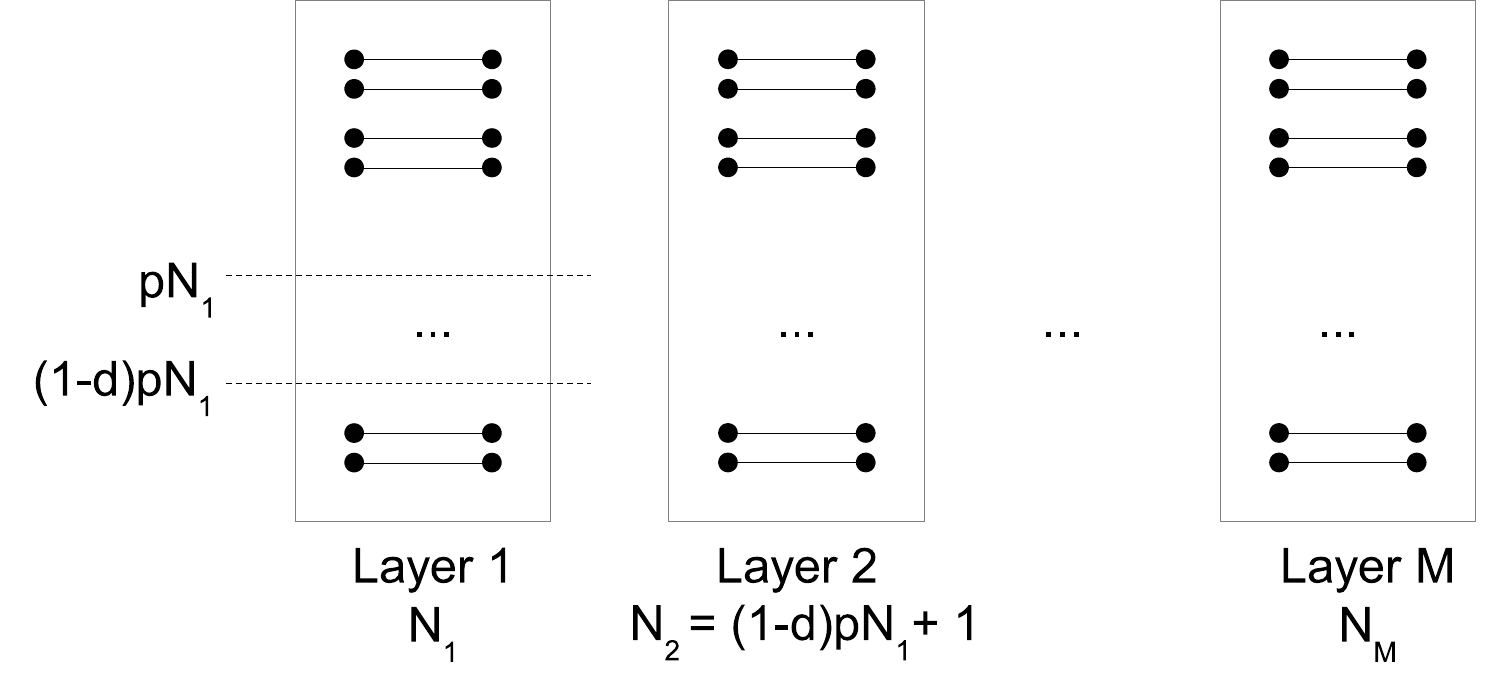}
}
\caption{M rounds of entanglement distillation}\label{app:fig:robustness}
\end{figure}
Then the Chernoff inequality implies that the probability that less than $(1-d)p N_m$ basic distillation steps at round $m$ have succeeded is bounded by $\exp\left( - \frac{d^2}{2} p N_m \right)$, i.e.
\begin{align}\label{inequ.robustness.1}
p_{\text{abort},m} = \mathbb{P}\left(\sum_i X_i \leq (1-d) p N_m \right) \leq \exp\left( - \frac{d^2}{2} p N_m \right).
\end{align}
But this also implies that, with probability $1-p_{\text{abort},m}$, at least $(1-d) p N_m + 1$ basic distillation steps have succeeded at round $m$. Thus we may safely assume that $N_{m+1} = (1-d) p N_m + 1$. The situation is summarized in Fig. \ref{app:fig:robustness}. Furthermore we have $N_{1} = c 2^M$. Eliminating the recurrence relation yields $N_{m+1} = (1-d)^{m} p^{m} c 2^M + \sum_{i=0}^{m-1}(1-d)^i p^i$. This implies for (\ref{inequ.robustness.1})
\begin{align*}
p_{\text{abort},m} \leq \exp\left( - \frac{d^2}{2} p \left((1-d)^{m-1} p^{m-1} c 2^M + \underbrace{\sum_{i=0}^{m-2}(1-d)^i p^i}_{> 0} \right) \right) \leq \exp\left( - \frac{d^2}{2} (1-d)^{m-1} p^{m} c 2^M \right).
\end{align*}
Furthermore, we compute the probability of aborting the protocol at distillation round $m$ (assuming that the previous rounds $1,..,m-1$ succeeded) by
\begin{align}\label{inequ.robustness.2}
p_{\text{abort at round $m$}} = p_{\text{abort},m} \underbrace{\prod\limits^{m-1}_{k=1} p_{\text{succeed},k}}_{\leq 1} \leq p_{\text{abort},m} \leq \exp\left( - \frac{d^2}{2} (1-d)^{m-1} p^{m} c 2^M \right).
\end{align}
The events of aborting the distillation protocol at two different rounds $i$ and $j$ are disjoint. Thus we have for the probability of aborting in any of $m$ rounds $p_{\text{abort in any of $m$ rounds}} = \sum_{k=1}^m p_{\text{abort at round $k$}}$. A simple consequence thereof is
\begin{align}\label{inequ.robustness.3}
p_{\text{abort in any of $M$ rounds}} &= \sum\limits^M_{k=1} p_{\text{abort at round $k$}} \leq \sum\limits^M_{k=1} \exp\left( - \frac{d^2}{2} (1-d)^{k-1} p^{k} c 2^M \right)
\end{align}
where we have used (\ref{inequ.robustness.2}). Inserting $p=1/2$ and $d =1/2$ in (\ref{inequ.robustness.3}) yields
\begin{align} \notag
p_{\text{abort in any of $M$ rounds}} & \leq \sum\limits^M_{k=1} \exp\left( - \frac{1}{8} \frac{1}{2^{2k-1}} c 2^M \right) = \sum\limits^M_{k=1} \exp\left( - c 2^{M-2k-2} \right) \leq M \exp\left( - c 2^{M-2M-2} \right) \\ \label{inequ.robustness.4}
& = M \exp\left( - c 2^{-(M+2)} \right).
\end{align}
By assumption we have $c = 2^{M+2} \xi$ which implies for (\ref{inequ.robustness.4})
\begin{align*}
p_{\text{abort in any of $M$ rounds}} & \leq  M \exp\left( - \xi 2^{M+2} 2^{-(M+2)} \right) = M \exp\left( - \xi \right).
\end{align*}
Thus, the probability of aborting the protocol satisfies
\begin{align*}
\varepsilon_{R} & \leq p_{\text{pe-abort}} + (1-p_{\text{pe-abort}})p_{\text{abort in any of $M$ rounds}} \leq \exp\left(- (3\beta - 4 F_{min}(\alpha) - 1)^2 \sqrt{k}/128 \right) + M \exp\left( - \xi \right)
\end{align*}
which completes the proof.
\end{proof}

\section{Establishing a confidential quantum channel}\label{app:sqc}

For illustration purposes, we show how confidential quantum channels can be realized using our proposal in conjunction with standard teleportation. By our results, the joint state of Alice, Bob, and Eve after the distillation protocol is $\epsilon$ close to the output of the ideal protocol. The latter, since the register of L is not accessible to any of the parties and  thus is traced out, yields the state of the form (provided the protocol was not aborted)
\EQ{\rho_{final} =  \sum_{i,j} |\omega_{ij}(\alpha)|^2 \dm{B_{ij}}_{AB} \otimes \dm{\eta_{ij}}_{E}.\ \
}
The teleportation of any state $\rho$ from Alice to Bob will yield the state
\EQ{
\sum_{i,j} |\omega_{i,j}(\alpha)|^2 \px^j \pz^i\rho \pz^i \px^j \otimes \dm{\eta_{ij}}_{E}.
}
Thus the only information Eve can obtain is what noise operator was applied on the teleported state, and nothing more -- thus, the channel is confidential. Moreover, the probabilities for the different noise processes are not under Eve's control, but depend on the local devices.

\end{document}